%% file: article.tex
\documentclass[a4paper, 11pt]{article}
\usepackage{srcltx}
\usepackage{indentfirst}
\usepackage{graphicx}
\usepackage{overpic}
\usepackage{multirow}
\usepackage[dvipdfmx]{hyperref}
\usepackage[hang]{subfigure}
\usepackage{amsmath, amssymb, amsthm}
\usepackage{color}
\usepackage[mathscr]{euscript}

\newcommand\bbR{\mathbb{R}}
\newcommand\bbN{\mathbb{N}}
\newcommand\bxi{\boldsymbol{\xi}}
\newcommand\bx{\boldsymbol{x}}
\newcommand\bv{\boldsymbol{v}}
\newcommand\bu{\boldsymbol{u}}
\newcommand\bw{\boldsymbol{w}}
\newcommand\bn{\boldsymbol{n}}
\newcommand\bF{\boldsymbol{F}}
\newcommand\bA{\boldsymbol{A}}
\newcommand\htheta{\hat{\theta}}

\newcommand\dd{\,\mathrm{d}}
\newcommand\He{\mathit{He}}
\newcommand\Kn{\mathit{Kn}}

\numberwithin{equation}{section}

\graphicspath{{images/}}

{\theoremstyle{remark} \newtheorem{remark}{Remark}}
\newtheorem{proposition}{Proposition}

\title{Numerical Regularized Moment Method of Arbitrary Order for Boltzmann-BGK Equation}

\author{Zhenning Cai\thanks{School of Mathematical Sciences, Peking
    University, Beijing, China, email: {\tt cai\_zn1987@163.com}.}~~and Ruo
    Li\thanks{CAPT, LMAM \& School of Mathematical Sciences, Peking
      University, Beijing, China, email: {\tt rli@math.pku.edu.cn}.}}

\begin{document}
\maketitle
\input{article_abs_intro.tex}
\input{article_BGK.tex}
\input{article_num_ex.tex}
\input{article_conclusion.tex}
\input{article_appendix.tex}
\bibliographystyle{plain}
\bibliography{article}
\end{document}

%% file: article_abs_intro.tex
\begin{abstract}
  We introduce a numerical method for solving Grad's moment equations
  or regularized moment equations for arbitrary order of moments. In
  our algorithm, we do not explicitly need the moment equations.
  Instead, we directly start from the Boltzmann equation and perform
  Grad's moment method \cite{Grad} and the regularization technique
  \cite{Struchtrup2003} numerically. We define a conservative
  projection operator and propose a fast implementation which makes it
  convenient to add up two distributions and provides more efficient
  flux calculations compared with the classic method using explicit
  expressions of flux functions. For the collision term, the BGK model
  is adopted so that the production step can be done trivially based
  on the Hermite expansion. Extensive numerical examples for one- and
  two-dimensional problems are presented.  Convergence in moments can
  be validated by the numerical results for different number of
  moments.

\vspace*{4mm}
\end{abstract}

\section{Introduction}
In recent years, the simulation of rarefied fluids or microflows, which
contain significant non-equilibrium characteristics, became one of the
major directions of fluid dynamics. The Boltzmann equation, which is
considered as the basis of modern kinetic theory, is the starting
point of such simulations. Because of the high dimension of variables
and the complicated form of its collision operators, people tend to
use its discrete or simplified form instead of the Boltzmann equation
itself in numerical simulation.  Lots of work has been done to
simplify the collision operator, such as the BGK model \cite{BGK}, the
Shakhov model \cite{Shakhov}, the ES-BGK model \cite{Holway}, the
Liu's model \cite{Liu}, the Maxwell molecules model \cite{Ernst}, and
so on. Another way of simplification is to discretize the Boltzmann
equation by some expansion. In this field, the Chapman-Enskog
expansion \cite {Chapman,Enskog} and the Grad's expansion
\cite{Grad,Grad1958} achieved great success in the early exploration
of kinetic theory.

However, both methods of expansion suffer some problems which greatly
restrict their application. The Chapman-Enskog expansion shows an
unstable behavior in the case of high order expansions, such as
Burnett and super-Burnett equations \cite{Bobylev}; the Grad's moment
equations lead to unphysical subshocks when the Mach number is large
(see e.g. \cite{Struchtrup}). In order to extend their use, several
corrections are applied to these models, which include the R13 model,
proposed by H. Struchtrup and M. Torrilhon in \cite{Struchtrup2003}.
The regularization of Grad's moment equations is done by combining
Grad's technique with first-order Chapman-Enskog expansion.
\cite{Torrilhon2009} summarizes the recent work on R13 equations, in
which it is mentioned that such a regularization technique can be
extended to any system obtained by Grad's moment method. In this
paper, we introduce a method that numerically solves the Grad's
equations for arbitrary order of moments together with their
regularization. The numerical method generating the Grad's moment
equations for arbitrary number of moments has been proposed in
\cite{Torrilhon} and implemented in \cite{ETXX}. And the results for
shock tube are reported in \cite{Au, Weiss}. However, we solve these
equations ``only numerically,'' which means we do not need the
explicit forms of those equations in our algorithm. The regularized
moment equations are also considered, and to our knowledge, no results
about the generation of regularized equations for arbitrary number of
moments have been reported. Numerical methods for R13 equations are
discussed in \cite{Torrilhon2006,Gu,Torrilhon2008}, and the
computational framework for the R20 equations can be found in
\cite{Mizzi}.

Our starting point is the Boltzmann equation, and we adopt the BGK
model and a uniform rectangular mesh for discretizing the spatial
variable for simplicity.  The distribution function is expanded in an
Hermite series using the method in \cite{Grad}. The series is
truncated at a certain place and the coefficients in the expansion are
stored for each cell. Following the standard procedure, we adopt the
classic time splitting method in our scheme.  The convection term of
the Boltzmann equation is discretized by the finite volume method with
the HLL numerical fluxes. This requires algebraic operations on the
distributions between the neighbouring cells. Grad's method expands
the distribution function by the Hermite functions with center at the
local mean velocity and a scaling factor associated with the local
temperature, while these parameters are different in different places.
This makes it nontrivial to add up two distributions in different
cells. Thus it is extremely complicated to calculate the numerical
fluxes directly. As a key technique in this paper, we propose a fast
algorithm which projects a distribution function expanded in the
discrete space with one mean velocity and macroscopic temperature to
that with another mean velocity and macroscopic temperature. This
projection is conservative with respect to all the moments that are
not truncated. Moreover, we prove the projection is invertible so that
no information will be lost. With the help of this fast projection,
distributions in any two neighbouring cells can be transformed into
the same space efficiently. In order to estimate the signal velocities
in the HLL numerical fluxes, we prove that the eigenvalues of the
Jacobian matrix of the linearized flux function are actually the roots
of Hermite functions plus the mean velocity. Thus we can take the
eigenvalues with maximal absolute values as the approximated signal
velocities. Now, standard HLL numerical fluxes can be calculated
conveniently. After accumulating the contribution of the convection
term, the expansion is not of Grad's type any more, since Grad's
expansion requires all first-order coefficients and the trace of
second-order coefficients to be zero.  Again, the fast projection is
applied to correct the center and scaling factor of the expansion so
that the properties of Grad's expansion can be recovered. These
properties make it trivial to perform production of the BGK model by a
direct scaling of the moments with orders not lower than two.

In \cite {Struchtrup2003}, H. Struchtrup and M. Torrilhon used the
Chapman-Enskog expansion to deduce the regularized system of Grad's
13-moment equations --- R13 equations. The basic idea of this
regularization can be viewed as a strategy to ``guess'' the truncated
moments. Using this idea, we apply the Chapman-Enskog expansion of
the Boltzmann equation around the Grad non-equilibrium manifolds
\cite{Struchtrup} of arbitrary order. The closure of the system is
then achieved by the standard asymptotic techniques therein. This
method can be perfectly integrated into our numerical scheme without
deducing the macroscopic equations by intricate algebraic
calculations. For arbitrary order of moments, the regularization in
our algorithm introduces only first order derivative terms, which can
be numerically approximated by gradient reconstruction.

In our method, the computational cost of the fast conservative
projection is linear in terms of the number of moments. Thus it is
essentially faster to calculate the numerical fluxes than the classic
method using the flux functions in the macroscopic equations.
Actually, the macroscopic equations have never been deduced, but they
have been solved implicitly by our method. Then the framework of our
method can appear to be uniform for moment equations of any order.
This makes it very convenient to implement our algorithm. We need 
not deduce and code for the complicated flux functions at all in
the case of high order. Moreover, we need only to develop one copy of
the code for all different orders.

We carry out numerical experiments in both one- and two-dimensional
cases.  Different Knudsen numbers and different orders of moments are
examined to demonstrate the usefulness of large moment systems, and
the convergence in moments is validated numerically. Regularized
moment systems ranging from 20 moments up to 455 moments are simulated
in our one-dimensional examples. Two-dimensional examples for up to 84
moments are presented, and there is even an example demonstrating the
capacity of our method to simulate a three-dimensional non-equilibrium
process. To the best of our knowledge, it is the first time that the
method for arbitrary order regularized moment equations is numerically
implemented, and the moment method for large systems are applied to
two-dimensional problems.

The layout of this paper is as follows: in Section 2, an overview of
Boltzmann equation and the BGK model is given as the basis of our
algorithm. In section 3, the details of the algorithm which generate
numerical solution for arbitrary order Grad's moment equations are
introduced. In section 4, regularization of moment equations is
considered and then the whole algorithm is outlined. We present in
Section 5 four numerical examples including one- and two-dimensional
tests to make a comparison between results for different moment
equations, different Knudsen numbers and different meshes. At last,
some concluding remarks will be given in Section 6.


%% file: article_BGK.tex
\section{The Boltzmann equation and BGK collision model}
In the kinetic theory of gases, the flow of a dilute gas is described
by the Boltzmann equation (see e.g.
\cite{Cercignani,Degond,Struchtrup})
\begin{equation} \label{eq:Boltzmann}
\frac{\partial f}{\partial t} + 
  \bxi \cdot \nabla_{\bx} f = Q(f,f),
\end{equation}
where $f(t, \bx, \bxi)$ is the distribution function, and
$(t,\bx,\bxi) \in \bbR^+ \times \bbR^D \times \bbR^D$. $Q(f,f)$
is the collision term with a quadratic expression given by
\begin{equation}
Q(f,f) = \alpha \int_{\bbR^D} \int_{S_+^{D-1}} (f'f_*' - f f_*)
  |(\bxi - \bxi_*) \cdot \bn| \dd \bxi_* \dd \bn,
\end{equation}
where
\begin{equation} \label{eq:f}
f_* = f(\bxi_*), \quad f' = f(\bxi'), \quad f_*' = f(\bxi_*'),
\end{equation}
and $\bxi'$ and $\bxi_*'$ are velocities after collision of two
particles with original velocities $\bxi$ and $\bxi_*$ and with unit
vector $\bn \in S_+^{D-1}$ joining the centers of them.  $\alpha$ is a
constant equivalent to $N \sigma^2$, which keeps invariant when taking
Boltzmann-Grad limit $N \rightarrow \infty$, $\sigma \rightarrow 0$
(cf. \cite{Cercignani,Degond}).  Here $N$ is the number of particles
while $\sigma$ is the diameter of each particle.

However, such a collision term turns out to be too complicated
for numerical simulation, so a variety of variants are raised
to get it simplified. To some extent, the BGK operator \cite
{BGK} is the simplest one. It substitutes the collision term
$Q(f,f)$ by
\begin{equation} \label{eq:BGK}
Q_{\mathrm{BGK}}(f) = -\nu (f - f_M),
\end{equation}
where $\nu$ is the collision frequency, and $f_M$ is the local
Maxwellian defined as
\begin{equation}
f_M(t, \bx, \bxi) = \frac{\rho(t, \bx)}
  {[2\pi \theta(t, \bx)]^{D/2}}
  \exp \left(
    -\frac{|\bxi - \bu(t, \bx)|^2}{2\theta(t, \bx)}
  \right).
\end{equation}
It is related with $f$ by
\begin{equation} \label{eq:moments}
\begin{split}
\rho(t, \bx) &= \int_{\bbR^D} f(t, \bx, \bxi) \dd \bxi, \\
\rho(t, \bx) \bu(t, \bx) &=
  \int_{\bbR^D} \bxi f(t, \bx, \bxi) \dd \bxi, \\
\rho(t, \bx) |\bu(t, \bx)|^2 + D \rho(t, \bx) \theta(t, \bx)
  &= \int_{\bbR^D} |\bxi|^2 f(t, \bx, \bxi) \dd \bxi,
\end{split}
\end{equation}
where $\rho, \bu$ and $\theta$ can be viewed as macroscopic
variables density, velocity and temperature, respectively. This model
is much more easy to use in numerical methods. However, it suffers the
disadvantage of being unable to predict the correct Prandtl number,
which will be seen in the numerical tests.

\section{A numerical formation equivalent to Grad's moment method}
\label{sec:Grad}
\subsection{Discretization of the distribution function}
In order to solve the kinetic equations numerically, we first
expands the distribution function into Hermite functions as
in \cite{Grad}:
\begin{equation} \label{eq:expansion}
f(\bxi) = \sum_{\alpha \in {\bbN^D}} f_{\alpha}
  \mathcal{H}_{\theta,\alpha}(\bv),
\end{equation}
where $\alpha = (\alpha_1, \cdots, \alpha_D)$ is a $D$-dimensional
multi-index, and
\begin{equation} \label{eq:v}
\bv = \frac{\bxi - \bu}{\sqrt{\theta}}.
\end{equation}
The basis functions $\mathcal{H}_{\theta,\alpha}$ are chosen as
\begin{equation} \label{eq:H}
\mathcal{H}_{\theta,\alpha}(\bv) = \prod_{d=1}^D
  \frac{1}{\sqrt{2\pi}} \theta^{-\frac{\alpha_d + 1}{2}}
  \He_{\alpha_d}(v_d) \exp \left(-\frac{v_d^2}{2}\right),
\end{equation}
where $\He_{\alpha_d}$ is the Hermite polynomial defined by
\begin{equation}
\He_n(x) = (-1)^n \exp(x^2/2)
  \frac{\mathrm{d}^n}{\mathrm{d}x^n} \exp(-x^2/2).
\end{equation}
Note that this is formally inconsistent with Grad's original
expression, but will be more convenient for our deduction
below\footnote{Grad uses symbols as $\mathscr{H}^{(n)}_{i_1 i_2 \cdots
i_n}$ to denote basis functions. This symbol is equivalent to $C
\mathcal{H}_{\alpha}, \, \alpha = e_{i_1} + e_{i_2} + \cdots +
e_{i_n}$ which is used here, where $C$ is a constant factor.
Inversely, $\mathcal{\alpha}$ can be expressed as
$\mathscr{H}_{\mathit{sub}}^{(|\alpha|)}$ with $\alpha_1$ ones,
$\alpha_2$ twos, $\cdots$, and $\alpha_D$ $D$'s in the subscript
$\mathit{sub}$. Thus our expansion is actually the same as what Grad
has done.}. The properties of Hermite polynomials can be found in many
handbooks such as \cite{Abramowitz}. Some useful ones are listed
below:
\begin{enumerate}
\item Orthogonality: $\displaystyle \int_{\bbR} \He_m(x)
\He_n(x) \exp(-x^2/2) \dd x = m! \sqrt{2\pi} \delta_{m,n}$;
\item Recursion relation: $\He_{n+1}(x) = x \He_n(x) - n \He_{n-1}(x)$;
\item Differential relation: $\He_n'(x) = n \He_{n-1}(x)$.
\end{enumerate}
It can be derived from the recursion relation and the
differential relation that
\begin{equation} \label{eq:differential}
[\He_n(x) \exp(-x^2/2)]' = -\He_{n+1}(x) \exp(-x^2/2).
\end{equation}
The expansion \eqref{eq:expansion} together with \eqref
{eq:moments} yields
\begin{equation} \label{eq:coef_restriction}
f_0 = \rho, \quad f_{e_i} = 0, \quad
  \sum_{d=1}^D f_{2e_d} = 0, \qquad i=1,\cdots,D,
\end{equation}
where $e_i$ stands for the multi-index with the $i$th component
$1$ and all other components $0$.

It is known that \eqref{eq:expansion} will result in an
``infinite moment system''. In order to make it numerically
solvable, we choose a positive integer $M \geqslant 2$ and
approximate \eqref{eq:expansion} by
\begin{equation} \label{eq:finite_expansion}
f(\bxi) \approx \sum_{|\alpha| \leqslant M}
  f_{\alpha} \mathcal{H}_{\theta,\alpha}(\bv).
\end{equation}
Using $F_M(\bu,\theta)$ to denote the linear space spanned by
all $\mathcal{H}_{\theta,\alpha}(\bv)$'s, where $|\alpha|
\leqslant M$ and $\bv$ is defined by \eqref{eq:v}, then
$F_M(\bu,\theta)$ is a finite dimensional subspace of
$L^2(\bbR^N, \exp(|\bv|^2/2) \dd \bv)$.

\begin{remark}
Based on such an expansion, the Maxwellian in the BGK collision
operator \eqref{eq:BGK} can simply be expressed by
\begin{equation}
f_M(\bxi) = \rho \mathcal{H}_{\theta,0}(\bv).
\end{equation}
With \eqref{eq:coef_restriction}, we have
\begin{equation} \label{eq:num_BGK}
Q_{\mathrm{BGK}}(f) = -\nu \sum_{1< |\alpha| \leqslant M}
  f_{\alpha} \mathcal{H}_{\theta,\alpha}(\bv).
\end{equation}
\end{remark}

\begin{remark}
The definition of $\bv$ (eq. \eqref{eq:v}) implies that we take the
mean velocity $\bu$ as the ``origin'' and $\theta$ as the scaling
factor when discretizing the distribution function.  Thus
\eqref{eq:coef_restriction} holds. If \eqref{eq:coef_restriction} is
violated, and we suppose $\bu' \in \bbR^D$, $\theta' \in \bbR^+$, and
a distribution function $f$ is approximated as
\begin{equation} \label{eq:f'}
f(\bxi) \approx \sum_{|\alpha| \leqslant M}
  f_{\alpha} \mathcal{H}_{\theta',\alpha}(\bv'),
  \quad \bv' = \frac{\bxi - \bu'}{\sqrt{\theta'}}
\end{equation}
with $f_{e_j}$'s or $\sum f_{2e_j}$ nonzero, then the associated
$\rho, \bu$ and $\theta$ can be calculated by substituting
\eqref{eq:f'} into \eqref{eq:moments}. Since
\begin{equation}
\He_0(x) = 1, \quad \He_1(x) = x, \quad \He_2(x) = x^2 - 1,
\end{equation}
employing the orthogonality of Hermite polynomials, the integrals
on the right hand sides of \eqref{eq:moments} can be directly
worked out as
\begin{equation} \label{eq:macro_vars}
\begin{split}
\rho &= f_0, \\
\rho \bu &= \rho \bu' + (f_{e_d})_{d=1,\cdots,D}^T, \\
\rho |\bu|^2 + D\rho \theta &= 2\rho\bu \cdot \bu' - \rho |\bu'|^2
  + \sum_{d=1}^D (\theta' f_0 + 2f_{2e_d}).
\end{split}
\end{equation}
T\"olke, Krafczyk, Schulz and Rank \cite{Tolke} use $\bu' \equiv
0$ in their discretization. Additionally, the heat flux can be
calculated by
\begin{equation}
q_j = \frac{1}{2} \int_{\bbR^D} |\bxi - \bu|^2 (\xi_j - u_j) \dd \bxi
    = 3\theta' f_{e_j} + 2 f_{3e_j} + \sum_{d=1}^D f_{e_j + 2e_d},
\quad \forall j=1,\cdots,D.
\end{equation}
\end{remark}

\subsection{Outline of the fractional step method} \label{sec:outline}
In this subsection, our numerical scheme will be outlined. Suppose
$\bx \in \bbR^N$ and $N \leqslant D$. $\mathcal{T}_h$ is a uniform
rectangular mesh in $\bbR^N$, with all grid lines parallel with the
axes, and each cell is identified by an $N$-dimensional multi-index
$\beta$. That is, for a fixed $\bx_0 \in \bbR^N$ and $\Delta x_j > 0,
\,j = 1,\cdots, N$,
\begin{equation}
\mathcal{T}_h = \{T_{\beta} = \bx_0 + [\beta_1 \Delta x_1,
  (\beta_1 + 1) \Delta x_1] \times \cdots \times
  [\beta_N \Delta x_N, (\beta_N + 1) \Delta x_N] :
  \beta \in \mathbb{Z}^N \}.
\end{equation}
Using $f_{\beta}^n(\bxi)$ to approximate the average distribution
function over the cell $T_{\beta}$ at time $t^n$, the Boltzmann
equation \eqref{eq:Boltzmann} can be solved by a standard fractional
step method:
\begin{enumerate}
\item {\it Convection step:} $\displaystyle
f_{\beta}^{n+1*}(\bxi) =
  f_{\beta}^n(\bxi) - \sum_{j=1}^N \frac{\Delta t^n}{\Delta x_j}
  [F_{\beta + \frac{1}{2} e_j}^n(\bxi) -
  F_{\beta - \frac{1}{2} e_j}^n(\bxi)].$
\item {\it Production step:} $f_{\beta}^{n+1}(\bxi) =
  f_{\beta}^{n+1*}(\bxi) + \Delta t^n Q_h(f_{\beta}^{n+1*})$.
\end{enumerate}
In the convection step, the finite volume method is employed
and $F_{\beta + \frac{1}{2}e_j}$ is the numerical flux between
cell $T_{\beta}$ and $T_{\beta + e_j}$. In the production step,
$Q_h$ is a transform over $F_M(\bu_{\beta}^{n+1*}, \theta_%
{\beta}^{n+1*})$, and is considered as an approximation to
$Q(\cdot, \cdot)$. Here $\bu_{\beta}^{n+1*}$ and $\theta_%
{\beta}^{n+1*}$ are the mean velocity and temperature corresponding
to the distribution function $f_{\beta}^{n+1*}$.

The numerical flux can be chosen from the standard ones in the
finite volume method, but in our framework, only central schemes
are available since the characteristic factorization of the flux
function is not available yet. A series of central schemes can
be found in textbook \cite{Toro}, and we choose the HLL scheme
\cite{HLL} in our numerical experiments, which reads
\begin{equation} \label{eq:HLL_flux}
F_{\beta + \frac{1}{2} e_j}^n(\bxi) = 
\begin{cases}
  \xi_j f_{\beta}^n(\bxi), & 0 \leqslant \lambda_j^L, \\[10pt]
  \displaystyle \frac{\lambda_j^R \xi_j f_{\beta}^n(\bxi) -
    \lambda_j^L \xi_j f_{\beta + e_j}^n(\bxi) +
    \lambda_j^L \lambda_j^R [f_{\beta + e_j}^n(\bxi) - f_{\beta}^n(\bxi)]}
    {\lambda_j^R - \lambda_j^L},
  & \lambda_j^L < 0 < \lambda_j^R, \\
  \xi_j f_{\beta + e_j}^n(\bxi), & 0 \geqslant \lambda_j^R.
\end{cases}
\end{equation}
$\lambda_j^L$ and $\lambda_j^R$ are the fastest signal velocities
arising from the solution of the Riemann problem, which will be
discussed in Section \ref{sec:est}. For all $\beta \in
\mathbb{Z}^N$, given
\begin{equation}
f_{\beta}^n(\bxi) = \sum_{|\alpha| \leqslant M}
  f_{\beta,\alpha}^n \mathcal{H}_{\beta,\alpha}^n(\bv_{\beta}^n),
\quad \bv_{\beta}^n =
  (\bxi - \bu_{\beta}^n) / (\theta_{\beta}^n)^{1/2},
\end{equation}
where $\mathcal{H}_{\beta,\alpha}^n = \mathcal{H}_{\theta_{\beta}^n,
\alpha}$, and $\bu_{\beta}^n, \theta_{\beta}^n$ are the mean
velocity and temperature in cell $\beta$, then $\xi_j
f_{\beta}^n(\bxi)$ can be calculated according to the recursion
relation of Hermite polynomials:
\begin{equation} \label{eq:flux}
\begin{split}
\xi_j f_{\beta}^n(\bxi) &=
  \left[(\theta_{\beta}^n)^{1/2} (v_{\beta}^n)_j + (u_{\beta}^n)_j \right]
  \sum_{|\alpha| \leqslant M} f_{\beta,\alpha}^n
  \mathcal{H}_{\beta,\alpha}^n (\bv_{\beta}^n) \\
&= \sum_{|\alpha| \leqslant M} f_{\beta,\alpha}^n \left[
  \theta_{\beta}^n \mathcal{H}_{\beta,\alpha + e_j}^n (\bv_{\beta}^n)
  + (u_{\beta}^n)_j \mathcal{H}_{\beta,\alpha}^n (\bv_{\beta}^n)
  + \alpha_j \mathcal{H}_{\beta,\alpha - e_j}^n (\bv_{\beta}^n)
\right].
\end{split}
\end{equation}
Since $|\alpha + e_j| = M + 1$ when $|\alpha| = M$, $\xi_j
f_{\beta}^n(\bxi)$ no longer exists in the space $F_M(\bu_{\beta}^n,
\theta_{\beta}^n)$. Thus, we need an additional ``projection step'' to
drag \eqref{eq:flux} back into $F_M(\bu_{\beta}^n, \theta_{\beta}^n)$.
This can be done by simply dropping the terms with $|\alpha + e_j| = M
+ 1$, since when $|\alpha| > M$, $\mathcal{H}_{\alpha}(\bv)$ is
orthogonal to $F_M(\bu, \theta)$ with respect to the inner product
\begin{equation}
(f, g) = \int_{\bbR^D} f(\bv) g(\bv) \exp
  \left( \frac{|\bv|^2}{2} \right) \dd \bv.
\end{equation}
However, the convection step is still uncompleted since it is
nontrivial to add up two functions lying in $F_M(\bu_{\beta}^n,
\theta_{\beta}^n)$ and $F_M(\bu_{\beta + e_j}^n, \theta_{\beta %
+ e_j}^n)$ respectively. This is a part of our major work and will
be discussed in \ref{sec:completion}.

As to the production step, the main job is to construct the
numerical collision operator $Q_h$. This is also implemented
by projecting $Q(f_{\beta}^{n+1*}, f_{\beta}^{n+1*})$ into
$F_M(\bu_{\beta}^{n+1*}, \theta_{\beta}^{n+1*})$. Precisely,
$Q_h$ is defined as
\begin{equation}
Q_h(f) = \sum_{|\alpha| \leqslant M}
  Q_{\alpha} \mathcal{H}_{\theta,\alpha}(\bv), \quad
  \forall f \in F_M(\bu, \theta),
\end{equation}
where
\begin{gather}
\label{eq:Q_alpha}
Q_{\alpha} = C_{\theta,\alpha} \int_{\bbR^D}
    Q(f, f)(\bxi) \mathcal{H}_{\theta,\alpha}(\bv)
    \exp(|\bv|^2/2) \dd \bv, \\
\label{eq:C}
C_{\theta,\alpha} = \frac{(2\pi)^{D/2}\theta^{D + |\alpha|}}
  {\alpha_1! \cdots \alpha_d!}.
\end{gather}
Further calculation requires the concrete forms of $f'$ and
$f_*'$ in \eqref{eq:f}. For BGK model \eqref{eq:BGK}, the numerical
collision operator has a simple explicit form \eqref{eq:num_BGK}.
In this situation, all $f_{\alpha}$'s can be decoupled, so
the production step can be performed by solving each $f_{\alpha}$
analytically. The scheme reads
\begin{enumerate}
\setcounter{enumi}{1}
\item {\it Production step (only for BGK model):}
\begin{displaymath}
  f_{\beta,\alpha}^{n+1} =
    f_{\beta,\alpha}^{n+1*} \exp(-\nu \Delta t^n),
  \quad \forall \alpha \in \mathbb{N}^D,
  \quad 0 < |\alpha| \leqslant M.
\end{displaymath}
\end{enumerate}

\begin{remark}
For the time integration, we use a single step Euler scheme for
both convection and production step. Actually, such formation
can be smoothly generalized to Runge-Kutta and Strang splitting
schemes.
\end{remark}

\begin{remark}
For other collision terms, such as ES-BGK model \cite{Holway} or
Maxwell molecules \cite{Ernst}, the expression of $Q_h$ can be much
more complicated. For the ES-BGK model, it is always possible to get
$Q_{\alpha}$'s by direct integration. For Maxwell molecules and
linearized Boltzmann collision operator, the method in \cite
{Torrilhon} can be employed to generate the numerical collision
operator $Q_h$. For simplicity, these models are not considered in
this paper.
\end{remark}

\subsection{Completion of the convection step} \label{sec:completion}
Two points remain unclear for the convection step. One is that we
need to find a way to add up two functions in different spaces
$F_M(\bu_1, \theta_1)$ and $F_M(\bu_2, \theta_2)$, so that it is
applicable to calculate the numerical fluxes and to accumulate them to
the solution at the last time step. And the other is the estimation of
the characteristic velocities $\lambda_j^L$ and $\lambda_j^R$.

\subsubsection{Projection between two different spaces} \label{sec:projection}
Assume $f_1 \in F_M(\bu_1, \theta_1)$ and $f_2 \in F_M(\bu_2,
\theta_2)$. Obviously, when $\bu_1 \neq \bu_2$ or $\theta_1 \neq
\theta_2$, direct calculation of $f_1 + f_2$ is inapplicable.
Therefore, we want to find $\tilde{f}_1 \in F_M(\bu_2, \theta_2)$ such
that $\tilde{f}_1$ is some approximation of $f_1$ in $F_M(\bu_2,
\theta_2)$. In order to realize such transformation, we propose a fast
projection method which has a time complexity of $O(M^D)$ below.

First, let us consider the case of $M = \infty$, and $f \in
F_{\infty}(\bu_1, \theta_1) \cap F_{\infty}(\bu_2, \theta_2)$.
Then, $f$ has the following two representations
\begin{equation} \label{eq:f1}
f(\bxi) = \sum_{\alpha \in \bbN^D}
  f_{1,\alpha} \mathcal{H}_{\theta_1,\alpha} (\bv_1),
\quad \bv_1 = (\bxi - \bu_1) / \sqrt{\theta_1},
\end{equation}
and
\begin{equation} \label{eq:f2}
f(\bxi) = \sum_{\alpha \in \bbN^D}
  f_{2,\alpha} \mathcal{H}_{\theta_2,\alpha} (\bv_2),
\quad \bv_2 = (\bxi - \bu_2) / \sqrt{\theta_2}.
\end{equation}
Suppose all $f_{1,\alpha}$'s are known, and we want to solve
all $f_{2,\alpha}$'s. Let $\htheta = \sqrt{\theta_1 / \theta_2}$
and $\bw = (\bu_1 - \bu_2) / \sqrt{\theta_2}$. It is obvious that
\begin{equation}
\bv_2 = \htheta \bv_1 + \bw, \quad
\mathcal{H}_{\theta_2, \alpha} =
  \htheta^{|\alpha| + D} \mathcal{H}_{\theta_1, \alpha}.
\end{equation}
Joining \eqref{eq:f1} and \eqref{eq:f2}, we have
\begin{equation} \label{eq:f_relation}
\sum_{\alpha \in \bbN^D} f_{1,\alpha}
  \mathcal{H}_{\theta_1, \alpha}(\bv_1)
= \sum_{\alpha \in \bbN^D} f_{2,\alpha}
  \htheta^{|\alpha| + D} \mathcal{H}_{\theta_1, \alpha}
  (\htheta \bv_1 + \bw).
\end{equation}
Now we introduce an auxiliary function $F(\bv, \tau)$, defined
as
\begin{equation} \label{eq:F}
F(\bv, \tau) = \sum_{\alpha \in \bbN^D}
  F_{\alpha}(\tau) [(\htheta - 1) \tau + 1]^{|\alpha| + D}
  \mathcal{H}_{\theta_1, \alpha} \left(
    [(\htheta - 1) \tau + 1] \bv + \tau \bw
  \right),
\end{equation}
which satisfies
\begin{equation} \label{eq:F_restriction}
F_{\alpha}(0) = f_{1,\alpha}, \quad
  \forall \alpha \in \bbN^D,
\qquad \textrm{and} \qquad F(\bv, 0) = F(\bv, 1).
\end{equation}
Comparing \eqref{eq:F} \eqref{eq:F_restriction} with \eqref
{eq:f_relation}, it can be found that for any $\alpha \in
\bbN^D$, $F_{\alpha}(1)$ is just $f_{2,\alpha}$ which
is to be solved. Moreover, if we suppose
\begin{equation} \label{eq:const_func}
\frac{\partial F}{\partial \tau} \equiv 0, \quad
  \forall \tau \in [0,1],
\end{equation}
then an infinite ordinary differential system of $\left\{
F_{\alpha}(\tau) \right\}_{\alpha \in \bbN^D}$
can be obtained.

The detailed calculation of $\dfrac{\partial F}{\partial \tau}$ can be
found in Appendix \ref{sec:df_dt}, and we only show the final result
here:
\begin{equation}
\frac{\partial}{\partial \tau} F(\bv, \tau)
= \sum_{\alpha \in \bbN^D} S^{-(|\alpha| + D)}
  \mathcal{H}_{\theta_1,\alpha} \left\{
    \frac{\mathrm{d}}{\mathrm{d}\tau} F_{\alpha}
    - \sum_{d=1}^D S^2 \left[
      \theta_1 R F_{\alpha-2e_d}
      + w_d \sqrt{\theta_1} F_{\alpha - e_d}
    \right]
  \right\},
\end{equation}
where
\begin{equation} \label{eq:RS}
R(\tau) = \frac{\htheta - 1}{(\htheta - 1) \tau + 1}, \quad
S(\tau) = 1 - \tau R(\tau) = \frac{1}{(\htheta - 1) \tau + 1},
\end{equation}
and the parameter of $\mathcal{H}_{\theta_1,\alpha}$ is
$[(\hat{\theta} - 1) \tau + 1] \bv + \tau \bw$. As required in
\eqref{eq:const_func}, 
\begin{equation} \label{eq:ode}
\frac{\mathrm{d}}{\mathrm{d}\tau} F_{\alpha} =
  \sum_{d=1}^D S^2 \left[
    \theta_1 R F_{\alpha - 2e_d}
    + w_d \sqrt{\theta_1} F_{\alpha - e_d}
  \right], \quad \forall \alpha \in \bbN^D,
    \quad \forall \tau \in [0,1]
\end{equation}
must hold. \eqref{eq:ode} is an infinite system, but for
any $M \geqslant 2$, if we consider only a subsystem
containing all equations with $|\alpha| \leqslant M$, it is
still closed. Therefore, in order to project a function
$f \in F_{\infty}(\bu_1, \theta_1) \cap F_{\infty}(\bu_2,
\theta_2)$ into $F_M(\bu_2, \theta_2)$, it is only needed
to solve \eqref{eq:ode} for all $|\alpha| \leqslant M$.

For an arbitrary function $f(\bxi)$ which is defined on
$\bbR^D$, its projection to $F_M(\bu,\theta)$ is defined as
\begin{equation} \label{eq:project}
\Pi_{\bu,\theta} f = \sum_{|\alpha| \leqslant M} \left[
  C_{\theta,\alpha} \int_{\bbR^D} f(\bxi)
    \mathcal{H}_{\theta,\alpha}(\bv)
    \exp(|\bv|^2/2) \dd \bv
\right] \mathcal{H}_{\theta,\alpha}(\bv),
\quad \bv = (\bxi - \bu) / \sqrt{\theta}.
\end{equation}
See \eqref{eq:C} for the definition of $C_{\theta,\alpha}$.
The following proposition provides an algorithm to project
a function in $F_M(\bu_1,\theta_1)$ to $F_M(\bu_2,\theta_2)$.
\begin{proposition} \label{prop:projection}
Suppose $f \in F_M(\bu_1, \theta_1)$ can be represented by
\begin{equation}
f(\bxi) = \sum_{|\alpha| \leqslant M} f_{1,\alpha}
  \mathcal{H}_{\theta_1,\alpha}(\bv_1),
\quad \bv_1 = (\bxi - \bu_1) / \sqrt{\theta_1}.
\end{equation}
For some $\bu_2 \in \mathbb{R}^D$ and $\theta_2 > 0$,
$\{F_{\alpha}(\tau)\}_{|\alpha| \leqslant M}$ satisfies
\begin{equation} \label{eq:finite_ode}
\begin{cases}
\displaystyle
\frac{\mathrm{d}}{\mathrm{d}\tau} F_{\alpha} =
  \sum_{d=1}^D S^2 \left[
    \theta_1 R F_{\alpha - 2e_d}
    + w_d \sqrt{\theta_1} F_{\alpha - e_d}
  \right], & \forall \tau \in [0,1], \\[15pt]
F_{\alpha}(0) = f_{1,\alpha},
\end{cases}
\end{equation}
where $S$ and $R$ are given by \eqref{eq:RS}, and $\bw =
(\bu_1 - \bu_2) / \sqrt{\theta_2}$. Let
\begin{equation} \label{eq:projection}
g(\bxi) = \sum_{|\alpha| \leqslant M} F_{\alpha}(1)
  \mathcal{H}_{\theta_2,\alpha}(\bv_2),
\quad \bv_2 = (\bxi - \bu_2) / \sqrt{\theta_2}.
\end{equation}
Then $g(\bxi) \in F_M(\bu_2, \theta_2)$ and $g(\bxi)$
satisfies
\begin{equation} \label{eq:conservation}
\int_{\bbR^D} p(\bxi) f(\bxi) \dd \bxi =
  \int_{\bbR^D} p(\bxi) g(\bxi) \dd \bxi, \quad
  \forall p(\bxi) \in P_M(\bbR^D).
\end{equation}
\end{proposition}

\begin{proof}
Let $\Lambda = [-1,1]^D$ and
\begin{equation}
P_{\theta_1,\alpha}(\bv_1) =
  \mathcal{H}_{\theta_1,\alpha}(\bv_1) \exp(|\bv_1|^2/2),
  \quad \forall \alpha \in \bbN^D.
\end{equation}
Thus $\{P_{\theta_1,\alpha}(\bv_1)\}
_{\alpha \in \bbN^D}$ forms a complete basis of
$L^2(\Lambda)$. It follows that for an arbitrary set of
$\{b_{\alpha}\}_{|\alpha| \leqslant M}$, there exists 
a function $\tilde{f} \in L^2(\Lambda)$, such that
\begin{equation}
\int_{\Lambda} \tilde{f}(\bv_1)
  P_{\theta_1, \alpha} (\bv_1) \dd \bv_1 = b_{\alpha},
  \quad \forall \alpha \in \bbN^D, |\alpha| \leqslant M.
\end{equation}
When $\displaystyle b_{\alpha} = \int_{\bbR^D} f(\bxi)
P_{\theta_1, \alpha} (\bv_1) \dd \bv_1$, a zero extension
of $\tilde{f}$ onto $\bbR^D$ gives that
\begin{equation} \label{eq:equiv}
\int_{\bbR^D} \tilde{f}(\bv_1)
  \mathcal{H}_{\theta_1,\alpha}(\bv_1)
  \exp \left( \frac{|\bv_1|^2}{2} \right) \dd \bv_1
= \int_{\bbR^D} f(\bxi)
  \mathcal{H}_{\theta_1,\alpha}(\bv_1)
  \exp \left( \frac{|\bv_1|^2}{2} \right) \dd \bv_1
\end{equation}
holds for all $\alpha \in \bbN^D$, $|\alpha| \leqslant M$.
Let $\tilde{g}(\bxi) = \tilde{f}(\bv_1)$. Since $\tilde{g}$
has a compact support on $\bbR^D$, we have $\tilde{g} \in
F_{\infty}(\bu_1, \theta_1) \cap F_{\infty}(\bu_2,\theta_2)$.
Now \eqref{eq:equiv} and the orthogonality of Hermite
polynomials implies that if
\begin{equation}
\tilde{g}(\bxi) = \sum_{\alpha \in \bbN^D} \tilde{g}_{1,\alpha}
  \mathcal{H}_{\theta_1,\alpha}(\bv_1),
\end{equation}
then $\tilde{g}_{1,\alpha} = f_{1,\alpha}$ for all $\alpha
\in \bbN^D$, $|\alpha| \leqslant M$.

The preceding analysis shows that if
\begin{equation}
\tilde{g}(\bxi) = \sum_{\alpha \in \bbN^D} \tilde{g}_{2,\alpha}
  \mathcal{H}_{\theta_2,\alpha}(\bv_2),
\end{equation}
then $\tilde{g}_{2,\alpha} = F_{\alpha}(1)$ for all
$\alpha \in \bbN^D$, $|\alpha| \leqslant M$. Employing the
orthogonality of Hermite polynomials again, we can deduce that
for any $p(\bxi) \in P_M(\bbR^D)$,
\begin{equation}
\int_{\bbR^D} p(\bxi) f(\bxi) \dd \bxi =
  \int_{\bbR^D} p(\bxi) \tilde{g}(\bxi) \dd \bxi =
  \int_{\bbR^D} p(\bxi) g(\bxi) \dd \bxi.
\end{equation}
Thus \eqref{eq:conservation} is established.
\end{proof}

Based on this proposition, the projection requires to solve an
ordinary differential system \eqref{eq:finite_ode}. Let row vectors
$\bF_m (\tau)$ and $\bF(\tau)$ be
\begin{equation}
\bF_m(\tau) = \left( F_{\alpha}(\tau) \right)_{|\alpha|=m},
  \quad 0 \leqslant m \leqslant M,
\end{equation}
and
\begin{equation}
\bF(\tau) = (\bF_0(\tau), \bF_1(\tau), \cdots, \bF_M(\tau)).
\end{equation}
Thus the ordinary differential equations \eqref{eq:finite_ode} can be
simplified as
\begin{equation} \label{eq:ode_vec}
\frac{\mathrm{d}}{\mathrm{d} \tau} \bF(\tau)
  = \bF(\tau) {\bf A}(\tau), \quad \tau \in [0,1].
\end{equation}
Eq. \eqref{eq:finite_ode} reveals that ${\bf A}(\tau)$ is an upper
triangular matrix with vanished diagonal entries. Therefore,
\eqref{eq:ode_vec} can actually be solved by recursive integration.
However, the direct integration will leads to $O(M^{2D})$
calculations, so we solve \eqref{eq:ode_vec} by applying $O(1)$ steps
of Runge-Kutta numerical integration, which is unconditionally stable
due to the special form of ${\bf A}(\tau)$. Since ${\bf A}(\tau)$ is
sparse, each Runge-Kutta step costs only $O(M^D)$ calculations. Thus
the whole projection has a time complexity of $O(M^D)$.

Now let us return to the convection step. Using $\Pi_{\beta}^n$
to denote the operator that projects any function to the space
$F_M(\bu_{\beta}^n, \theta_{\beta}^n)$, then the convection step
is described as following:
\begin{enumerate}
\item {\it Convection step:}
  \begin{enumerate}
  \item \label{conv_step1}
    Apply the convection within $F_M(\bu_{\beta}^n, \theta_%
    {\beta}^n)$:
    \begin{equation} \label{eq:step1a}
    f_{\beta}^{n**}(\bxi) = f_{\beta}^n(\bxi)
      - \sum_{j=1}^N \frac{\Delta t^n}{\Delta x_j}
      [(\Pi_{\beta}^n F_{\beta + \frac{1}{2} e_j}^n)(\bxi) -
        (\Pi_{\beta}^n F_{\beta - \frac{1}{2} e_j}^n)(\bxi)];
    \end{equation}
  \item \label{conv_step2}
    Use \eqref{eq:macro_vars} to calculate the mean velocity
    $\bu_{\beta}^{n**}$ and temperature $\theta_{\beta}^{n**}$
    for the distribution function $f_{\beta}^{n**}$;
  \item \label{conv_step3}
    Let $\bu_{\beta}^{n+1*} = \bu_{\beta}^{n**}$, $\theta_
    {\beta}^{n+1*} = \theta_{\beta}^{n**}$, and $f_{\beta}^{n+1*}
    = \Pi_{\beta}^{n+1*} f_{\beta}^{n**}$.
  \end{enumerate}
\end{enumerate}
When implementing Step \ref{conv_step1}, $\Pi_{\beta}^n$ is
actually applied on each term of $F_{\beta \pm \frac{1}{2} e_j}$,
and the result $f_{\beta}^{n**}$ no longer satisfies \eqref
{eq:coef_restriction}. So the mean velocity and temperature
need to be recalculated in Step \ref{conv_step2}. In Step
\ref{conv_step3}, $f_{\beta}^{n**}$ is adjusted to $f_{\beta}
^{n+1*}$ such that \eqref{eq:coef_restriction} holds for
$f_{\beta}^{n+1*}$. Later on, in the production step, the
mean velocity and temperature are not changed, so \eqref
{eq:coef_restriction} still holds for $f_{\beta}^{n+1}$.
The conservation of the convection step follows from proposition
\ref{prop:projection} and the conservative form of the finite
volume scheme. To be specific, for any $p(\bxi) \in P_M(\bbR^D)$, we
have
\begin{equation}
\begin{split}
& \phantom{={}} \sum_{\beta \in \mathbb{Z}^N}
  \int_{\bbR^D} p(\bxi) f_{\beta}^{n**}(\bxi) \dd \bxi \\
& = \sum_{\beta \in \mathbb{Z}^N}
  \int_{\bbR^D} p(\bxi) f_{\beta}^{n}(\bxi) \dd \bxi
  - \sum_{\beta \in \mathbb{Z}^N} \sum_{j=1}^N
    \frac{\Delta t^n}{\Delta x_j} \cdot \\
& \qquad \quad \left[
    \int_{\bbR^D} p(\bxi)
      (\Pi_{\beta}^n F_{\beta + \frac{1}{2} e_j}^n)(\bxi)
    \dd \bxi -
    \int_{\bbR^D} p(\bxi)
      (\Pi_{\beta}^n F_{\beta - \frac{1}{2} e_j}^n)(\bxi)
    \dd \bxi
  \right] \\
& = \sum_{\beta \in \mathbb{Z}^N}
  \int_{\bbR^D} p(\bxi) f_{\beta}^{n}(\bxi) \dd \bxi
  - \sum_{\beta \in \mathbb{Z}^N} \sum_{j=1}^N
    \frac{\Delta t^n}{\Delta x_j} \cdot \\
& \qquad \quad \left[
    \int_{\bbR^D} p(\bxi)
      F_{\beta + \frac{1}{2} e_j}^n(\bxi)
    \dd \bxi -
    \int_{\bbR^D} p(\bxi)
      F_{\beta - \frac{1}{2} e_j}^n(\bxi)
    \dd \bxi
  \right] \qquad \text{[Using \eqref{eq:conservation}]} \\
& = \sum_{\beta \in \mathbb{Z}^N}
  \int_{\bbR^D} p(\bxi) f_{\beta}^{n}(\bxi) \dd \bxi.
\end{split}
\end{equation}
Thus quantities such as mass, total momentum and total energy are
conservative.

\subsubsection{Estimation of the characteristic velocities} \label{sec:est}
In order to estimate $\lambda_j^L$ and $\lambda_j^R$ that are used in
\eqref{eq:HLL_flux}, we need to investigate into the expression of
numerical flux $\Pi_{\beta}^n F_{\beta \pm \frac{1}{2} e_j}$
carefully. Precisely, we should make sure the Riemann problem that
such a numerical flux solves.  In order to simplify the notation, we
consider only the following form:
\begin{equation} \label{eq:num_flux}
F_1 = \frac{\lambda_j^R \Pi_{f_1} (\xi_j f_1) -
  \lambda_j^L \Pi_{f_2,f_1} \Pi_{f_2} (\xi_j f_2) +
  \lambda_j^L \lambda_j^R (\Pi_{f_2,f_1} f_2 - f_1)}
  {\lambda_j^R - \lambda_j^L},
\end{equation}
where $f_1 \in F_M(\bu_1, \theta_1)$ and $f_2 \in F_M(\bu_2,
\theta_2)$, and both satisfy \eqref{eq:coef_restriction}.
For $f \in F_M(\bu, \theta)$ which satisfies \eqref
{eq:coef_restriction}, $\Pi_f$ is the projection
operator from $F_{\infty}(\bu, \theta)$ to $F_M(\bu, \theta)$,
which simply discards the terms with orders higher than $M$.
$\Pi_{f_2,f_1}$ is the projection operator from $F_M(\bu_2,
\theta_2)$ to $F_M(\bu_1, \theta_1)$. Then if we take
\begin{equation}
f_1 = f_{\beta}^n, \quad f_2 = f_{\beta + e_j}^n,
\end{equation}
$F_1$ is exactly the same as $\Pi_{\beta}^n F_{\beta + \frac
{1}{2} e_j}$ in the case of $\lambda_j^L < 0 < \lambda_j^R$.
Similarly, let
\begin{equation}
F_2 = \frac{\lambda_j^R \Pi_{f_1, f_2} \Pi_{f_1} (\xi_j f_1)
  - \lambda_j^L \Pi_{f_2} (\xi_j f_2)
  + \lambda_j^L \lambda_j^R (f_2 - \Pi_{f_1,f_2} f_1)}
  {\lambda_j^R - \lambda_j^L}.
\end{equation}
Then $F_2$ is just $\Pi_{\beta}^n F_{\beta - \frac {1}{2} e_j}$
if $f_1 = f_{\beta - e_j}^n$ and $f_2 = f_{\beta}^n$ in the
case of $\lambda_j^L < 0 < \lambda_j^R$. Due to the similar
forms of $F_1$ and $F_2$, only $F_1$ is considered below.

The nature of $F_1$ can be depicted with the help of the following
proposition:
\begin{proposition} \label{prop:invertible}
$\Pi_{f_2,f_1}$ is invertible.
\end{proposition}

\begin{proof}
Denote the projection operator from $F_M(\bu_1, \theta_1)$
to $F_M(\bu_2, \theta_2)$ as $\Pi_{f_1,f_2}$. We are going to
prove that $\Pi_{f_1,f_2} \Pi_{f_2,f_1}$ is the identity
operator. Proposition \ref{prop:projection} shows that for any $f
\in F_M(\bu_2, \theta_2)$ and $p \in P_M(\bbR^D)$,
\begin{equation}
\int_{\bbR^D} p(\bxi) f(\bxi) \dd \bxi =
  \int_{\bbR^D} p(\bxi) (\Pi_{f_2,f_1}f)(\bxi) \dd \bxi =
  \int_{\bbR^D} p(\bxi)
    (\Pi_{f_1,f_2} \Pi_{f_2,f_1}f)(\bxi) \dd \bxi.
\end{equation}
That is,
\begin{equation}
\int_{\bbR^D} p(\bxi)
  [(I - \Pi_{f_1,f_2}\Pi_{f_2,f_1})f](\bxi) \dd \bxi = 0.
\end{equation}
Choosing $p(\bxi) = \mathcal{H}_{\theta_1,\alpha}(\bv_1)
\exp(|\bv_1|^2/2)$ for $\alpha \in \bbN^D$, $|\alpha|
\leqslant M$ respectively, and making use of the orthogonality
of Hermite polynomials, it follows that
\begin{equation}
(I - \Pi_{f_1,f_2} \Pi_{f_2,f_1})f \equiv 0, \quad
  \forall f \in F_M(\bu_2, \theta_2).
\end{equation}
Similarly, it can be proved that $\Pi_{f_2,f_1} \Pi_{f_1,f_2}$
is also the identity operator. Thus $\Pi_{f_2,f_1}$ is invertible.
\end{proof}

Now let us turn back to the numerical flux \eqref{eq:num_flux}.
Let $\tilde{f}_2 = \Pi_{f_2,f_1} f_2$. Based on proposition
\ref{prop:invertible}, we rewrite \eqref{eq:num_flux} as
\begin{equation}
F_1 = \frac{\lambda_j^R \Pi_{f_1,f_1} \Pi_{f_1}
  (\xi_j \Pi_{f_1, f_1}^{-1} f_1) -
  \lambda_j^L \Pi_{f_2,f_1} \Pi_{f_2}
    (\xi_j \Pi_{f_2,f_1}^{-1} \tilde{f}_2) +
  \lambda_j^L \lambda_j^R (\tilde{f}_2 - f_1)}
  {\lambda_j^R - \lambda_j^L},
\end{equation}
where $\Pi_{f_1,f_1}$ is the projection from $F_M(\bu_1,
\theta_1)$ to itself, which is actually the identity operator.
Now it is clear that the corresponding Riemann problem of
$F_1$ is
\begin{equation} \label{eq:Riemann}
\begin{split}
& \frac{\partial f}{\partial t} + \frac{\partial}{\partial x}
  [\Pi_{f,f_1} \Pi_f (\xi_j \Pi_{f,f_1}^{-1}) f] = 0, \\
& \begin{cases}
    f(0, x) = f_1, & x < 0, \\
    f(0, x) = \tilde{f}_2, & x > 0.
  \end{cases}
\end{split}
\end{equation}
Here $f$ always lies in $F_M(\bu_1,\theta_1)$, and the meanings
of $\Pi_{f,f_1}$ and $\Pi_f$ have been changed a little. Suppose
$\bu$ and $\theta$ are the mean velocity and temperature
associated with $f$, whose explicit expressions can be obtained
from \eqref{eq:macro_vars}.  Then $\Pi_{f,f_1}$ is defined as
the projection operator from $F_M(\bu, \theta)$ to $F_M(\bu_1,
\theta_1)$, and $\Pi_f$ is defined as the projection operator
form $F_{M+1}(\bu, \theta)$ to $F_M(\bu, \theta)$.

The characteristic velocities of Riemann problem \eqref{eq:Riemann} 
seem to be difficult to obtain. Therefore, in order to give an
estimation of $\lambda_j^L$ and $\lambda_j^R$, we choose a fixed
distribution function $f^* \in F_M(\bu_1, \theta_1)$ that lies
``between'' $f_1$ and $\tilde{f}_2$, and linearize \eqref{eq:Riemann}
as
\begin{equation} \label{eq:lattice}
\frac{\partial f}{\partial t} + \Pi_{f^*,f_1} \Pi_{f^*} \left(
  \xi_j \Pi_{f^*,f_1}^{-1} \frac{\partial f}{\partial x}
\right) = 0.
\end{equation}
Thus, we only need to estimate the eigenvalues of $\Pi_{f^*,f_1}
\Pi_{f^*} \xi_j \Pi_{f^*,f_1}^{-1}$, which is an operator on
$F_M(\bu_1,\theta_1)$. Since $\Pi_{f^*,f_1}$ is linear and invertible,
the problem can be further simplified as the estimation of eigenvalues
of $\Pi_{f^*} \xi_j$, which is an operator on $F_M(\bu^*, \theta^*)$.
Taking $\bv^* = (\bxi - \bu^*) / \sqrt{\theta^*}$, we have
\begin{equation} \label{eq:pi_star}
\Pi_{f^*} \xi_j = \Pi_{f^*} (u_j^* + v_j^*\sqrt{\theta^*})
  = u_j^* I + \sqrt{\theta^*} \Pi_{f^*} v_j^*.
\end{equation}
For the eigenvalues of $\Pi_{f^*} v_j^*$, we have the following
proposition:
\begin{proposition} \label{prop:eigenvalue}
The eigenvalues of $\Pi_{f^*} v_j^*$ are formed by all the zeros
of $\He_{m+1}(x)$, $m = 0,\cdots,M$.
\end{proposition}
\begin{proof}
Suppose $m \in \{0, \cdots, M\}$, and all zeros of $\He_{m+1}(x)$
are denoted as $x_0, \cdots, x_m$. Note that all $x_i$'s are
real and different (see e.g. \cite{Shen}), so we can assume that
\begin{equation}
x_0 < \cdots < x_m.
\end{equation}
For any $i \in \{0, \cdots, m\}$, there exists a unique polynomial
$p_{i,m}(x) \in P_m(x)$ that satisfies $p_{i,m}(x_k) = \delta_{ik}$,
$k = 0,\cdots,m$. Let $\alpha \in \mathbb{N}^D$ satisfy $|\alpha|
= M - m$ and $\alpha_j = 0$. We are going to prove
\begin{equation} \label{eq:eigenvector}
\Pi_{f^*} [v_j^* p_{i,m}(v_j^*) \mathcal{H}_{\theta^*,\alpha}(\bv^*)]
  = x_i p_{i,m}(v_j^*) \mathcal{H}_{\theta^*,\alpha}(\bv^*).
\end{equation}
The definition of $\Pi_{f^*}$ shows
\begin{equation} \label{eq:proj}
\Pi_{f^*} [v_j^* p_{i,m}(v_j^*) \mathcal{H}_{\theta^*,\alpha}(\bv^*)]
  = v_j^* p_{i,m}(v_j^*) \mathcal{H}_{\theta^*,\alpha}(\bv^*) -
    C_{j,\alpha} \He_{m+1}(v_j^*)
      \mathcal{H}_{\theta^*,\alpha}(\bv^*),
\end{equation}
where $C_{j,\alpha}$ is a properly selected constant such that
\eqref{eq:proj} lies in $F_M(\bu^*, \theta^*)$. Thus, for any
$k \in \{0, \cdots, m\}$ and $\bv_k^*$ satisfying $v_{k,j}^* =
x_k$, we have
\begin{equation}
\Pi_{f^*} [v_j^* p_{i,m}(v_j^*)
  \mathcal{H}_{\theta^*,\alpha}(\bv^*)] \big|_{\bv^* = \bv_k^*}
= x_k \delta_{ik} \mathcal{H}_{\theta^*,\alpha} (\bv_k^*).
\end{equation}
Then \eqref{eq:eigenvector} holds due to the uniqueness of $p_{i,m}$.

It remains to prove that $\Pi_{f^*} v_j^*$ has no other eigenvalues.
Let us count how many eigenvectors are included in the form of
$p_{i,m}(v_j^*) \mathcal{H}_{\theta^*, \alpha}(\bv^*)$. Consider an
arbitrary $\tilde{\alpha} \in \bbN^D$ and $|\tilde{\alpha}| \leqslant
M$. Let
\begin{equation}
i = \tilde{\alpha}_j, \quad \alpha = \tilde{\alpha} - i e_j,
  \quad m = M - |\alpha|.
\end{equation}
Obviously $i, \alpha$ and $m$ satisfy $0 \leqslant i \leqslant M$ and
$\alpha_j = 0$. Thus each $\tilde{\alpha}$ uniquely determines an
eigenvector. However, the number of such $\tilde{\alpha}$'s are equal
to the dimension of space $F_M(\bu^*, \theta^*)$, so $\Pi_{f^*} v_j^*$
has no other eigenvectors, thus no other eigenvalues, either.
\end{proof}

According to Proposition \ref{prop:eigenvalue}, the smallest and
largest eigenvalues of $\Pi_{f^*} v_j^*$ are the smallest and largest
zeros of $\He_{M+1}(x)$, denoted by $x_0$ and $x_M$.
\eqref{eq:pi_star} shows that the smallest and largest eigenvalues of
$\Pi_{f^*} \xi_j^*$ are $u_j^* + x_0 \sqrt{\theta^*}$ and $u_j^* + x_M
\sqrt{\theta^*}$ respectively. Since $f^*$ lies ``between'' $f_1$ and
$\tilde{f}_2$, we use
\begin{gather}
\label{eq:lambda_L}
\lambda_j^L = \min \{u_{1,j} + x_0 \sqrt{\theta_1},
  u_{2,j} + x_0 \sqrt{\theta_2} \}, \\
\label{eq:lambda_R}
\lambda_j^R = \max \{u_{1,j} + x_M \sqrt{\theta_1},
  u_{2,j} + x_M \sqrt{\theta_2} \}
\end{gather}
while computing numerical fluxes. In our implementation, we use
the subroutine in ALGLIB \cite{AIGLIB} to calculate the roots of
$\He_{M+1}$, and $\lambda_j^L$ and $\lambda_j^R$ are also used
in the CFL condition to determine the time step length.

\begin{remark}
The numerical method described in this section is only of the first
order. In order to extend it to higher order schemes, reconstruction
techniques need to be added to the finite volume scheme. Since
addition and subtraction between two distribution functions are
already available, it is only needed to determine a proper ``slope'',
which can probably be done with the help of the standard slope limters
used in the normal finite volume schemes.
\end{remark}

\subsection{Relation with the Grad-type moment method and the
LBE model}
The relation and difference between the Grad-type moment method
\cite{Grad, Grad1958} and the LBE (lattice Boltzmann equations)
model \cite{Succi} are summarized in \cite{Shan}, where both models
are considered to be some approximation to the Boltzmann equation
by an Hermite polynomial expansion. The expansion \eqref{eq:expansion}
and truncation \eqref{eq:finite_expansion} are exactly the same
as what have been done by Grad \cite{Grad}, which means the method
described above is actually solving the Grad-type moment equations.
For $M = 3, 4, 5$ and $D = 3$, it corresponds to the $20, 35,
56$-moment equations which take the complete $M$th order moments.
However, systems such as $13, 26, 45$-moment equations are not
included. Those complete $M$th order moment equations are popular
in extended thermodynamics, see e.g. \cite{Muller,Au,Weiss}. A
software called $ET_{XX}$ \cite{ETXX} is developed by J. Au,
H. Struchtrup and M. Torrilhon to generate equations for arbitrary
order of moments, however, explicit moment equations, when written
in conservative form, require $O(M^{2D})$ calculations for the flux
function, and this is reduced to $O(M^D)$ in our numerical
formation.

Now that our numerical strategy is equivalent to Grad's moment
method, one can refer to \cite{Shan} for the precise difference
between our method and the LBE model. According to \cite{Shan}, the
linearized equation \eqref{eq:lattice} can be interpreted as the LBE
model, which places the center of the lattice at $\bu^*$. Thus such
linearization is reasonable.

\section{Regularization of the moment method}
The main drawback of Grad's moment method is that its hyperbolicity
yields unphysical subshocks \cite{Grad1952}. The behavior in the
case of high order moment equations can be found in \cite{Au, Weiss}.
\cite{Grad1958} provided a way to regularize Grad's moment equations,
and it was further studied in \cite{Struchtrup2003, Torrilhon2004} as
R13 equations. Now we follow the regularization technique in
\cite{Struchtrup2003} and regularize the numerical method introduced
in Section \ref{sec:Grad} in exactly the same way.

\subsection{Chapman-Enskog expansion around the truncated distribution}
Independent of Grad's method, Chapman-Enskog expansion \cite{Chapman,
Enskog} is another important method for deriving equations of
macroscopic variables. Following the generic procedure of
Chapman-Enskog expansion, a scaling parameter $\varepsilon$ is
introduced on the right hand side of the Boltzmann equation. However,
according to \cite{Struchtrup2003, Struchtrup}, only the high order
part instead of the whole collision term is scaled:
\begin{equation} \label{eq:asym}
\frac{\partial f}{\partial t} + 
  \bxi \cdot \nabla_{\bx} f
= Q(f^0, f^0) + \frac{1}{\varepsilon} [Q(f,f) - Q(f^0, f^0)],
\end{equation}
where $f^0$ is a truncation of the distribution function \eqref
{eq:expansion} defined by
\begin{equation} \label{eq:truncation}
f^0(\bxi) = \sum_{|\alpha| \leqslant M} f_{\alpha}
  \mathcal{H}_{\theta, \alpha}(\bv).
\end{equation}
The scaled part $Q(f,f) - Q(f^0,f^0)$ will be denoted as $\tilde{Q}$
below. Let us apply the Chapman-Enskog expansion around $f^0$, i.e.
expand $f$ by
\begin{equation}
f = f^0 + \varepsilon f^1 + \varepsilon^2 f^2 + \cdots,
\end{equation}
and we require that the truncation at any term of this expansion
keeps the same values of all the moments with orders less than or
equal to $M$.  Suppose the scaled part of the collision term
$\tilde{Q}$ has a corresponding expansion
\begin{equation}
\tilde{Q} = \varepsilon \tilde{Q}^1 +
  \varepsilon^2 \tilde{Q}^2 + \cdots.
\end{equation}
It is reasonable to assume that $\tilde{Q}$ has no zeroth order
term since it has been taken away from $Q(f,f)$. Match the
zeroth order term on both sides of \eqref{eq:asym}, and we have
\begin{equation}
\tilde{Q}^1 = \frac{\partial f^0}{\partial t}
  + \bxi \cdot \nabla_{\bx} f^0 - Q(f^0, f^0).
\end{equation}
For any multi-index $\alpha$ with $|\alpha| > M$, multiplying
$\mathcal{H}_{\theta,\alpha}(\bv) \exp(|\bv|^2/2)$ on both sides
and then integrating the whole equality over $\bbR^D$ with
respect to $\bv$, since
\begin{equation}
\begin{split}
\frac{\partial f^0}{\partial t} \exp(|\bv|^2/2)
  & = \frac{\partial [f^0 \exp(|\bv|^2/2)]}{\partial t}
    - f^0 \cdot \frac{\partial [\exp(|\bv|^2/2)]}{\partial t} \\
& = \frac{\partial [f^0 \exp(|\bv|^2/2)]}{\partial t}
  + \left[ \frac{\bxi - \bu}{\sqrt{\theta}} \cdot
    \frac{\partial}{\partial t} \left(
      \frac{\bu}{\sqrt{\theta}}
    \right) \right] f^0 \exp(|\bv|^2/2),
\end{split}
\end{equation}
the orthogonality of Hermite polynomials leads to
\begin{equation} \label{eq:high-order}
\begin{split}
& \int_{\bbR^D} \mathcal{H}_{\theta,\alpha}(\bv)
  \tilde{Q}^1 \exp(|\bv|^2/2) \dd \bv \\
= & \int_{\bbR^D} \mathcal{H}_{\theta,\alpha}(\bv)
    G(\bxi f^0) \exp(|\bv|^2/2) \dd \bv
  - \int_{\bbR^D} \mathcal{H}_{\theta,\alpha}(\bv)
    Q(f^0,f^0) \exp(|\bv|^2/2) \dd \bv,
\end{split}
\end{equation}
where
\begin{equation} \label{eq:G}
G(\bxi f^0) = \left[
  \frac{1}{\sqrt{\theta}} \frac{\partial}{\partial t}
  \left( \frac{\bu}{\sqrt{\theta}} \right) + \nabla_{\bx}
\right] \cdot (\bxi f^0).
\end{equation}

Now the concrete form of the production term is required for
further calculation. Still, we adopt the simplest BGK model,
which gives
\begin{equation}
Q_{\mathrm{BGK}}(f^0) = -\nu (f^0 - f_M),
  \quad \tilde{Q}^1 = -\nu f^1.
\end{equation}
Thus the second term on the right hand side of \eqref{eq:high-order}
also vanishes, and \eqref{eq:high-order} becomes
\begin{equation}
\int_{\bbR^D} \mathcal{H}_{\theta,\alpha}(\bv)
  f^1 \exp(|\bv|^2/2) \dd \bv
= -\frac{1}{\nu} \int_{\bbR^D} \mathcal{H}_{\theta,\alpha}(\bv)
    G(\bxi f^0) \exp(|\bv|^2/2) \dd \bv.
\end{equation}
This shows that $f^1$ can be represented as
\begin{equation} \label{eq:f1_dis}
f^1 = -\frac{1}{\nu} \sum_{|\alpha| > M}
  [G(\bxi f^0)]_{\alpha} \mathcal{H}_{\theta,\alpha}(\bv),
\end{equation}
where the coefficient $[G(\bxi f^0)]_{\alpha}$ is the corresponding
coefficient of $G(\bxi f^0)$'s expansion in $F_{\infty}(\bu, \theta)$.

At last, we set $\varepsilon = 1$ and approximate $f$ by $f
\approx f^0 + f^1$. Since $f^1$ can be obtained from $f^0$
which is of finite dimension, when substituting such $f$
into the Boltzmann-BGK equation, a closed system can be
obtained without more truncations. For $D = 3$ and $M = 3,
4,5$, the R20, R35, R56 equations for the BGK model can be
obtained.

\subsection{The numerical method} \label{sec:reg_num}
In this subsection, we restrict our focus on the BGK model.
As in Section \ref{sec:Grad}, only $f^0$ is stored at each
time step. Note that for the BGK model, $f^0$ and $f^1$ are
decoupled in the production step, which can be implemented
exactly the same as that in Section \ref{sec:Grad}. Therefore,
we concentrate only on the convection step below.

Consider the original form of the numerical flux \eqref
{eq:HLL_flux}, and now $f_{\beta}^n$ is recognized as
$f_{\beta}^{n,0} + f_{\beta}^{n,1}$, where
\begin{equation}
f_{\beta}^{n,0} \in F_M(\bu_{\beta}^n, \theta_{\beta}^n),
  \quad \Pi_{\beta}^n f_{\beta}^{n,1} = 0.
\end{equation}
The latter equation follows from \eqref{eq:f1_dis}, \eqref
{eq:project} and the orthogonality of the Hermite polynomials.
Similarly, we have
\begin{equation}
\Pi_{\beta}^n f_{\beta \pm e_j}^{n,1} = 0,
  \quad \forall \beta \in \mathbb{Z}^n.
\end{equation}
Since only $f_{\beta}^{n+1*,0}$ is desired after the
convection step, the form of \eqref{eq:step1a} is
still adoptable. Therefore, we can mimic \eqref
{eq:num_flux} and write the numerical flux as
\begin{equation} \label{eq:new_flux}
\begin{split}
F_1 &= \frac{\lambda_j^R \Pi_{f_1} (\xi_j f_1^0) -
  \lambda_j^L \Pi_{f_2,f_1} \Pi_{f_2} (\xi_j f_2^0) +
  \lambda_j^L \lambda_j^R (\Pi_{f_2,f_1} f_2^0 - f_1^0)}
  {\lambda_j^R - \lambda_j^L} \\
& \qquad +\frac{\lambda_j^R \Pi_{f_1} (\xi_j f_1^1) -
  \lambda_j^L \Pi_{f_2,f_1} \Pi_{f_2} (\xi_j f_2^1)}
  {\lambda_j^R - \lambda_j^L}, \\
& := F_{11} + F_{12},
\end{split}
\end{equation}
where $f_1 = f_1^0 + f_1^1$, $f_2 = f_2^0 + f_2^1$, and $\Pi_{f_1}
f_1^1 = \Pi_{f_2} f_2^1 = 0$. As described in Section \ref{sec:Grad},
$F_{11}$ is just the flux for Grad's moment equations. As to $F_{12}$,
\eqref{eq:flux} and \eqref{eq:f1_dis} show
\begin{equation} \label{eq:F12}
\Pi_{f_i} (\xi_j f_i^1) = -\frac{1}{\nu}
  \sum_{|\alpha| = M} (\alpha_j + 1) [G(\bxi f_i^0)]_{\alpha + e_j}
    \mathcal{H}_{\theta_i,\alpha} (\bv_i),
  \quad i = 1,2.
\end{equation}
It is easy to see that in the expansions of $f_1$ and $f_2$, only the
coefficients with $|\alpha| = M + 1$ have effect on the numerical flux
$F_1$. We can also find that when $M \geqslant 3$, $F_{12}$ in
\eqref{eq:new_flux} has actually no contribution to the velocity $\bu$
and temperature $\theta$ for the next time step, since \eqref{eq:F12}
reveals that the Grad's expansion of $F_{12}$ contains only terms with
orders higher than or equal to $M$, but \eqref{eq:macro_vars} tells
that $\bu$ and $\theta$ are only relevant with the coefficients with
orders less than or equal to $2$. That is to say, for all $\beta$'s,
$\bu_{\beta}^{n+1}$ and $\theta_{\beta}^{n+1}$ can be solved using the
method introduced in Section \ref{sec:Grad}, without adding the
``regularizing part of numerical flux'' $F_{12}$. Thus, the time
derivative in \eqref{eq:G} can be explicitly approximated by
\begin{equation}
\left[ \frac{\partial}{\partial t}
  \left( \frac{\bu}{\sqrt{\theta}} \right)
\right]_{\beta}^n \approx
  \frac{1}{\Delta t^n} \left(
    \frac{\bu_{\beta}^{n+1}}{\sqrt{\theta_{\beta}^{n+1}}} -
    \frac{\bu_{\beta}^n}{\sqrt{\theta_{\beta}^n}}
  \right).
\end{equation}
Now, in order to approximate $[G(\bxi f^0)]_{\alpha}$, it is only
needed to approximate
\begin{equation}
[\nabla_{\bx} \cdot (\bxi f^0)]_{\alpha},
  \quad \alpha \in \bbN^D, \quad |\alpha| = M + 1.
\end{equation}

For a fixed point $\bx_0 \in \bbR^N$, we have
\begin{equation}
\begin{split}
[\nabla_{\bx} \cdot (\bxi f^0)]_{\alpha}(\bx_0)
  &= \sum_{j=1}^N \left[
    \frac{\partial}{\partial x_j} (\xi_j f^0)
  \right]_{\alpha} (\bx_0) \\
&= \sum_{j=1}^N  \left[
  C_{\theta_0,\alpha} \int_{\bbR^D}
  \frac{\partial}{\partial x_j} (\xi_j f^0)
  \mathcal{H}_{\theta_0,\alpha} \dd \bv
\right]_{\bx = \bx_0} \quad
  (\bv = \frac{\bxi - \bu_0}{\sqrt{\theta_0}})\\
&= \sum_{j=1}^N \left[ \frac{\partial}{\partial x_j}
  \left(
    C_{\theta_0,\alpha} \int_{\bbR^D}
    (\xi_j f^0) \mathcal{H}_{\theta_0,\alpha} \dd \bv
  \right) \right]_{\bx = \bx_0} \\
&= \sum_{j=1}^N \left[
  \frac{\partial}{\partial x_j} \left(
    \tilde{\Pi}_{\bu_0, \theta_0} (\xi_j f^0)
  \right)_{\alpha} \right]_{\bx = \bx_0},
\end{split}
\end{equation}
where $C_{\theta_0,\alpha}$ is defined by \eqref{eq:C}, $\bu_0$ and
$\theta_0$ are the mean velocity and temperature at point $\bx_0$, and
$\tilde{\Pi}_{\bu_0, \theta_0}$ is the projection operator to the
space $F_{M+1}(\bu_0, \theta_0)$. Now consider the discrete
circumstance. For each $\beta \in \mathbb{Z}^N$,
$\tilde{\Pi}_{\bu_{\beta}^n, \theta_{\beta}^n} (\xi_j f_{\beta \pm
e_j}^{n,0})$ can be obtained according to \eqref{eq:flux} and
\eqref{eq:finite_ode}.  Without confusion, the superscript ``$n$''
will be omitted below, and $\tilde{\Pi}_{\bu_{\beta},\theta_{\beta}}$
is simplified as $\tilde{\Pi}_{\beta}$. Mimicing the method in
\cite{Torrilhon2006}, suppose
\begin{gather}
\label{eq:grad1}
d_1 = \frac{1}{\Delta x_j} \left[ \left(
  \tilde{\Pi}_{\beta} (\xi_j f_{\beta+e_j}^0)
\right)_{\alpha} - \left(
  \xi_j f_{\beta}^0
\right)_{\alpha} \right], \\
\label{eq:grad2}
d_2 = \frac{1}{\Delta x_j} \left[ \left(
  \xi_j f_{\beta}^0
\right)_{\alpha} - \left(
  \tilde{\Pi}_{\beta} (\xi_j f_{\beta-e_j}^0)
\right)_{\alpha} \right],
\end{gather}
and then we can reconstruct the spatial partial derivative
by a central difference in the smooth case:
\begin{equation}
\left[ \delta_{x_j} \left(
  \tilde{\Pi}_{\beta} (\xi_j f^0)
\right)_{\alpha} \right]_{\beta} = \frac{d_1 + d_2}{2},
\end{equation}
or the van Leer reconstruction in the discontinuous case:
\begin{equation}
\left[ \delta_{x_j} \left(
  \tilde{\Pi}_{\beta} (\xi_j f^0)
\right)_{\alpha} \right]_{\beta}^{\textrm{van Leer}} =
\frac{|d_1| d_2 + |d_2| d_1}{|d_1| + |d_2|}.
\end{equation}
Thus we have
\begin{equation} \label{eq:f^1}
\begin{split}
\tilde{\Pi}_{\beta} f_{\beta}^1 & = -\frac{1}{\nu} 
  \sum_{|\alpha| = M + 1} \Bigg\{
    \sum_{j=1}^N \left[ \delta_{x_j} \left(
      \tilde{\Pi}_{\beta} (\xi_j f^0)
    \right)_{\alpha} \right]_{\beta} \\
& \qquad + \sum_{j=1}^D \left[ \left(
    \sqrt{\frac{\theta_{\beta}^n}{\theta_{\beta}^{n+1}}}
      u_{\beta,j}^{n+1} - u_{\beta,j}^n
  \right) \frac{f_{\beta, \alpha - e_j}^0}{\Delta t} \right]
\Bigg\} \mathcal{H}_{\theta_{\beta},\alpha},
\end{split}
\end{equation}
where \eqref{eq:flux} has been incorporated into the expression
of $G(\bxi f^0)$.

From \eqref{eq:F12}, we can conclude that
\begin{equation} \label{eq:1ord_term}
\Pi_{\beta} (\xi_j f_{\beta}^1) =
  \Pi_{\beta} (\xi_j \tilde{\Pi}_{\beta} f_{\beta}^1),
  \quad j = 1,\cdots,N,
\end{equation}
which appears twice on the right hand side of \eqref
{eq:new_flux}. This expression can be computed cheaply
since $\tilde{\Pi}_ {\beta} f_{\beta}^1$ has only
$O(M^{D-1})$ non-zero coefficients. The last term
in \eqref{eq:new_flux} requires us to compute
\begin{equation}
\Pi_{\beta} \Pi_{\beta \pm e_j}
  (\xi_j f_{\beta \pm e_j}^1),
\end{equation}
i.e. to apply another projection to the result of
\eqref{eq:1ord_term}. Note that $F_{\alpha}$ does not
appear on the right hand side in \eqref{eq:finite_ode},
and for all $\alpha \in \bbN^D$, $|\alpha| \neq M$,
the coefficients of $\mathcal{H}_{\theta_{\beta},
\alpha}$ in the expansion of \eqref{eq:1ord_term} are
zero, so all coefficients do not change after projection, 
saying if
\begin{equation}
\Pi_{\beta \pm e_j} (\xi_j f_{\beta \pm e_j}^1) =
  \sum_{|\alpha| = M} C_{\pm,\alpha}
  \mathcal{H}_{\theta_{\beta \pm e_j}, \alpha}
    (\bv_{\beta \pm e_j}),
\end{equation}
then
\begin{equation} \label{eq:invariant_coef}
\Pi_{\beta} \Pi_{\beta \pm e_j} (\xi_j f_{\beta \pm e_j}^1)
  = \sum_{|\alpha| = M} C_{\pm,\alpha}
  \mathcal{H}_{\theta_{\beta}, \alpha} (\bv_{\beta}).
\end{equation}
This is the reason why a coefficient $\theta^{-\frac{\theta_d+1}{2}}$
is multiplied in the definition of basis functions \eqref{eq:H}.
Until now, the calculation of numerical fluxes for the regularized
moment equations is thoroughly clarified.

\begin{remark} \label{rem:combine}
In \eqref{eq:grad1} and \eqref{eq:grad2}, $\tilde{\Pi}_{\beta}
(\xi_j f_{\beta \pm e_j}^0)$ needs to be computed, while in the
convection term of Grad's moment equations, or rather, $F_{11}$ in
\eqref{eq:new_flux}, $\Pi_{\beta} (\xi_j f_{\beta \pm e_j}^0)$ also
needs to be computed. Since $\Pi_{\beta} = \Pi_{\beta}
\tilde{\Pi}_{\beta}$, and the second $\Pi_{\beta}$ is trivial, these
two projections can be combined into one.
\end{remark}

\begin{remark}
As is known, the regularized moment equations contain second order
derivative terms, so the CFL condition for the method above is
\begin{equation} \label{eq:CFL}
\Delta t \sum_{j=1}^N \frac{\lambda_{j,\max}}{\Delta x_j}
  \left(1 + \dfrac{4\nu_{\max}}{\Delta x_j}\right) < 1,
\end{equation}
where $\lambda_{j,\max}$ is the maximum of $|\lambda_j^L|$ and
$|\lambda_j^R|$ on all cells, and $\nu_{\max}$ is the maximal
collision frequency. Note that for this regularized model, when
calculating $\lambda_j^L$ and $\lambda_j^R$ using \eqref{eq:lambda_L}
and \eqref{eq:lambda_R}, the roots of $\He_{M+1}(x)$ should be
replaced by the roots of $\He _{M+2}(x)$ since we use the $(M + 1)$-th
order moments of $f^1$. Eq. \eqref{eq:CFL} leads to a relatively small
time step length. The time step length can be enlarged following the
methods in \cite{Xu,Torrilhon2006,Jeltsch}, which is not yet
implemented in our program.
\end{remark}

\begin{remark}
In our implementation, the classical fourth-order Runge-Kutta method
is used to solve \eqref{eq:finite_ode}. Since the solution of
regularized moment equations is generally smooth, most of the time,
only one Runge-Kutta step is able to provide enough accuracy for such
local projections. In the case of sharp initial values, some more steps
are performed. However, such a situation only appears at the very
beginning of the calculation.
\end{remark}

\subsection{Outline of the algorithm}
As a summarization, our numerical method for the regularized
moment equations is outlined as below:
\begin{enumerate}
\setlength{\itemsep}{0pt}
\setlength{\parskip}{0pt}
\item Let $n = 0$ and set the initial state $f_{\beta}%
  ^{n,0}$ for all $\beta$'s.
\item \label{item:CFL}
  Determine the time step length according to the CFL
  condition \eqref{eq:CFL}.
\item Apply the convection step in Section \ref{sec:Grad}.
\item Obtain $\tilde{\Pi}_{\beta} f_{\beta}^{n,1}$ as 
  in Section \ref{sec:reg_num}.
\item Add the ``regularizing part of numerical flux'' to
  $f_{\beta}^{n+1,0}$.
\item Apply the production step at the end of
  Section \ref{sec:outline}.
\item Let $n \leftarrow n+1$, and return to step \ref{item:CFL}.
\end{enumerate}

\subsection{Boundary conditions}
Currently, boundary conditions are not available in our numerical
scheme. The boundary condition is always a delicate issue in the
deduction of the macroscopic equations. Generally speaking, the
kinetic boundary condition introduced by Maxwell \cite{Maxwell} is
expected to be added to the algorithm. As in \cite{Grad} and \cite
{Torrilhon2008}, half-space integration needs to be performed during
the construction of the distribution function in the ghost cells. This
is possible to obtain for the discrete distributions $f^0 + \tilde
{\Pi}_{\bu,\theta} f^1$, thanks to the recursion relation of the
Hermite polynomials, but such integration requires a subroutine with a
time complexity of $O(M^{2D})$, which results in much more
computational time on the cells next to the wall. The details are
still in preparation.


%% file: article_num_ex.tex
\section{Numerical examples}
In this section, 1D and 2D numerical examples of our method for the
regularized moment equations are presented. In all these tests,
the global Knudsen number is denoted as $\Kn$, and the collision
frequency $\nu$ (see \eqref{eq:BGK}) is substituted by $\rho(t,\bx) /
\Kn$. The CFL number is always $0.8$. We use the POSIX multi-threaded
technique in our simulation, and at most $8$ CPU cores are used.

\subsection{One dimensional case}
Two 1D examples are studied as follows. Since the boundary
conditions are currently out of our consideration, only free or
periodic boundary conditions are used in the following examples.
In this section, some numerical solutions of the Boltzmann-BGK
equation are provided, which are obtained according to the algorithm
described in \cite{Mieussens}.

\subsubsection{Shock tube test}
The shock tube test has been investigated in many works due to its
fundamental role in characterizing the hyperbolicity of equations.
For Grad's 13-moment system, it is shown in \cite{Torrilhon2000} that
unphysical subshocks can be found. However, in \cite{Torrilhon2006},
the numerical result of the Riemann problem shows that R13 equations
are able to capture these waves correctly. Here we first repeat this
test in \cite{Torrilhon2006} to obtain similar results.

As in \cite{Torrilhon2006}, the initial conditions are
\begin{equation}
\rho(0, x) = \begin{cases}
  7.0, & x < 0, \\ 1.0, & x > 0,
\end{cases} \qquad
p(0, x) = \begin{cases}
  7.0, & x < 0, \\ 1.0, & x > 0,
\end{cases}
\end{equation}
where the pressure $p$ equals to $\rho \theta$. The velocity is zero
everywhere, and the fluid is in equilibrium everywhere.  The
computational domain is $[-1,1]$. In order to make a comparison, we
set $\Kn = 0.02$ and compute this model until $t = 0.3$ as in
\cite{Torrilhon2006}. The result of R20 equations ($M = 3$) is given
in Figure \ref{fig:ShockTube_Kn=0.02}. Compared with the result in
\cite{Torrilhon2006}, the plot of density agrees with that in
\cite{Torrilhon2006} very well. The plot of heat flux has a similar
shape with that in \cite{Torrilhon2006}, but they differ in magnitude.
This is due to the different models in the collision operator. The BGK
model fails to predict the correct Prandtl number, which results in
the incorrectness of heat flux.

\begin{figure}[!ht]
\centering
\includegraphics[scale=.45]{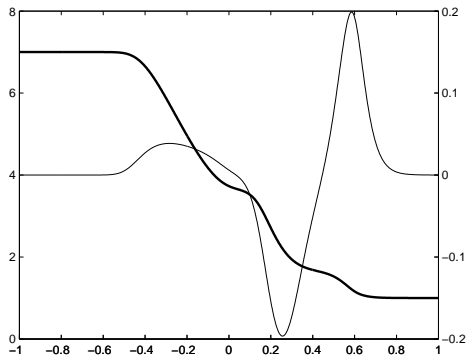}
\caption{R20 results for the shock tube test with $\Kn = 0.02$.  The
thick line with the left $y$-axis is the plot of density, while the
thin line with the right $y$-axis is the plot of heat flux. 1000 grids
are used for calculation.}
\label{fig:ShockTube_Kn=0.02}
\end{figure}

Since $Kn$ is small, almost the same results are produced by $M > 3$
and no subshocks are found. This illustrates the capability of our
method in capturing physical waves.  Note that the peak and valley of
the heat flux lies in the ``discontinuities'' of the density, which
indicates the non-equilibrium. The negative part ($x < 0$) has smaller
heat flux due to the larger density or collision term.

Now we set $\Kn = 0.5$ to investigate the numerical behavior of our
method in the case of greater Knudsen number. The curves of density
and temperature for $M = 3$ to $8$, as well as the numerical solutions
of the Boltzmann-BGK equation, are plotted in Figure
\ref{fig:ShockTube_Kn=0.5}. With this Knudsen number, the
hyperbolicity of the regularized moment equations clearly turns to be
dominant. The solutions in Figure \ref{fig:ShockTube_Kn=0.5} exhibit a
similar behavior with those in \cite{Au}. In the region of $x < 0$,
all results are similar because of the high density behind the initial
shock. In front of the initial shock, both the density and temperature
are converging to the BGK solution, although the convergence rate is
much slower.

\begin{figure}[!ht]
\centering
\subfigure[R20, $M=3$]{
  \includegraphics[scale=.4]{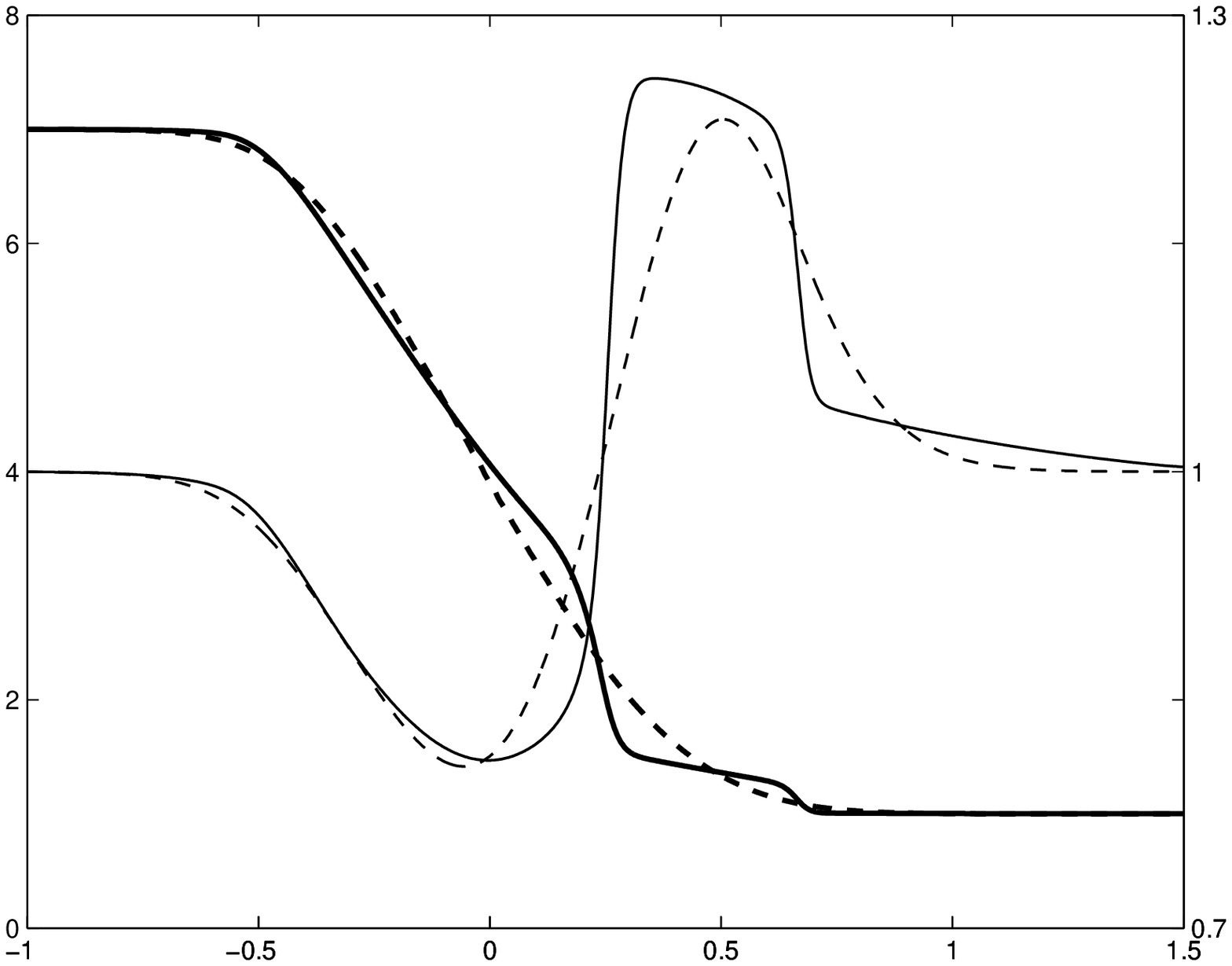}
}
\subfigure[R35, $M=4$]{
  \includegraphics[scale=.4]{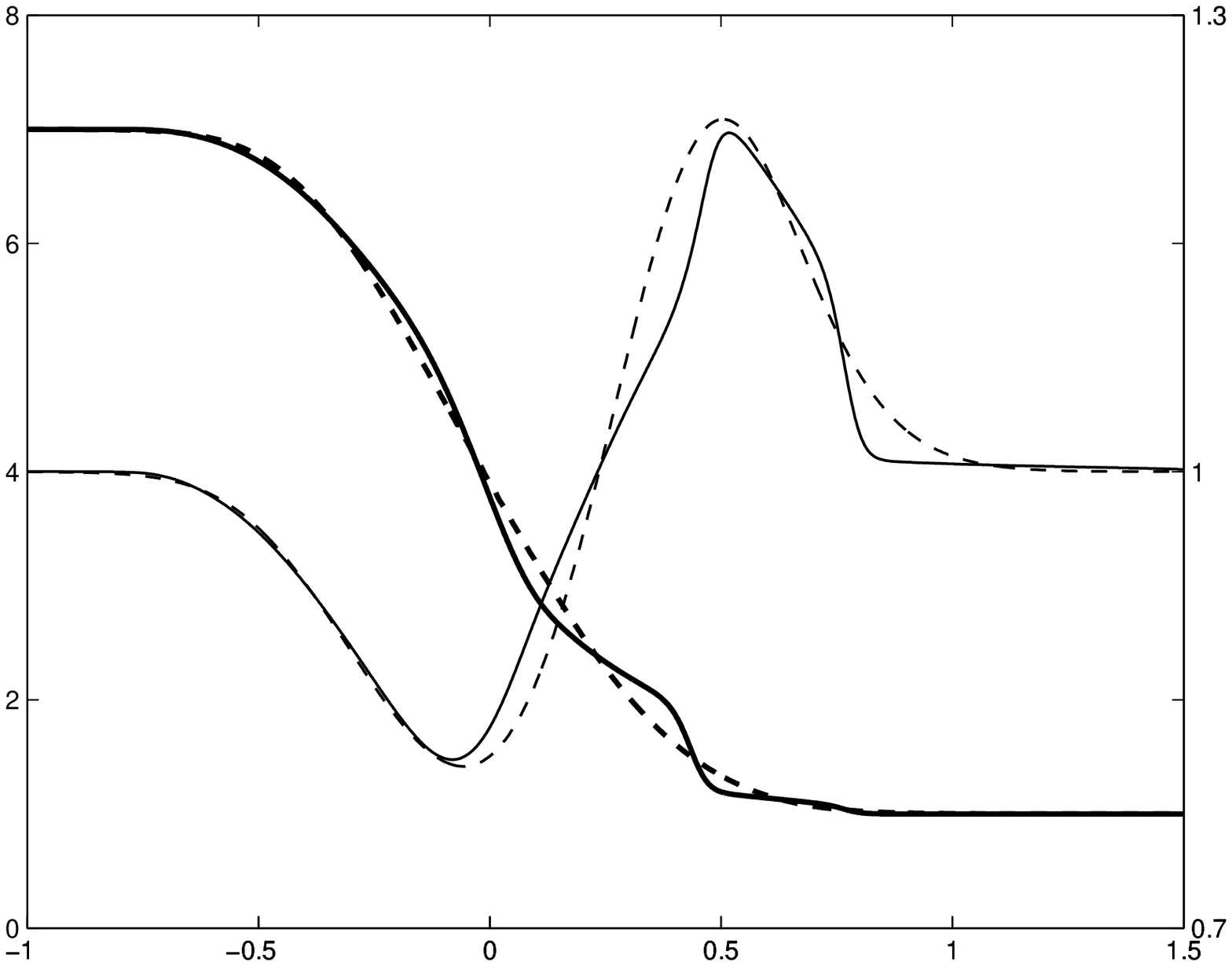}
}
\subfigure[R56, $M=5$]{
  \includegraphics[scale=.4]{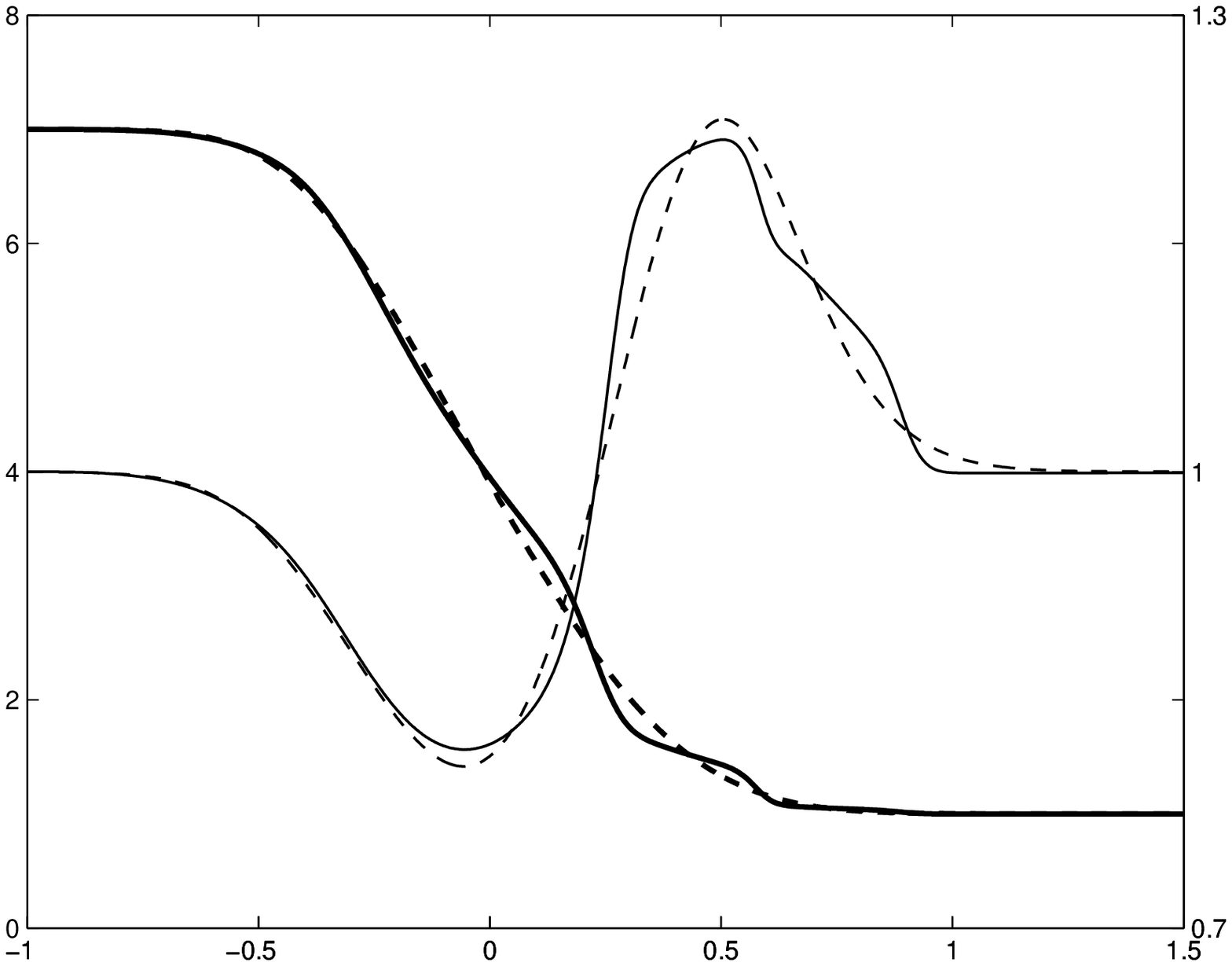}
}
\subfigure[R84, $M=6$]{
  \includegraphics[scale=.4]{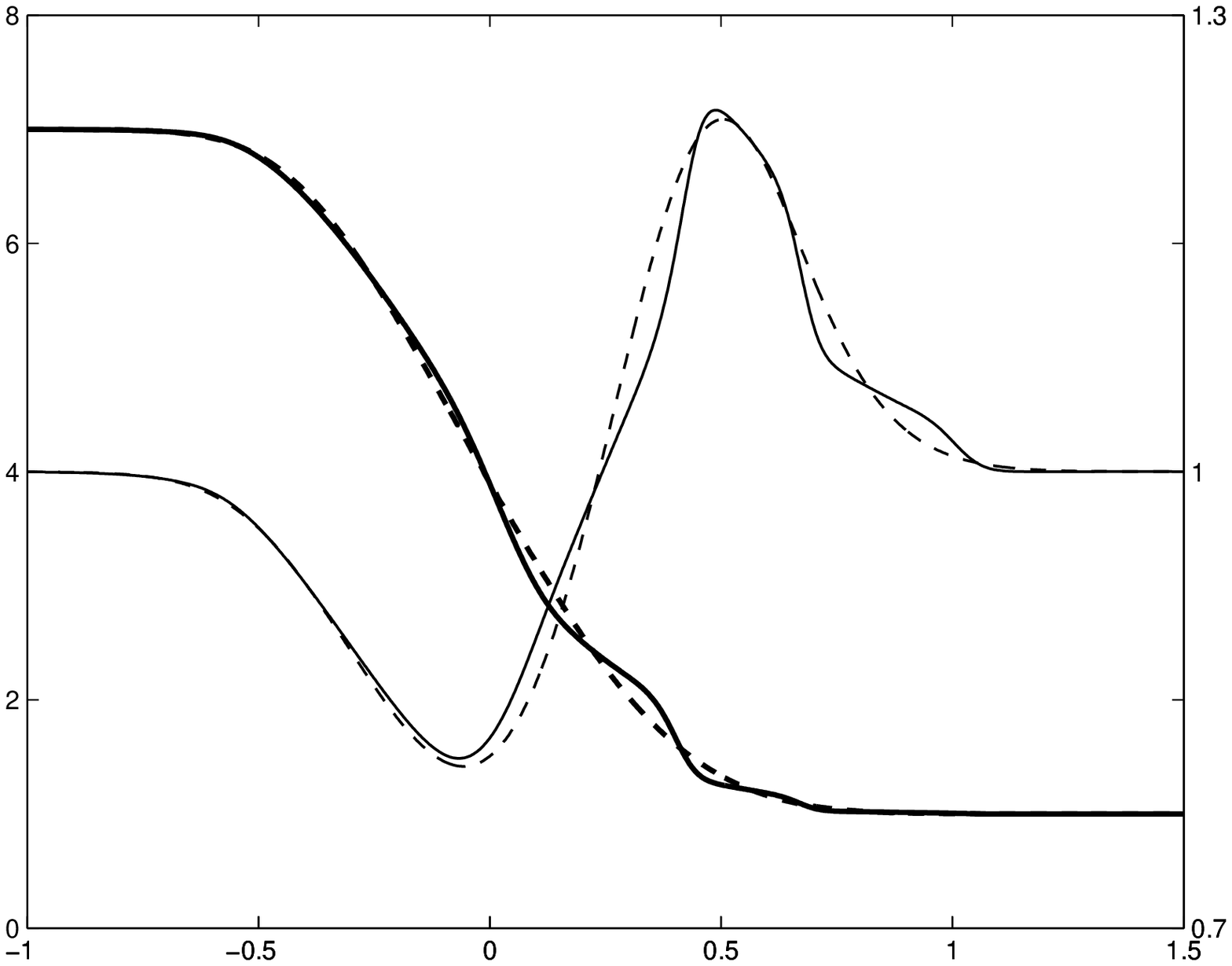}
}
\subfigure[R120, $M=7$]{
  \includegraphics[scale=.4]{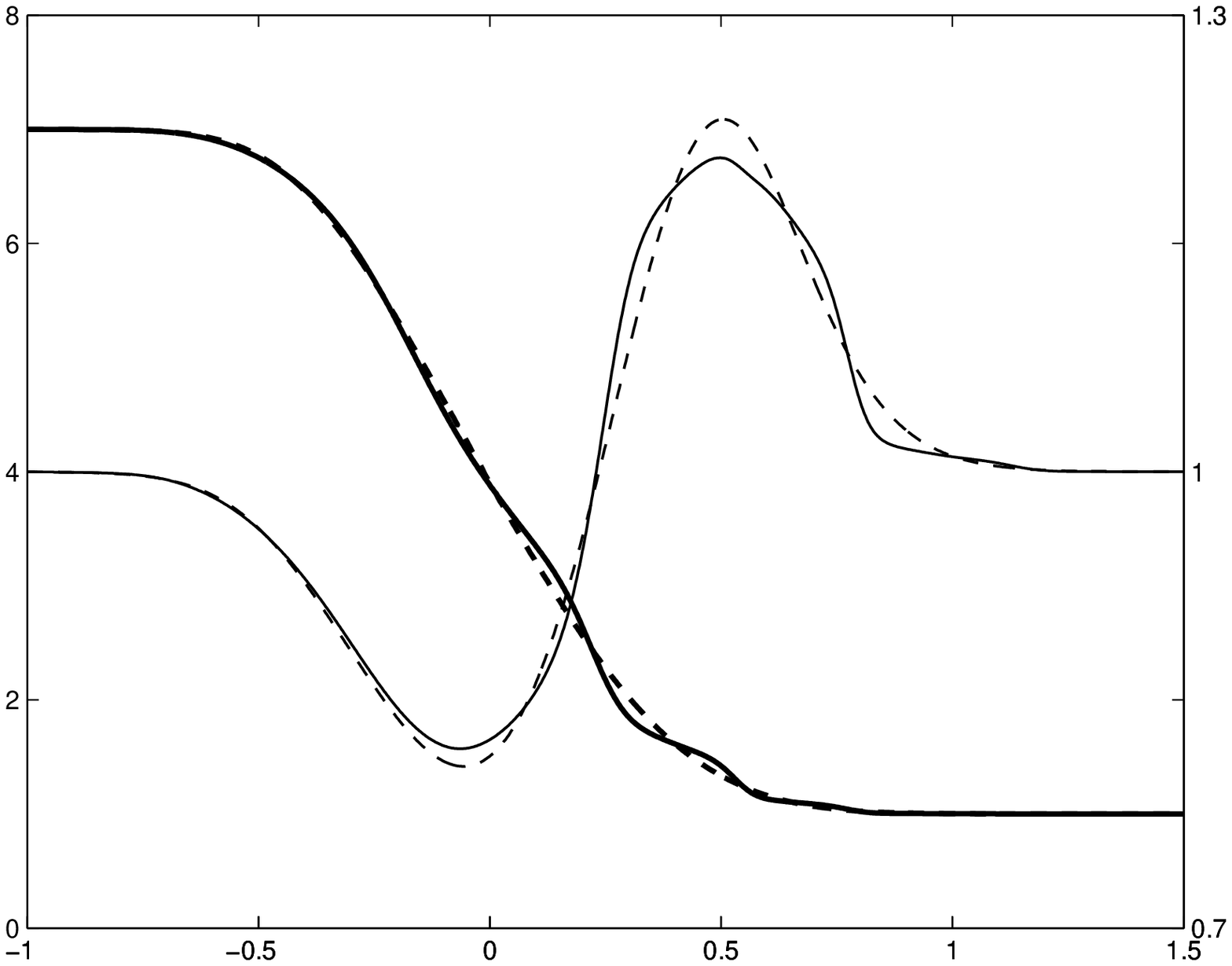}
}
\subfigure[R165, $M=8$]{
  \includegraphics[scale=.4]{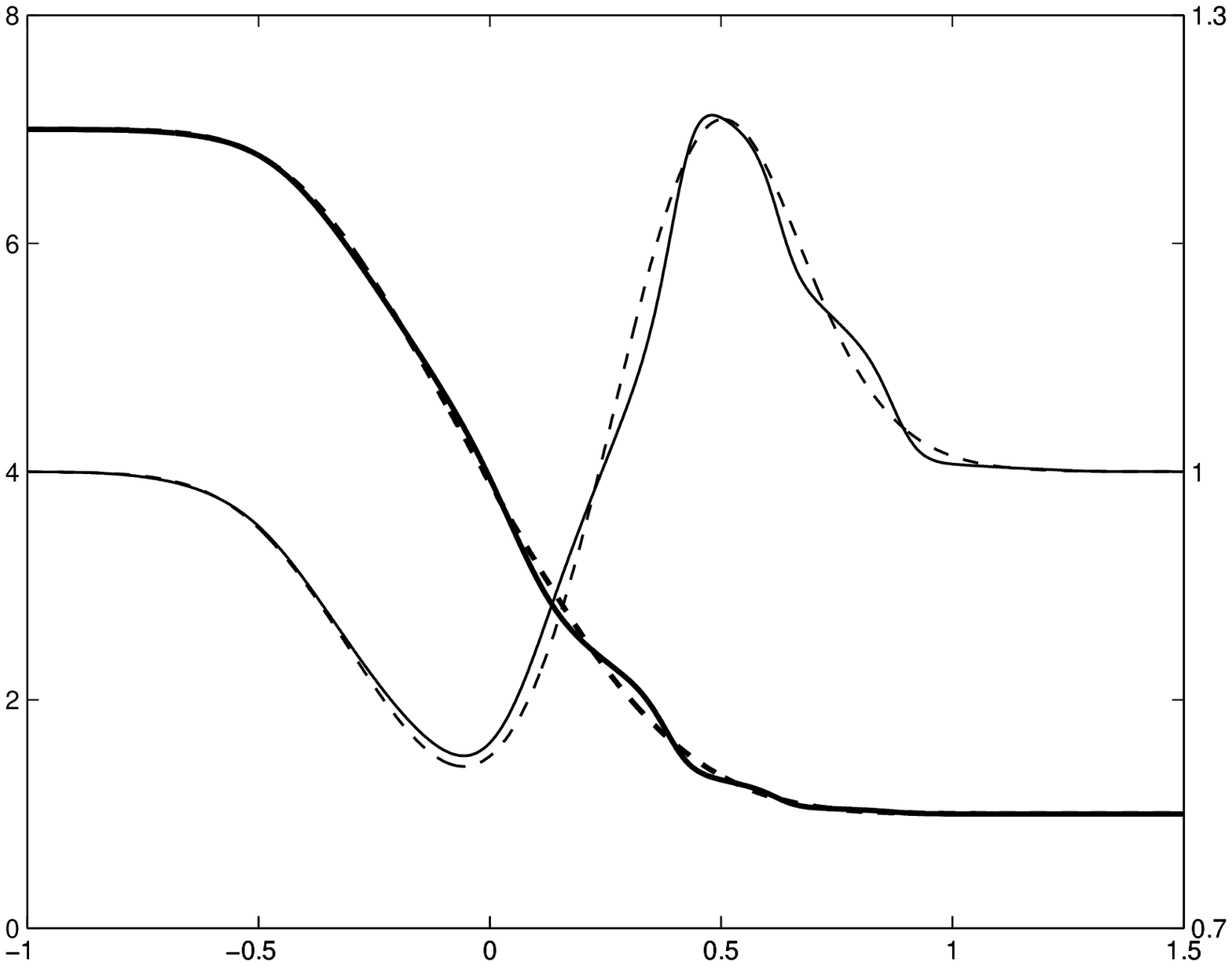}
}
\caption{Results for the shock tube test with $\Kn = 0.5$.
The thick line with the left $y$-axis is the plot of density,
while the thin line with the right $y$-axis is the plot of
temperature. The dashed lines are the numerical solution of
the BGK equation. 1250 grids are used for calculation.}
\label{fig:ShockTube_Kn=0.5}
\end{figure}

\subsubsection{A test with smooth initial values}
This example is again from \cite{Torrilhon2006}. The initial
conditions are
\begin{equation} \label{eq:periodic}
\rho(0, x) = 2 + \frac{1}{2} \cos(\pi x),
\quad \bu(0, x) = \left(1 + \frac{1}{2} \sin (\pi x),
  \frac{1}{2} \sin (\pi x), 0\right)^T,
\end{equation}
and the fluid is in equilibrium everywhere with $p(0, x) = 1$.
Periodic boundary condition is used and the computational domain is
the interval $[-1,1]$. In order to validate our method, we use $\Kn =
0.01, 0.1, 0.5$ in our numerical computation exactly as in
\cite{Torrilhon2006}. The end time is $t = 0.4$.

The results for different Knudsen numbers and different moment
equations are plotted in Figure \ref{fig:periodic}. All tests are
computed using $1000$ grids. The results in the first column are
almost identical, which indicates the correct behavior of our
method in the dense limit. The R20 equations, which should be the
closest to the R13 system, produce similar results as those of R13
reported in \cite{Torrilhon2006}. The temperature plots in the first
row can be used to make comparison.

The new results are presented in the second and third columns, where
the numerical solutions for high-order moment equations are listed.
For $\Kn = 0.1$, the R20 result shows an incorrect profile of density.
With increasing $M$, both the density and the temperature tend to
converge. For $\Kn = 0.5$, the R20 equations provide completely wrong
structures, although smooth initial values are used. Results for
even larger moment systems are plotted in Figure \ref
{fig:periodic_Kn=0.5}. In this plot, the satisfying temperature plot
is obtained when $M = 7$, but the density plots behave similarly as
those in Figure \ref{fig:ShockTube_Kn=0.5}. The curves for odd and
even order of moments hold different profiles, and they are toddling
close to each other gradually. With $M$ been increased up to $11$,
the density curve eventually exhibits a satisfying convergence. The
phenomenon illustrates the necessity of large moment systems in the
microcase.

\begin{figure}[!ht]
\centering
\begin{tabular}{c@{}ccc}
& $\Kn = 0.01$ & $\Kn = 0.1$ & $\Kn = 0.5$ \\
\raisebox{43pt}{\small R20}
  & \includegraphics[scale=.3,bb=111 247 509 544,clip]{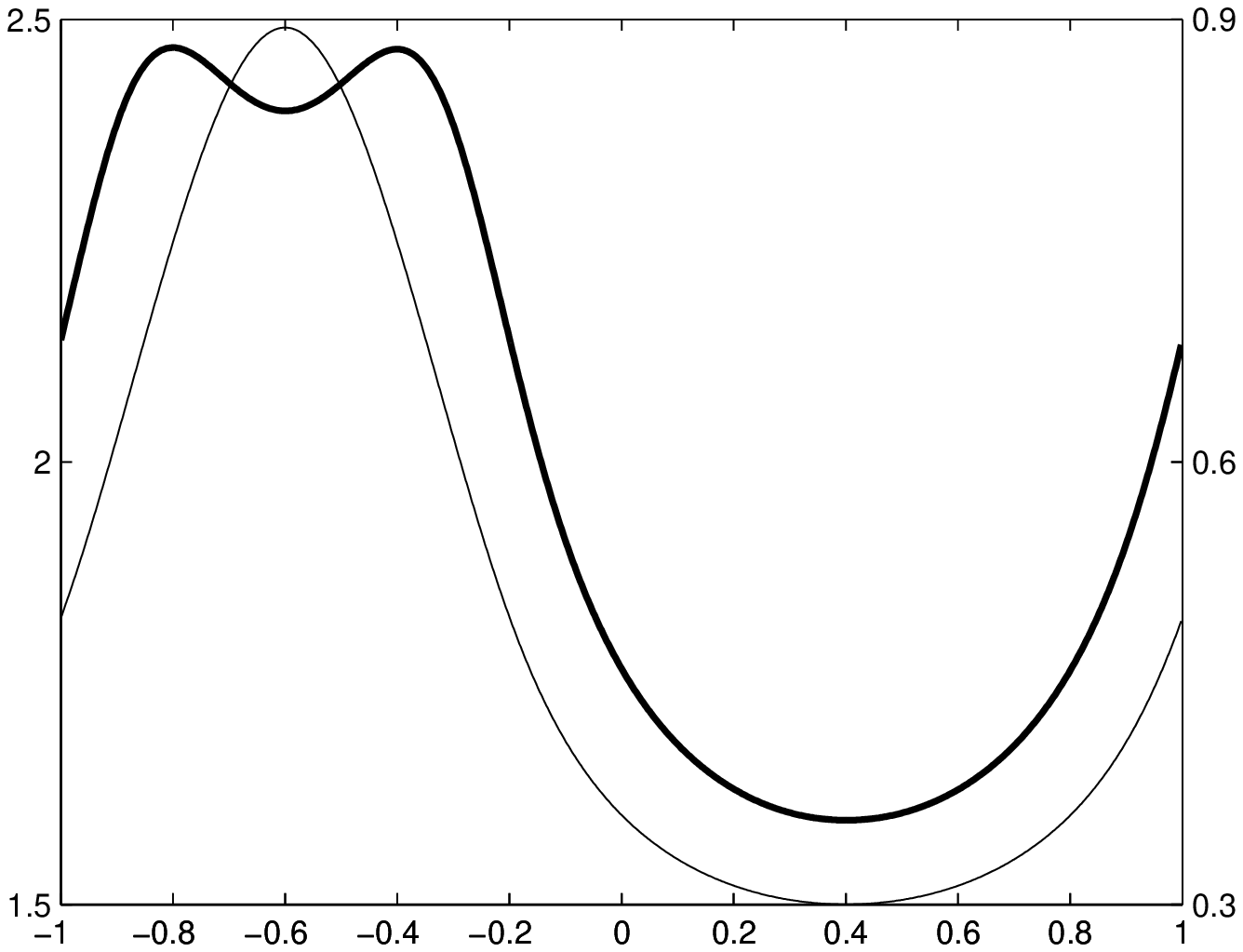}
  & \includegraphics[scale=.3,bb=111 247 509 544,clip]{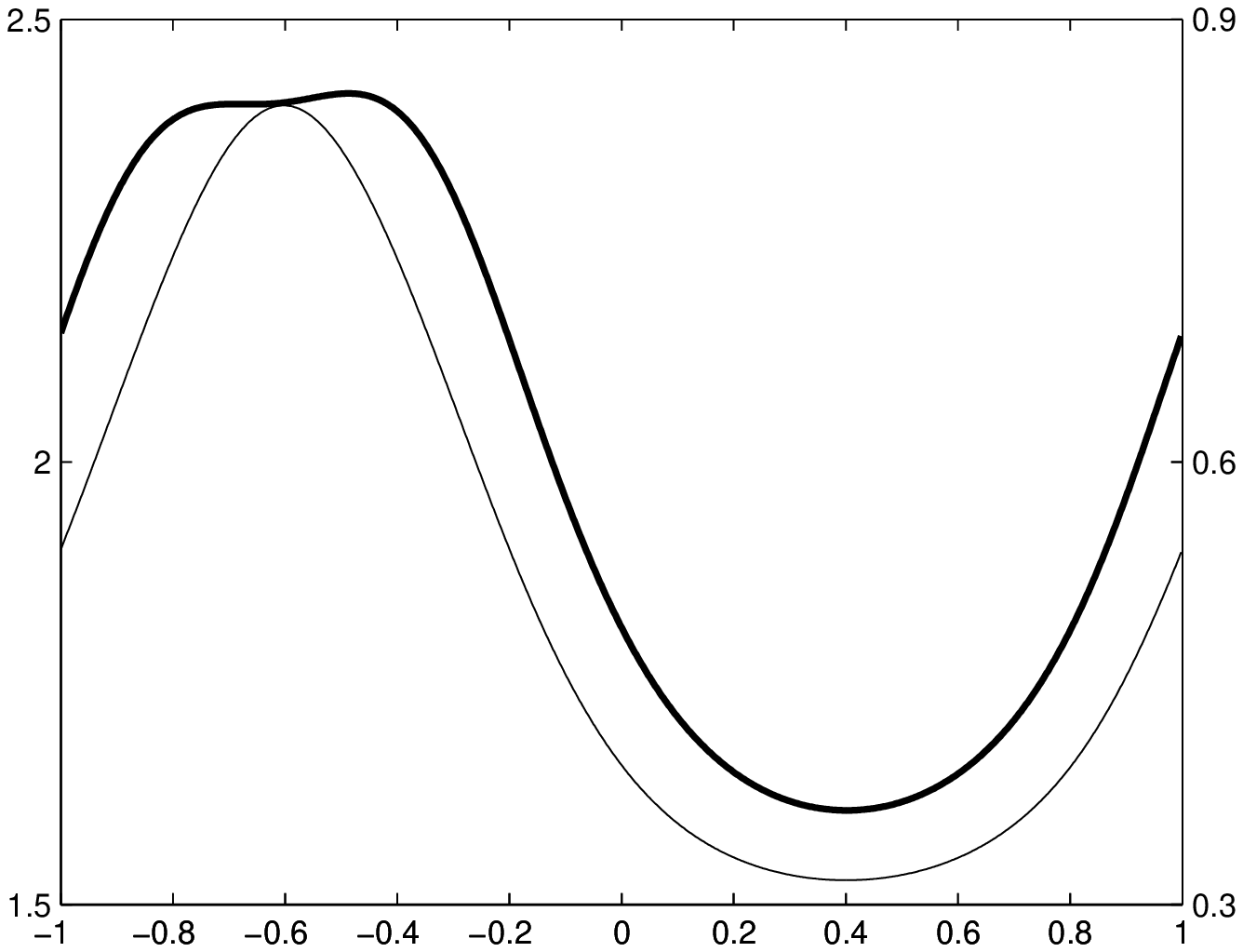}
  & \includegraphics[scale=.3,bb=111 247 509 544,clip]{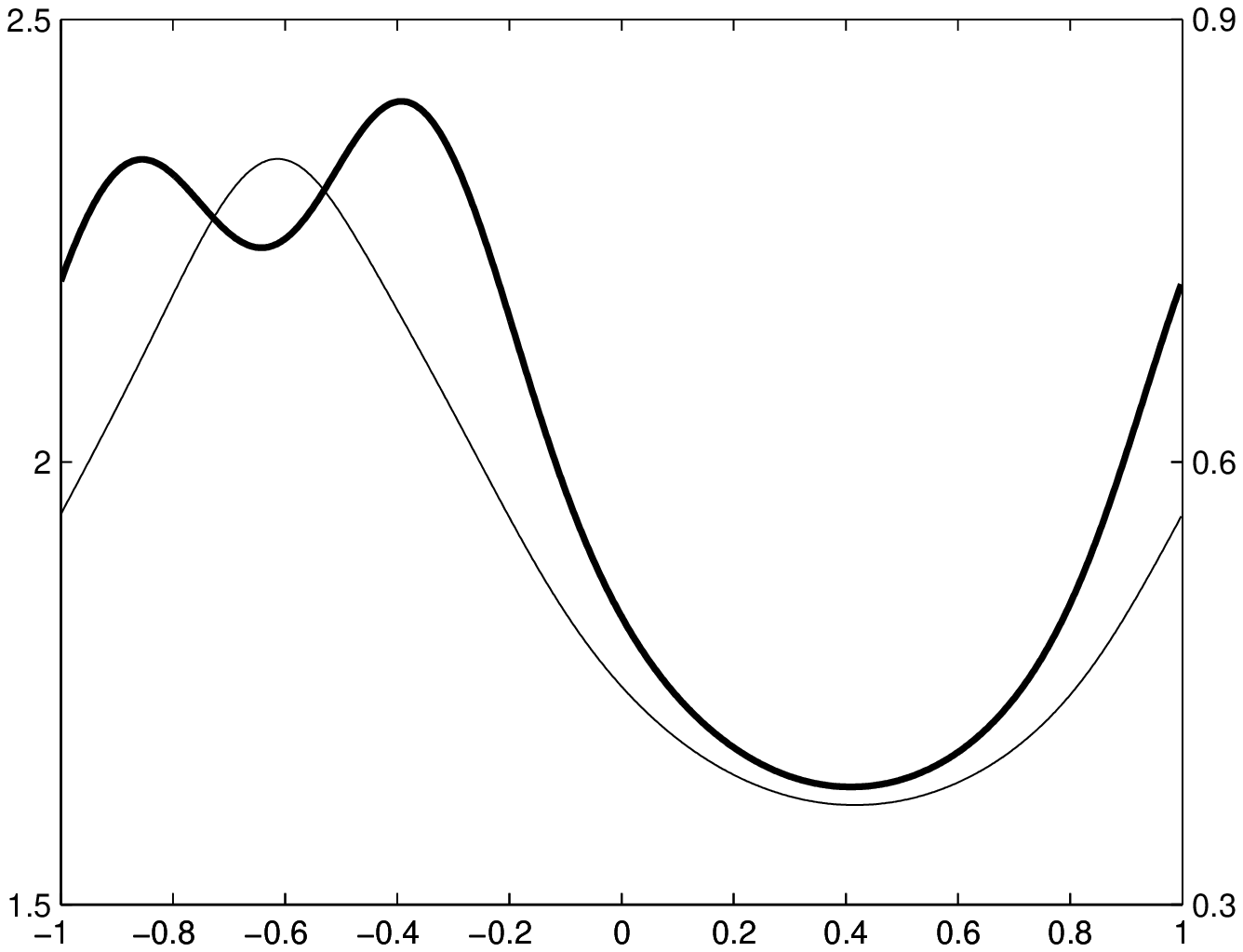} \\
\raisebox{43pt}{\small R35}
  & \includegraphics[scale=.3,bb=111 247 509 544,clip]{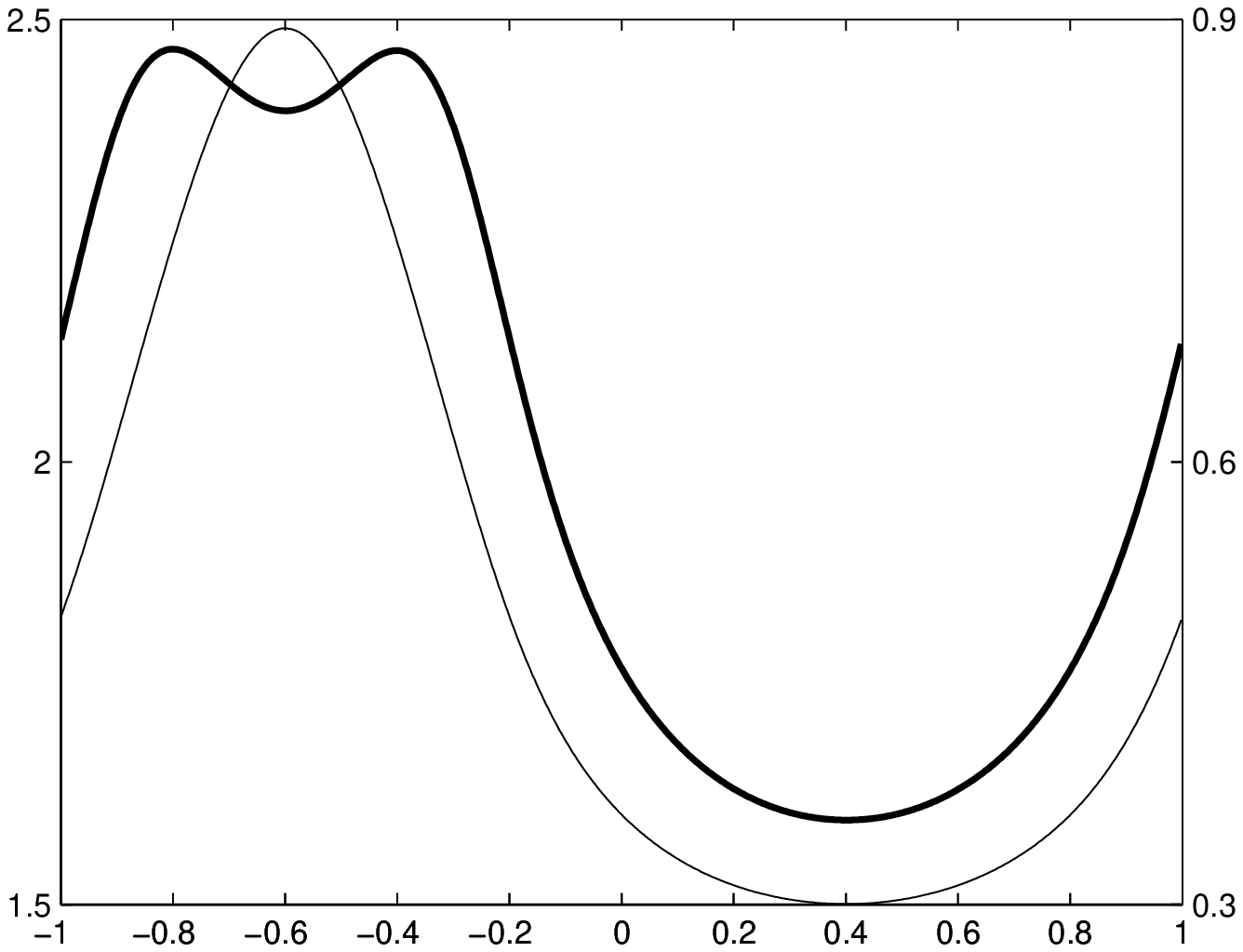}
  & \includegraphics[scale=.3,bb=111 247 509 544,clip]{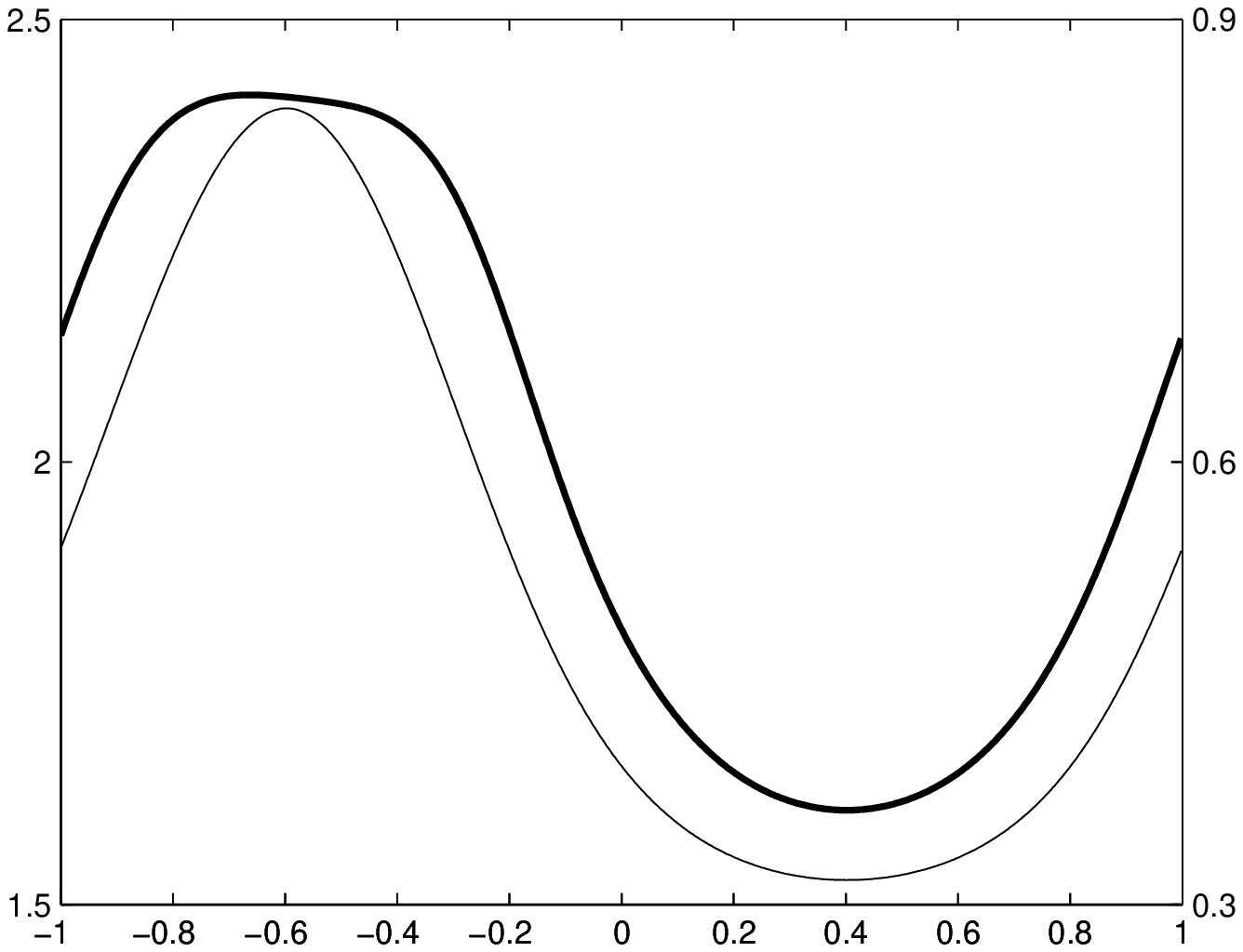}
  & \includegraphics[scale=.3,bb=111 247 509 544,clip]{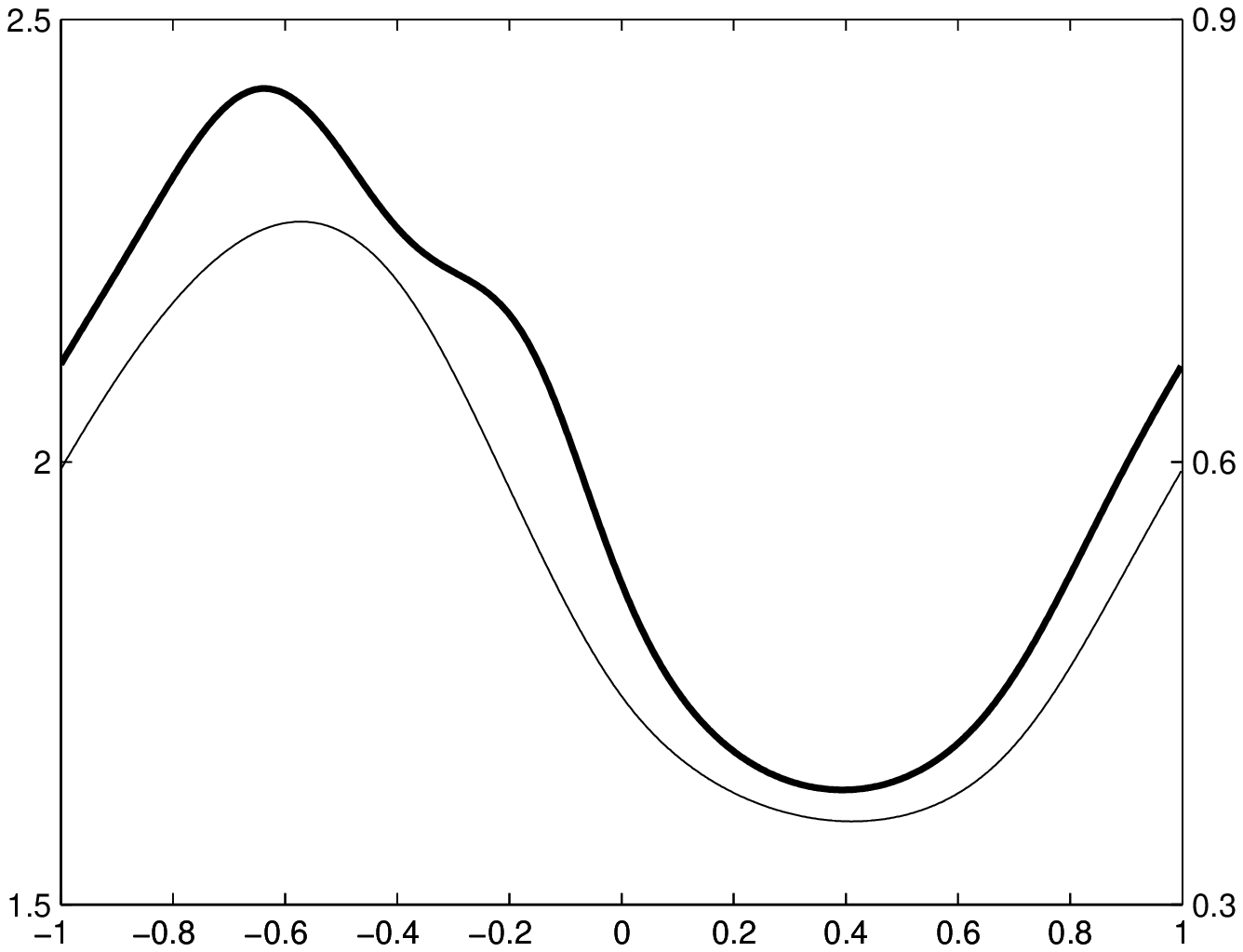} \\
\raisebox{43pt}{\small R56}
  & \includegraphics[scale=.3,bb=111 247 509 544,clip]{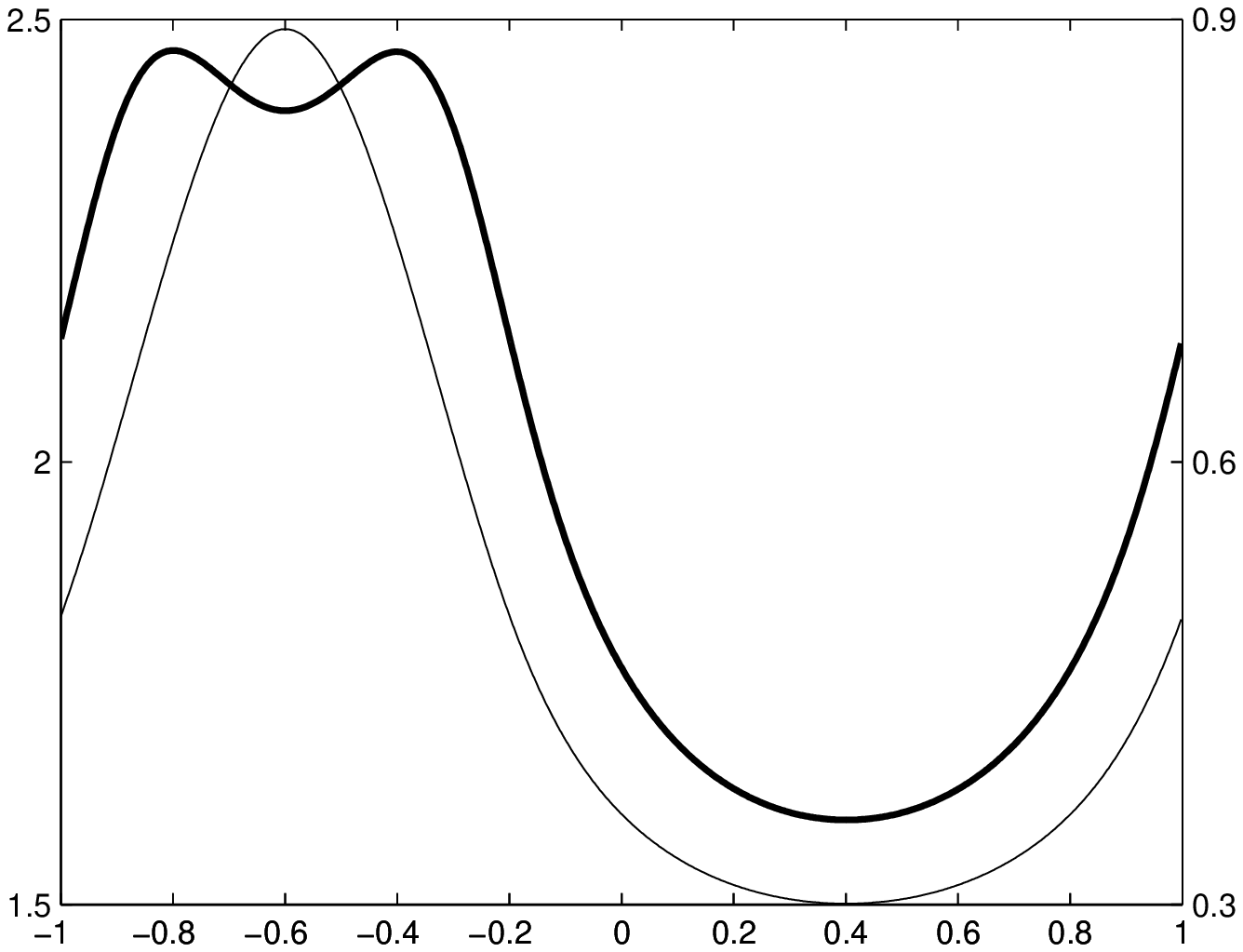}
  & \includegraphics[scale=.3,bb=111 247 509 544,clip]{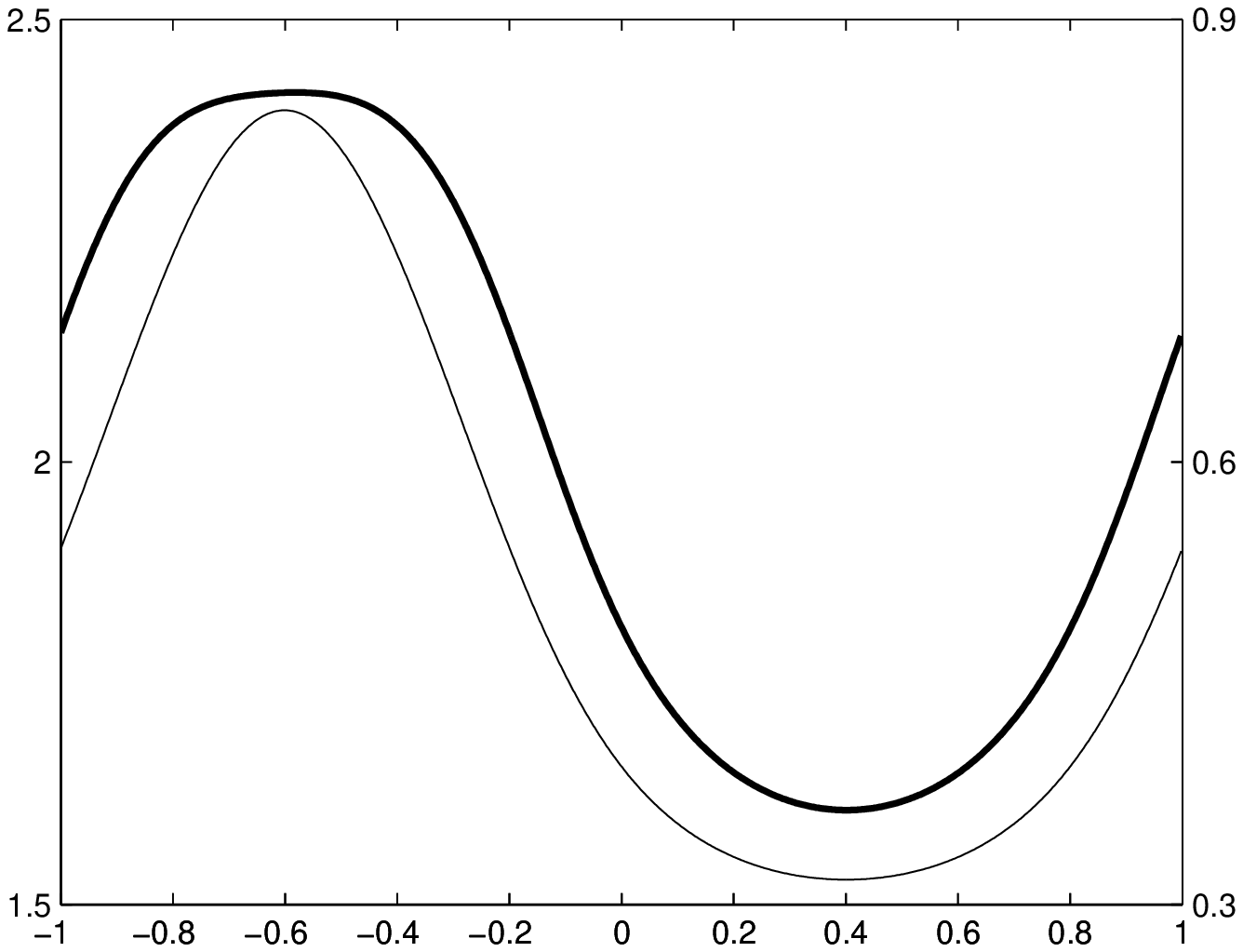}
  & \includegraphics[scale=.3,bb=111 247 509 544,clip]{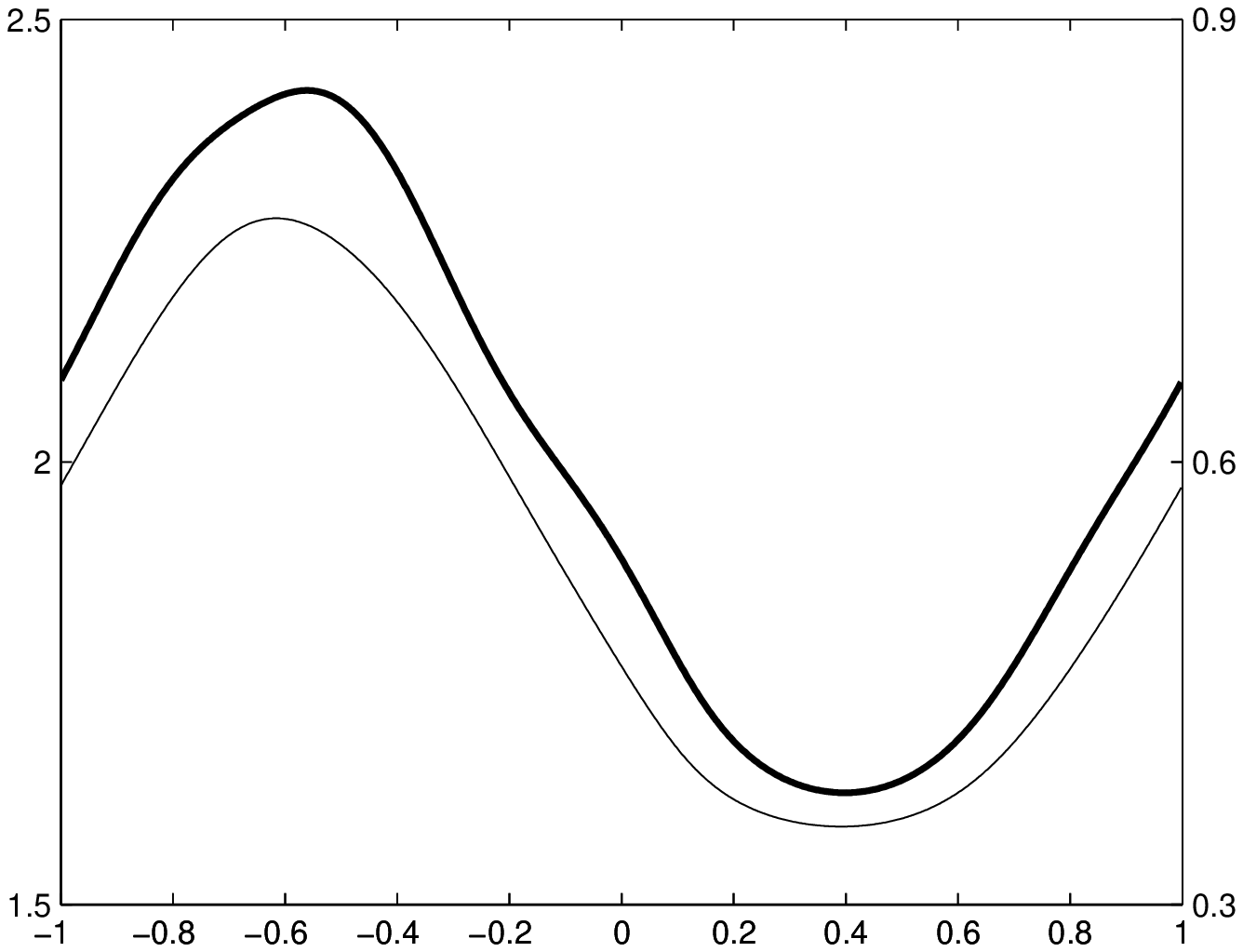} \\
\raisebox{43pt}{\small R84}
  & \includegraphics[scale=.3,bb=111 247 509 544,clip]{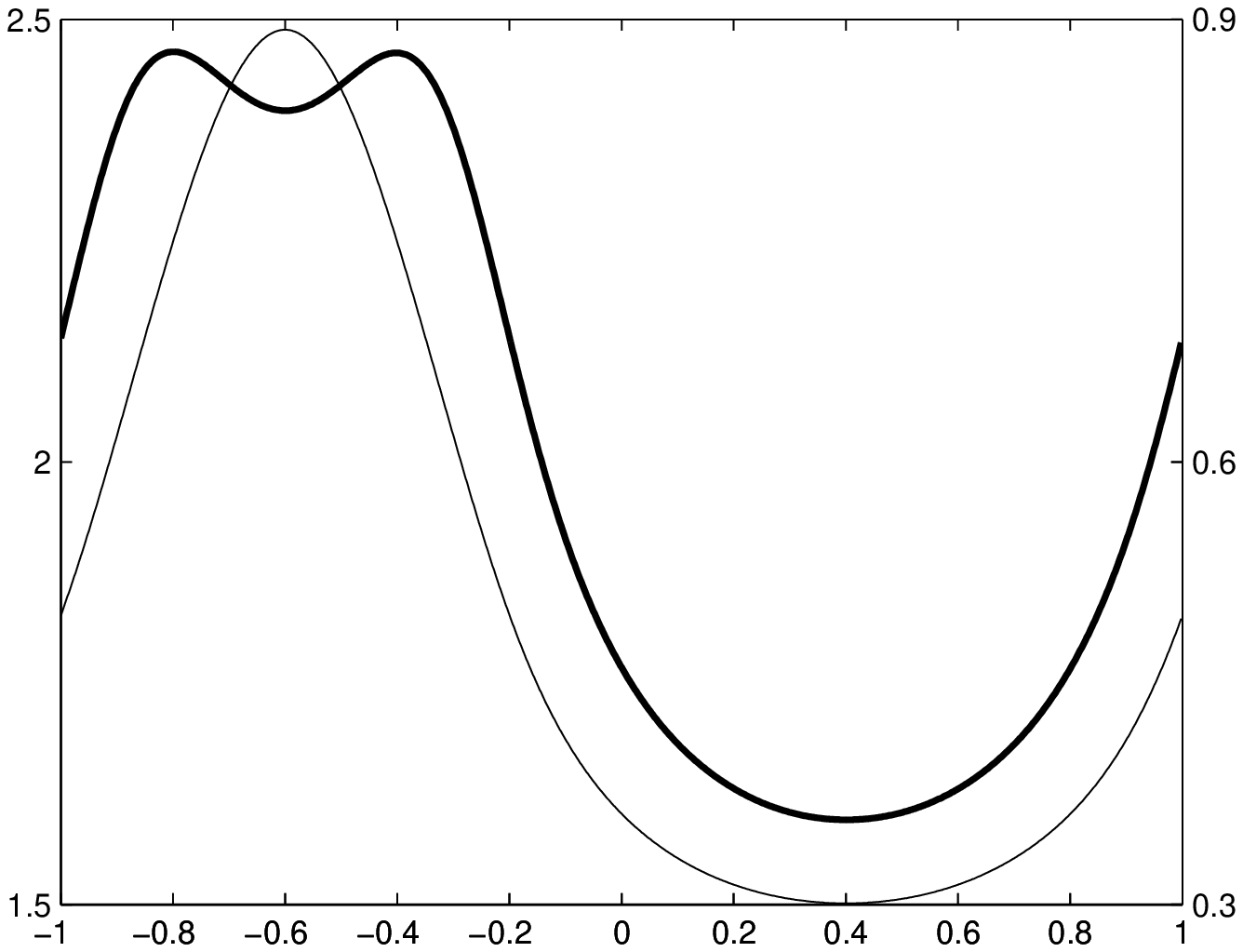}
  & \includegraphics[scale=.3,bb=111 247 509 544,clip]{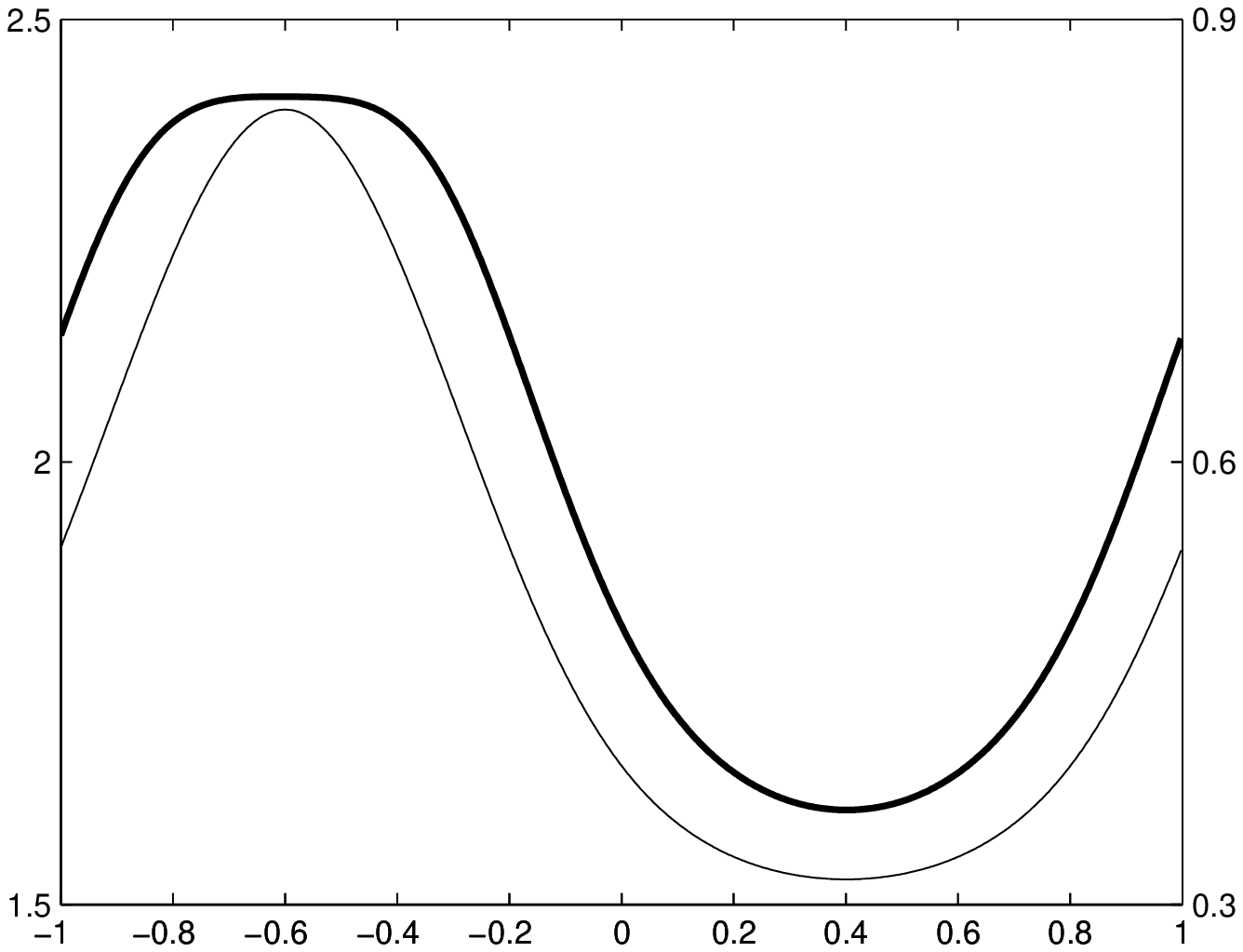}
  & \includegraphics[scale=.3,bb=111 247 509 544,clip]{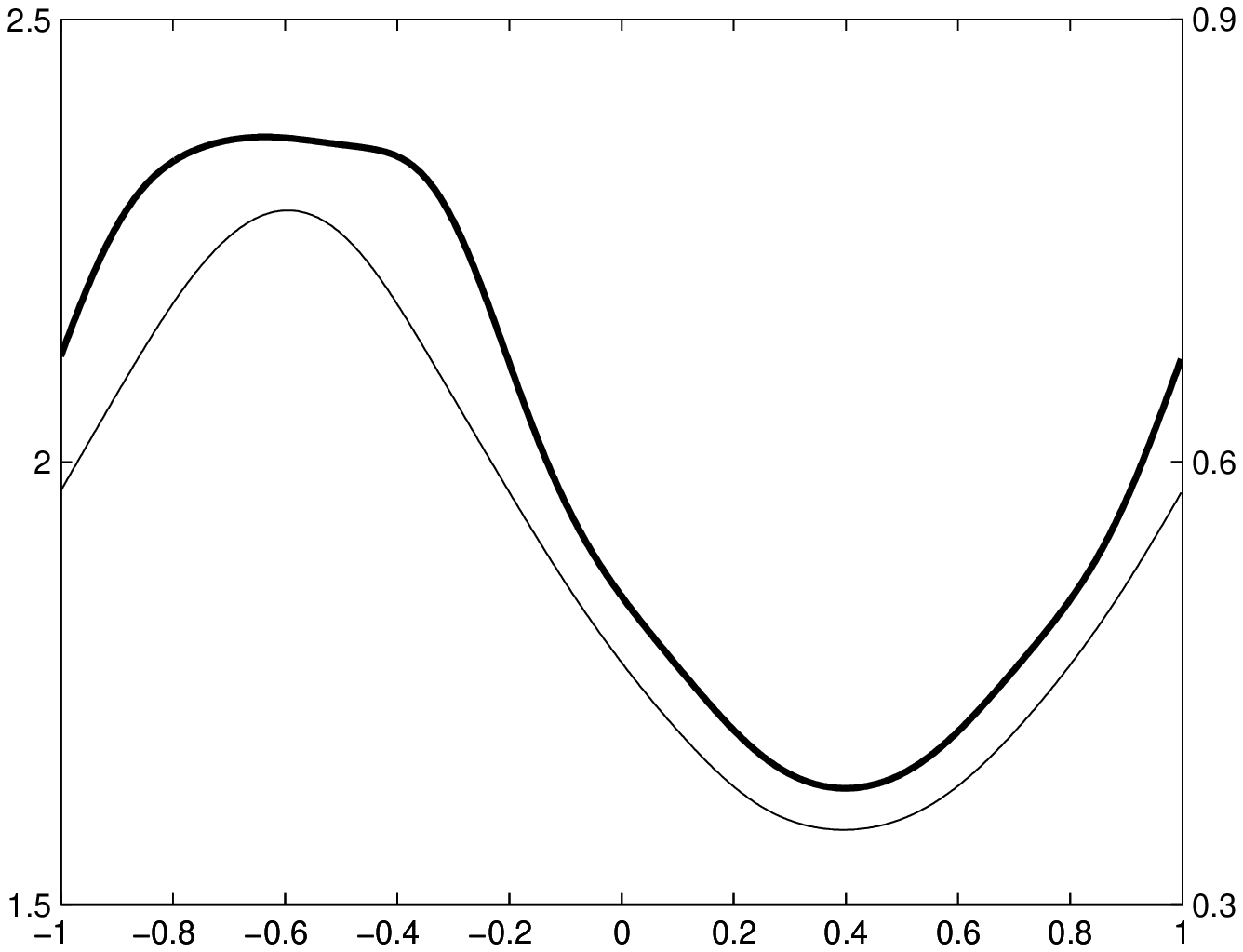}
\end{tabular}
\caption{The numerical results for problem \eqref{eq:periodic}.
The thick line with the left $y$-axis is the plot of density
while the thin line with the right $y$-axis is the plot of
temperature.}
\label{fig:periodic}
\end{figure}

\begin{figure}[!ht]
\centering
\subfigure[R120, $M=7$]{
  \includegraphics[scale=.4]{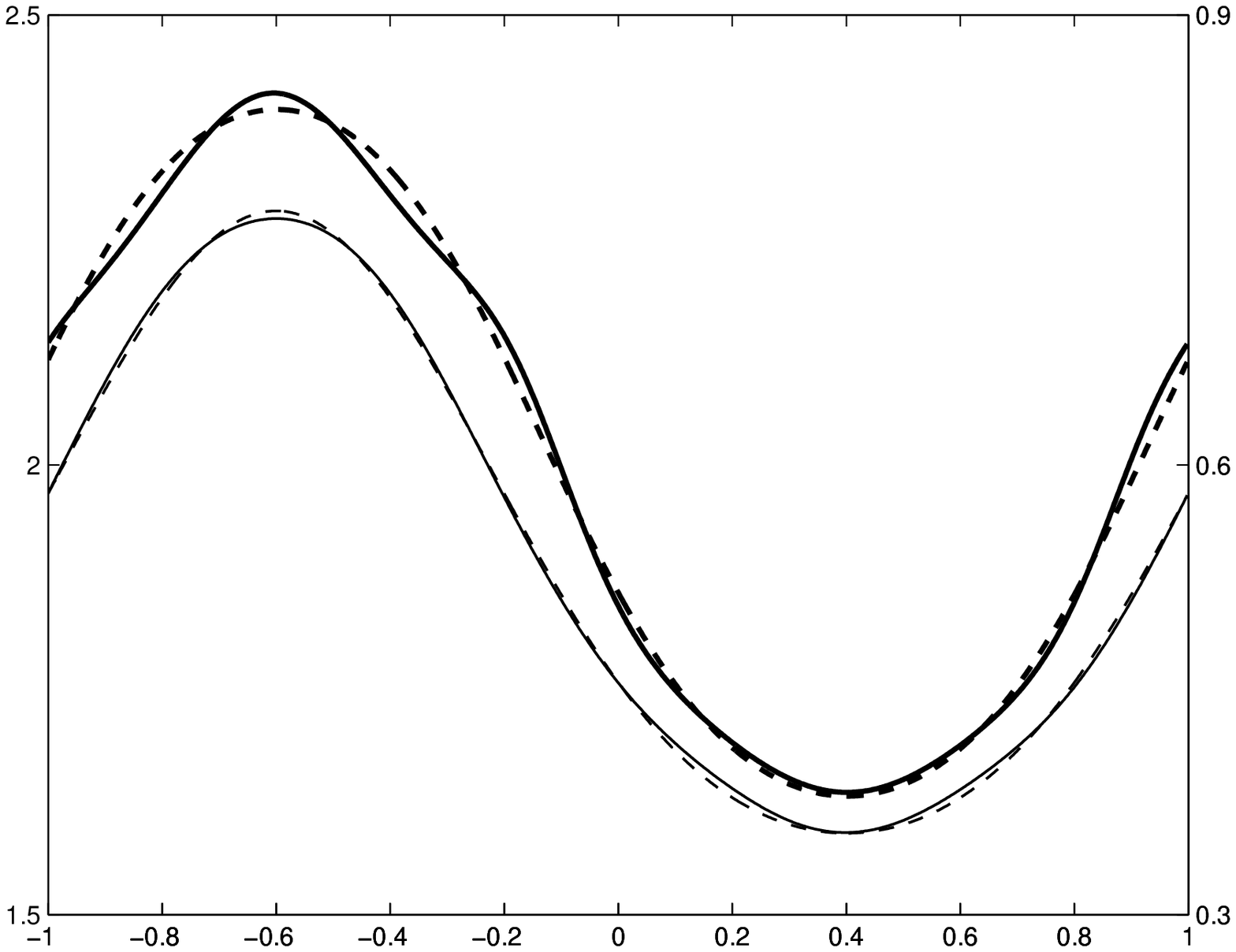}
}
\subfigure[R165, $M=8$]{
  \includegraphics[scale=.4]{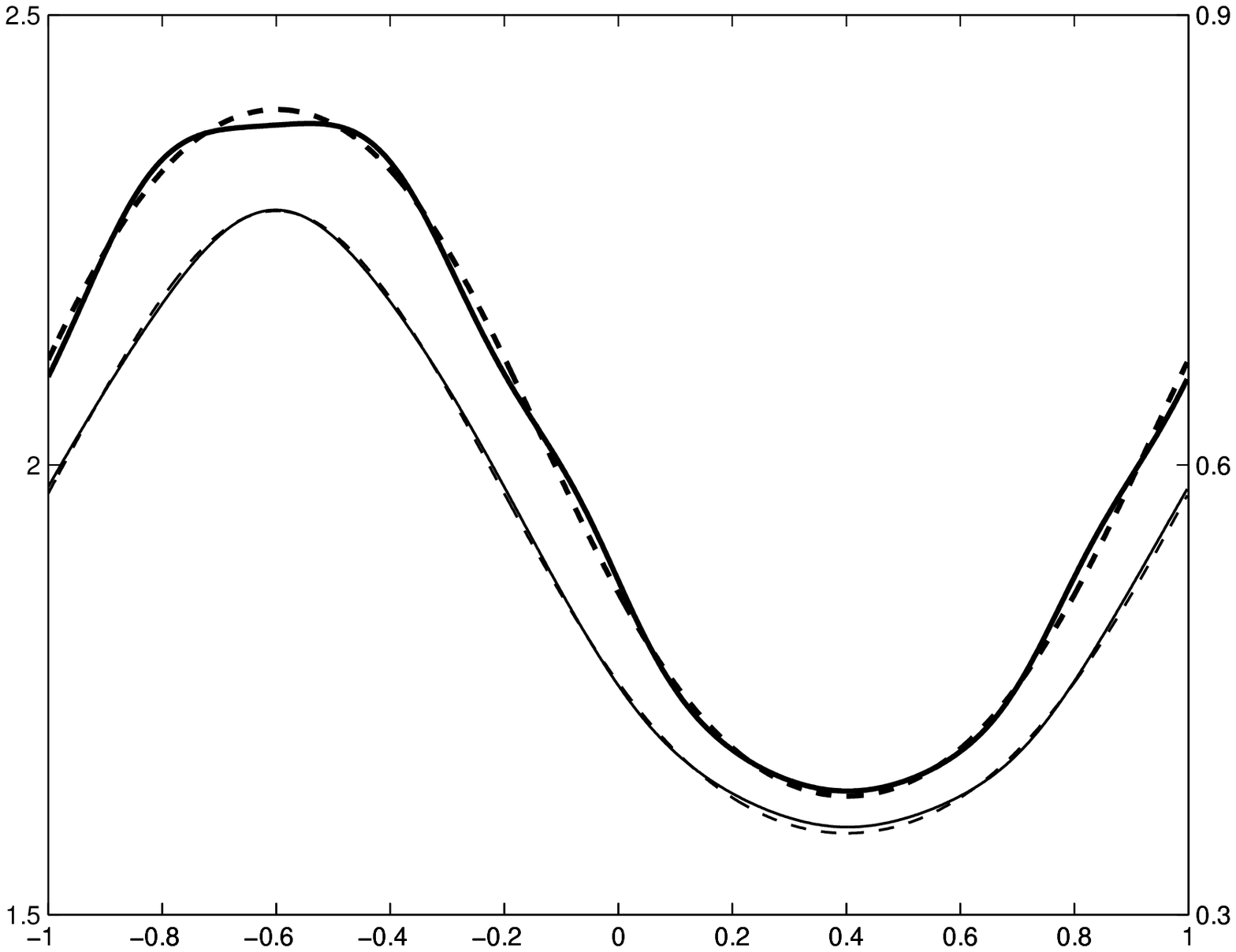}
}
\subfigure[R220, $M=9$]{
  \includegraphics[scale=.4]{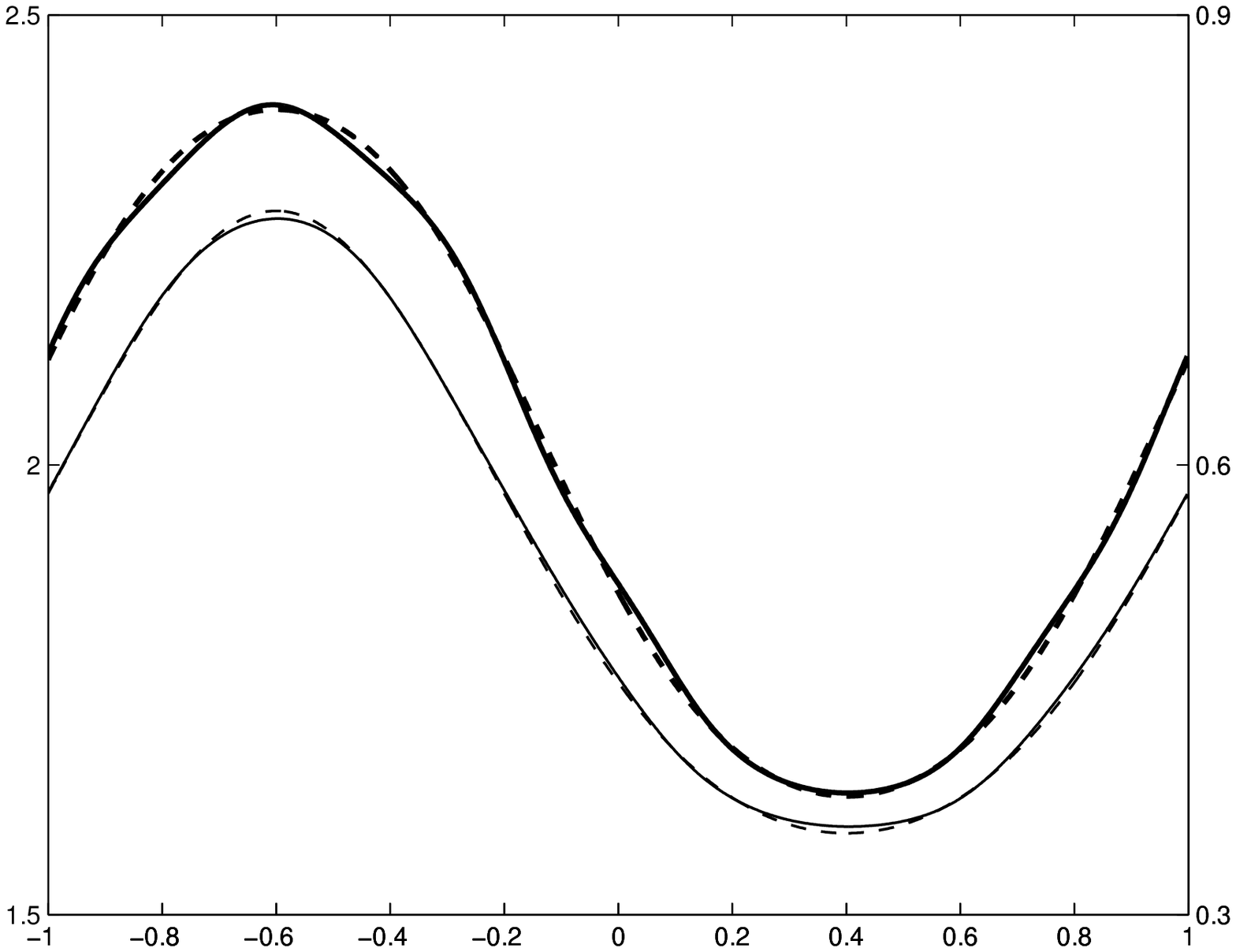}
}
\subfigure[R286, $M=10$]{
  \includegraphics[scale=.4]{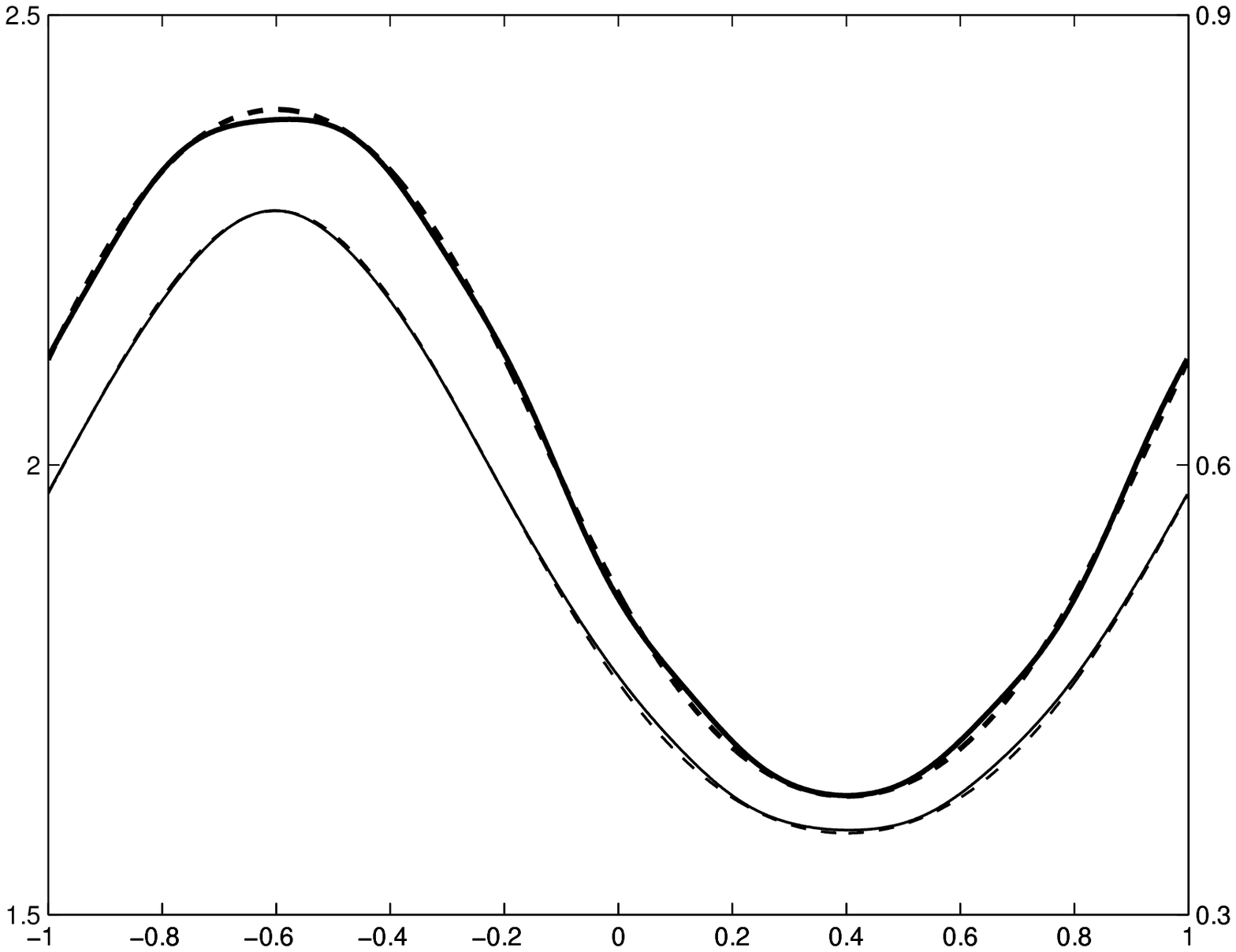}
}
\subfigure[R364, $M=11$]{
  \includegraphics[scale=.4]{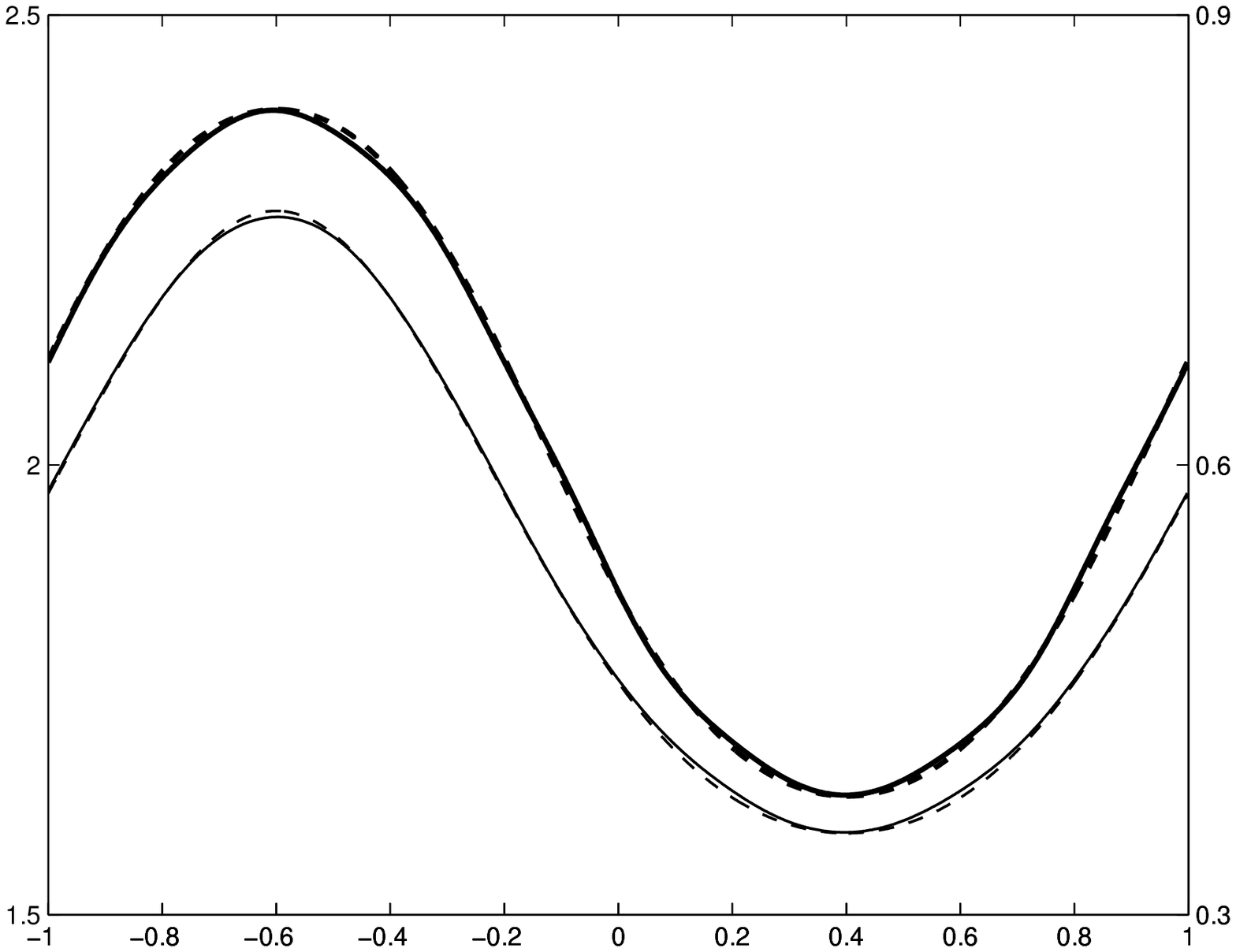}
}
\subfigure[R455, $M=12$]{
  \includegraphics[scale=.4]{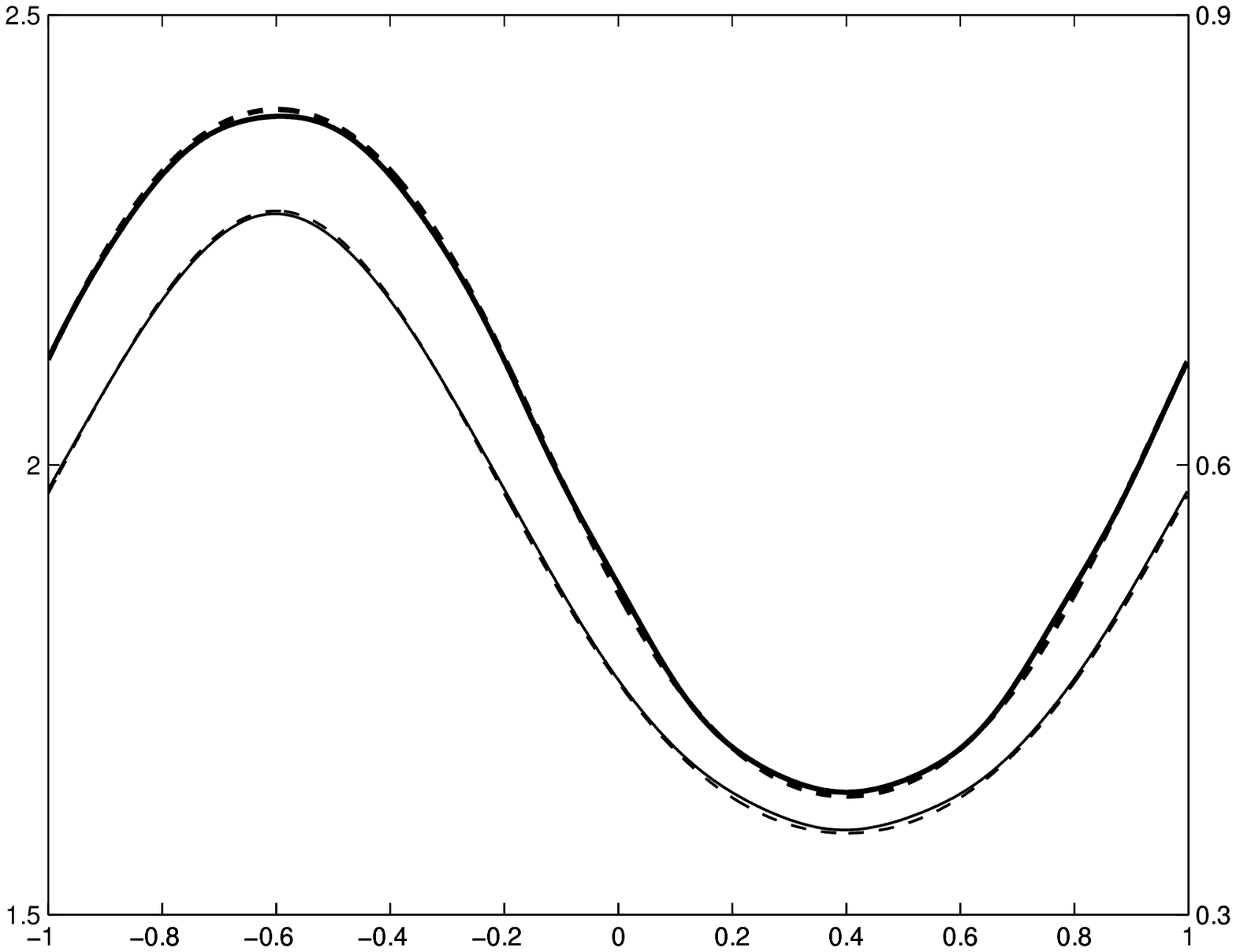}
}
\caption{The numerical results for problem \eqref{eq:periodic}
in the case of $Kn=0.5$. The thick line with the left $y$-axis
is the plot of density while the thin line with the right
$y$-axis is the plot of temperature. The dashed lines are the
numerical solution of the BGK equation.}
\label{fig:periodic_Kn=0.5}
\end{figure}

\subsection{Two dimensional case}
Two 2D examples are investigated in our numerical simulation. Both
examples use uniform grids in the spatial discretization. Though much
more computational cost are needed for 2D problems, the equations with
up to 84 moments are considered.

\subsubsection{Shock-bubble interaction}
In this section, the shock-bubble problem tested in \cite
{Torrilhon2006} is repeated. The initial state contains a shock wave
at $x=-1.0$ travelling with Mach number $M_0 = 2.0$ into an
equilibrium area with $(\rho, \bu, \theta) = (1, 0, 1)$. A bubble is
in front of the shock with density profile
\begin{equation}
\rho(0, \bx) = 1 + 1.5 \exp (-16 |\bx - \bx_0|^2),
\end{equation}
where $\bx_0 = (0.5, 0)^T$, and constant pressure $p = 1$. The shock
wave has a fully developed structure instead of a discontinuity. Thus a
pre-computation of the shock profile is needed. The initial density
surfaces for $\Kn = 0.05$ and $\Kn = 0.1$ are shown in Figure
\ref{fig:ShockBubble_Kn=0.05_init}. A uniform mesh with $1000\times
400$ grids is used in our numerical simulation.

\begin{figure}[!ht]
\centering
\subfigure[$\Kn=0.05$]{
  \includegraphics[scale=.5,clip]{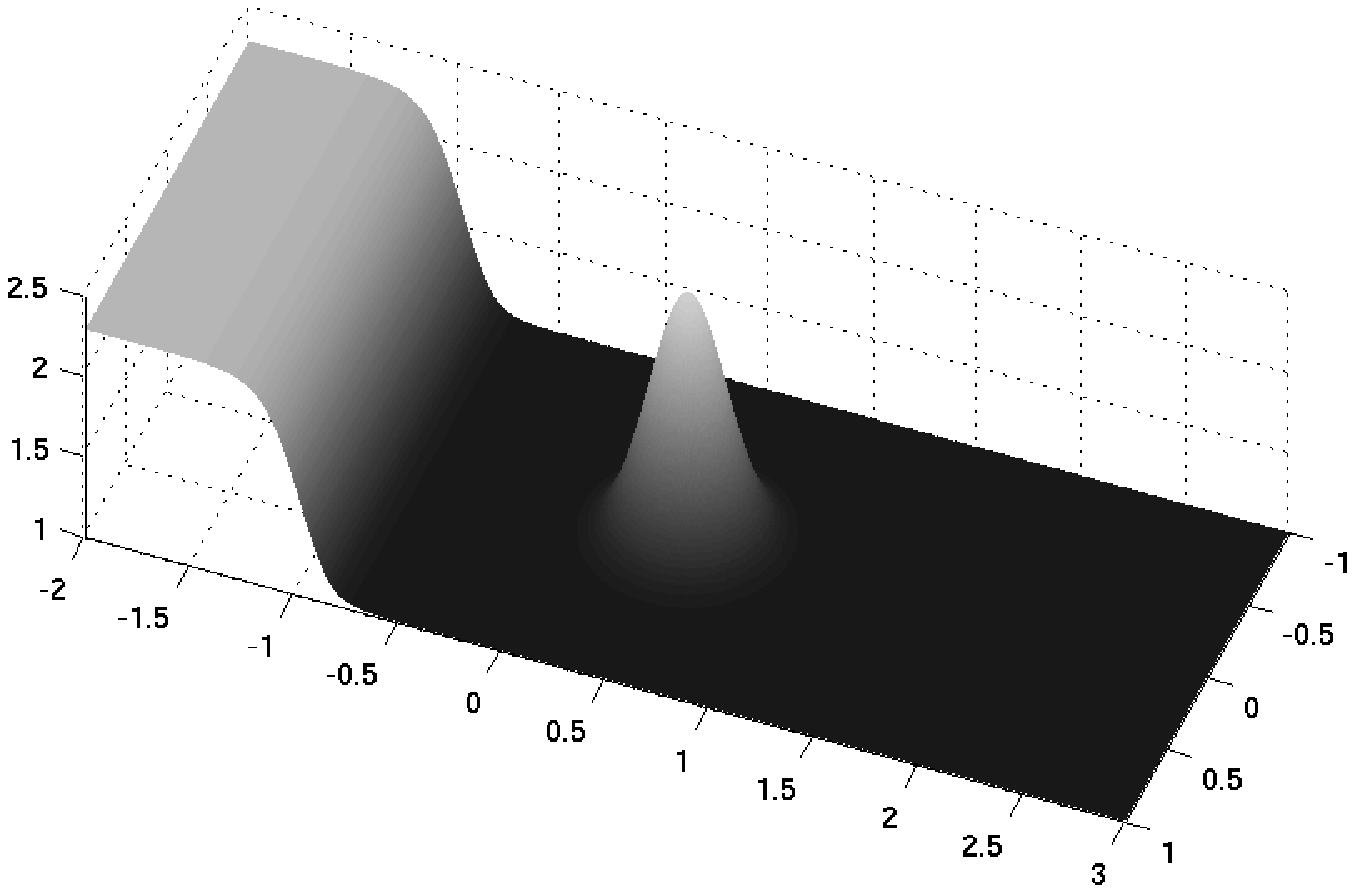}
}
\subfigure[$\Kn=0.1$]{
  \includegraphics[scale=.5,clip]{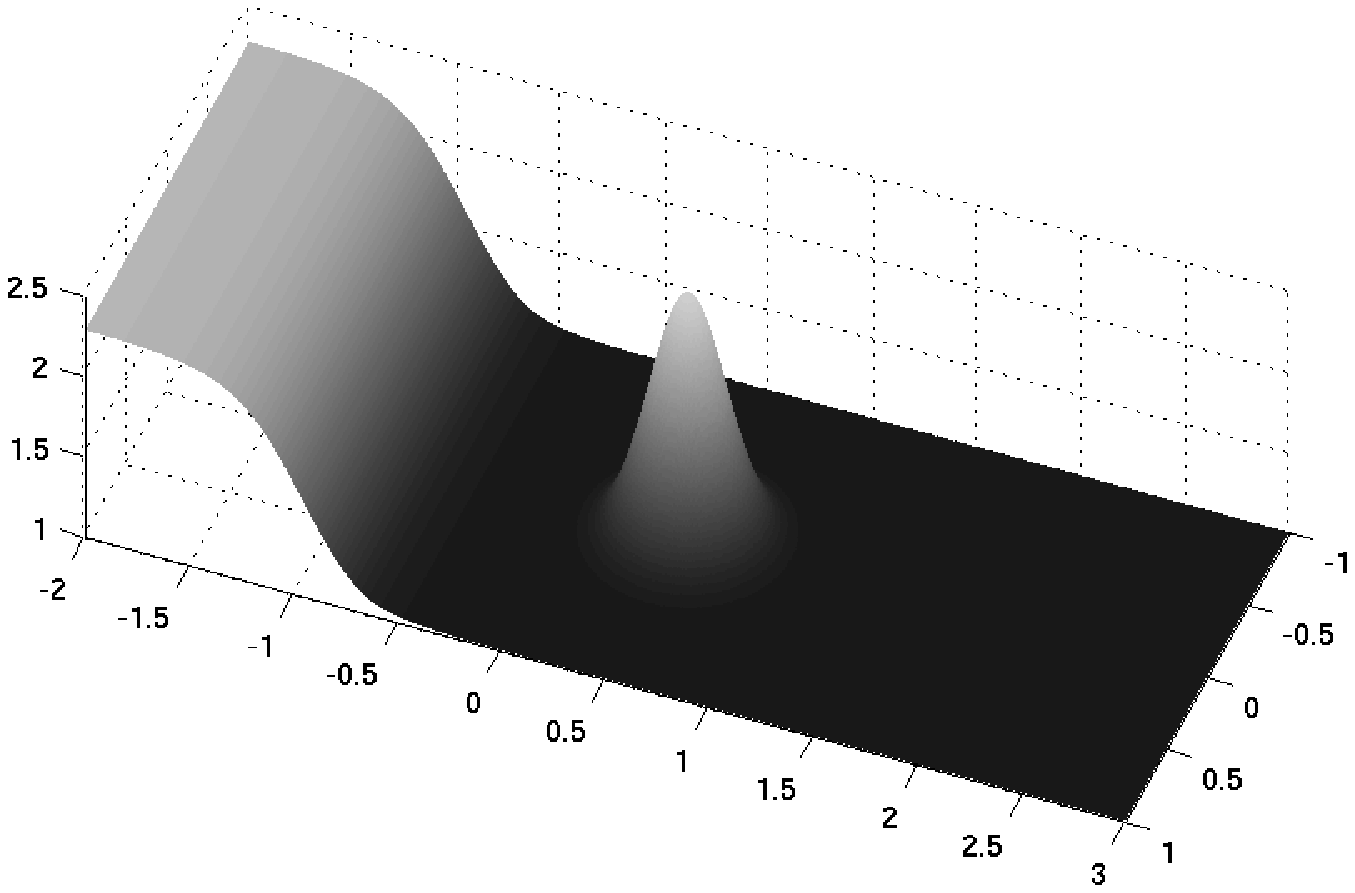}
}
\caption{The initial plot of density}
\label{fig:ShockBubble_Kn=0.05_init}
\end{figure}

The shock structure can be obtained by solving a 1D Riemann problem
constructed according to the Rankine-Hugoniot condition. The left
state is
\begin{equation}
\rho_l = \frac{4 M_0^2}{M_0^2 + 3}, \quad
\bu_l = \left( -\sqrt{\frac{5}{3}}
  \frac{M_0^2 + 3}{4M_0}, 0, 0\right)^T, \quad
p_l = \frac{5M_0^2 - 1}{4},
\end{equation}
and the right state is
\begin{equation}
\rho_r = 1, \quad
\bu_r = \left( -\sqrt{\frac{5}{3}} M_0, 0, 0 \right)^T,
\quad p_r = 1.
\end{equation}
Both states are in equilibrium.  After a sufficiently long time, a
stationary shock will form. It is quite convenient to transform a
stationary shock to an unstable one in our numerical framework.
Suppose a 1D steady shock is presented by
\begin{equation} \label{eq:shock}
f(x, \bxi) = \sum_{|\alpha| \leqslant M}
  f_{\alpha}(x) \mathcal{H}_{\theta(x),\alpha} (\bv(x)),
\quad \bv(x) = \frac{\bxi - \bu(x)}{\sqrt{\theta(x)}},
\quad x \in \bbR,
\end{equation}
and it satisfies \eqref{eq:coef_restriction}. Then, for
an arbitrary velocity $\boldsymbol{s} = (s,0,0)^T$, let
\begin{equation}
\bu'(x) = \bu(x) + \boldsymbol{s}, \quad \text{and}
  \quad \bv'(x) = \frac{\bxi - \bu'(x)}{\sqrt{\theta(x)}}.
\end{equation}
Substituting $\bu$ and $\bv$ by $\bu'$ and $\bv'$ in
\eqref{eq:shock} and keeping all the coefficients unchanged,
then \eqref{eq:shock} becomes an unsteady shock travelling
with speed $\boldsymbol{s}$. Let $\boldsymbol{s} =
(\sqrt{5M_0/3}, 0, 0)^T$, and then the desired shock wave
can be generated. The initial values of $\theta$ for
$\Kn = 0.01$, $0.05$ and $0.1$ at $y=0$ are plotted in
Figure \ref{fig:ShockBubble_init}.

\begin{figure}[!ht]
\centering
\includegraphics[scale=.5]{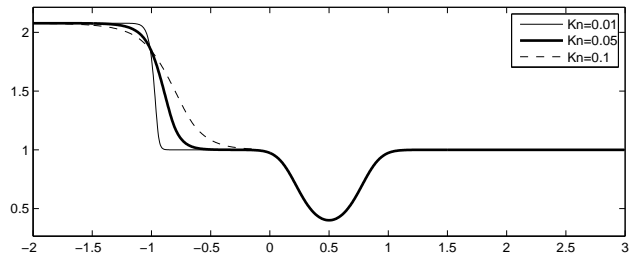}
\caption{The initial values of temperature at $y=0$}
\label{fig:ShockBubble_init}
\end{figure}

This example is aimed at the validation of our algorithm in the 2D
case. To make comparison with the results in \cite{Torrilhon2006},
$\Kn = 0.01$, $\Kn = 0.05$ and $\Kn = 0.1$ are considered and only the
R20 equations are simulated since it is the closest moment system
to R13. Results for the dense case $\Kn = 0.01$ at $t = 0.8$ are shown
in Figure \ref{fig:ShockBubble_Kn=0.01}. Comparing with the results in
\cite{Torrilhon2006}, the profile exhibits a qualitatively agreement
while the peak value between $x=0.5$ and $x=1$ disagrees. This is
believed to be caused by the highly dissipative numerical flux without
gradient reconstruction in our implementation. For $\Kn=0.05$, our R20
results of density and temperature (Figure
\ref{fig:ShockBubble_Kn=0.05}) are much closer to those presented
in \cite{Torrilhon2006}. But again, the heat fluxes show the same
profile with different magnitude, owing to the BGK model used
here. The whole structure after interaction with the bubble is drawn
in Figure \ref{fig:ShockBubble_Kn=0.05_t=0.9}, with a good agreement
with the former results.

\begin{figure}[!ht]
\centering
\subfigure[Density plot]{
  \includegraphics[scale=.45]{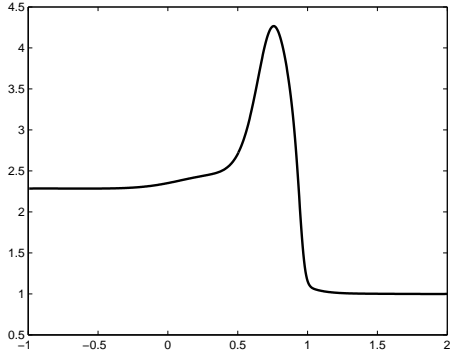}
}
\subfigure[Temperature plot]{
  \includegraphics[scale=.45]{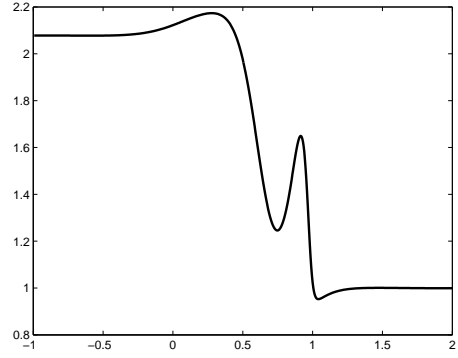}
}
\caption{R20 results of shock-bubble interaction for
$\Kn=0.01$ at $t=0.8$ and $y=0$}
\label{fig:ShockBubble_Kn=0.01}
\end{figure}

\begin{figure}[!ht]
\centering
\subfigure[Density plot at $y=0$]{
  \includegraphics[scale=.45]{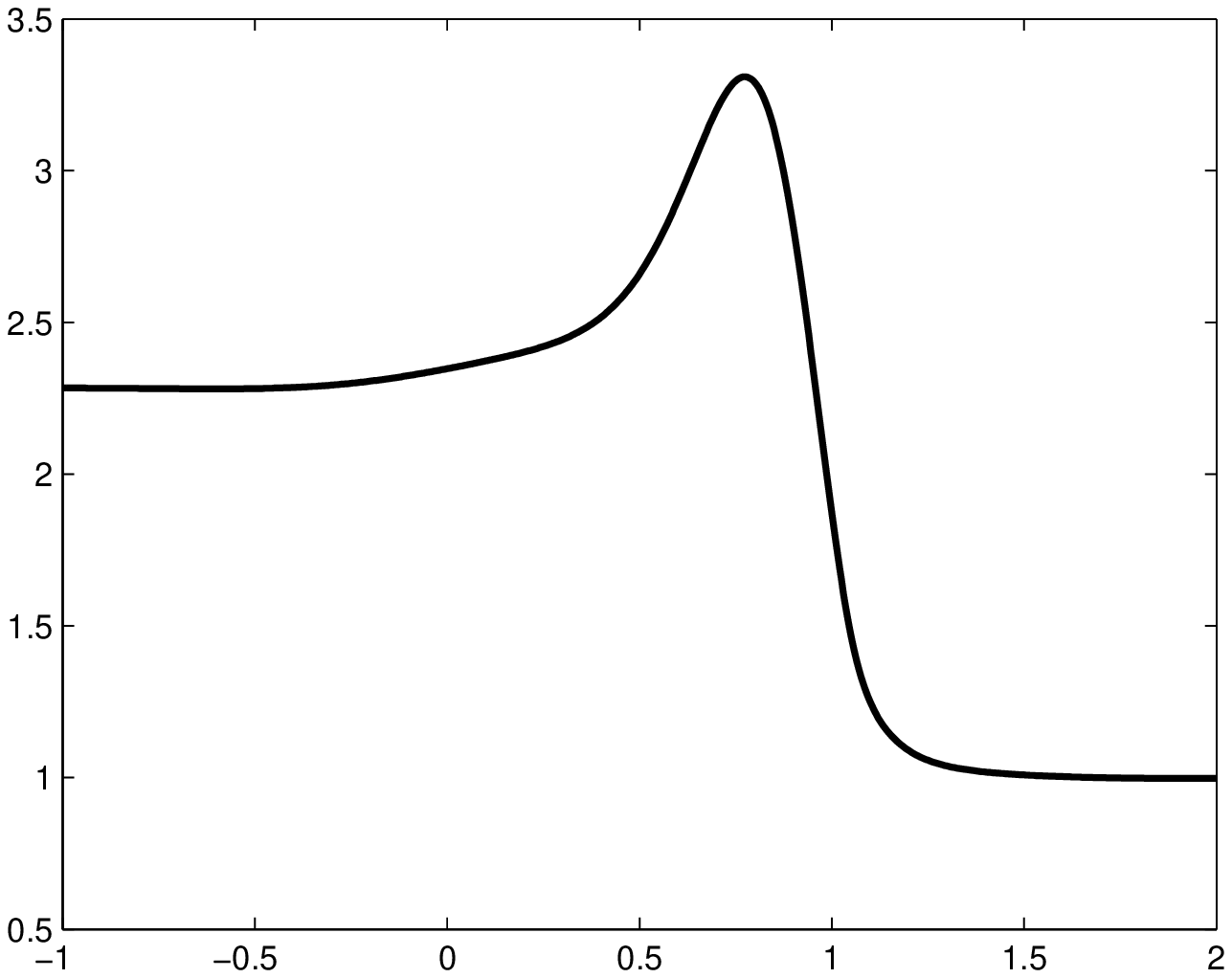}
}
\subfigure[Temperature plot at $y=0$]{
  \includegraphics[scale=.45]{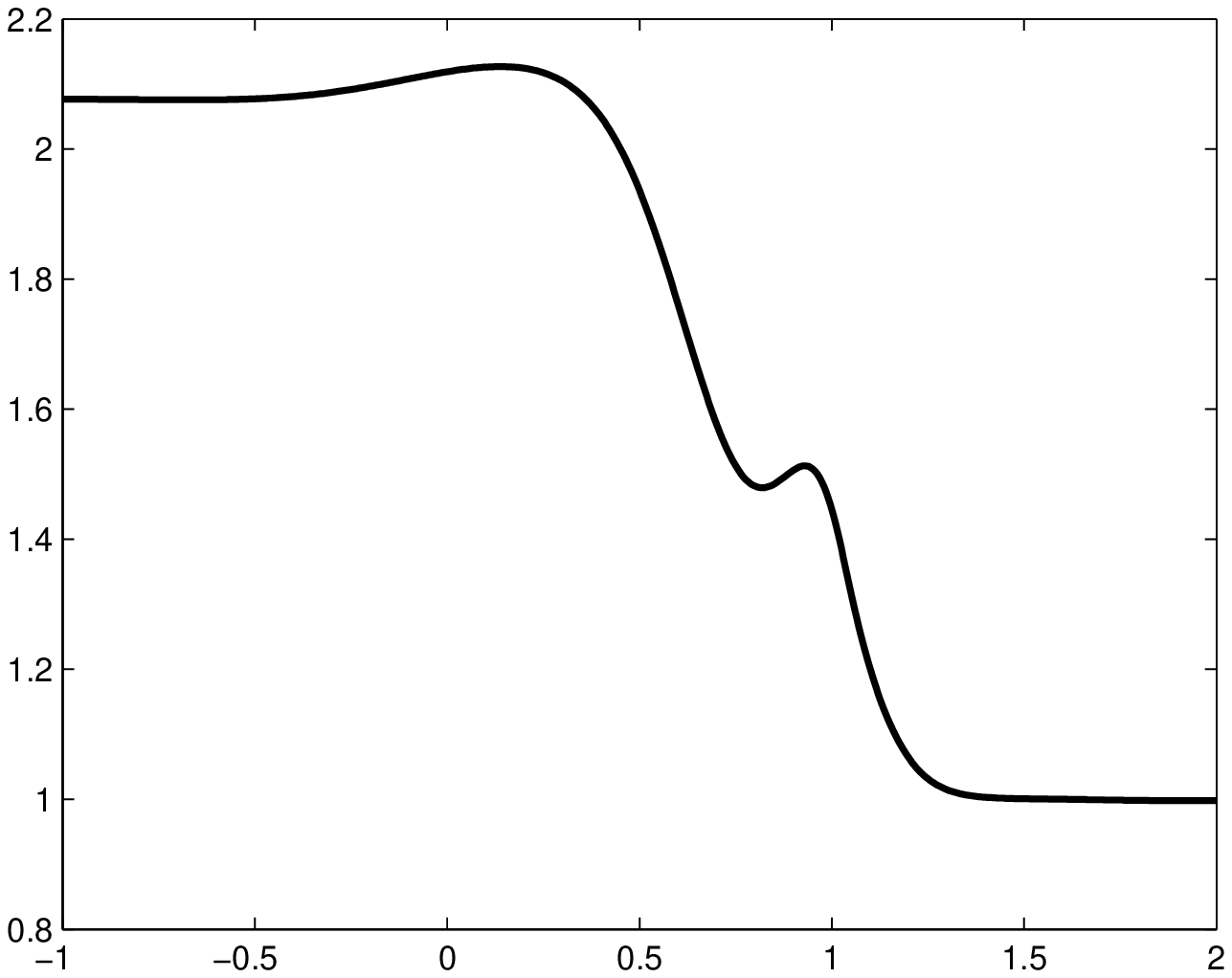}
}
\subfigure[Heat flux $q_y$ at $y=0.4$]{
  \includegraphics[scale=.45]{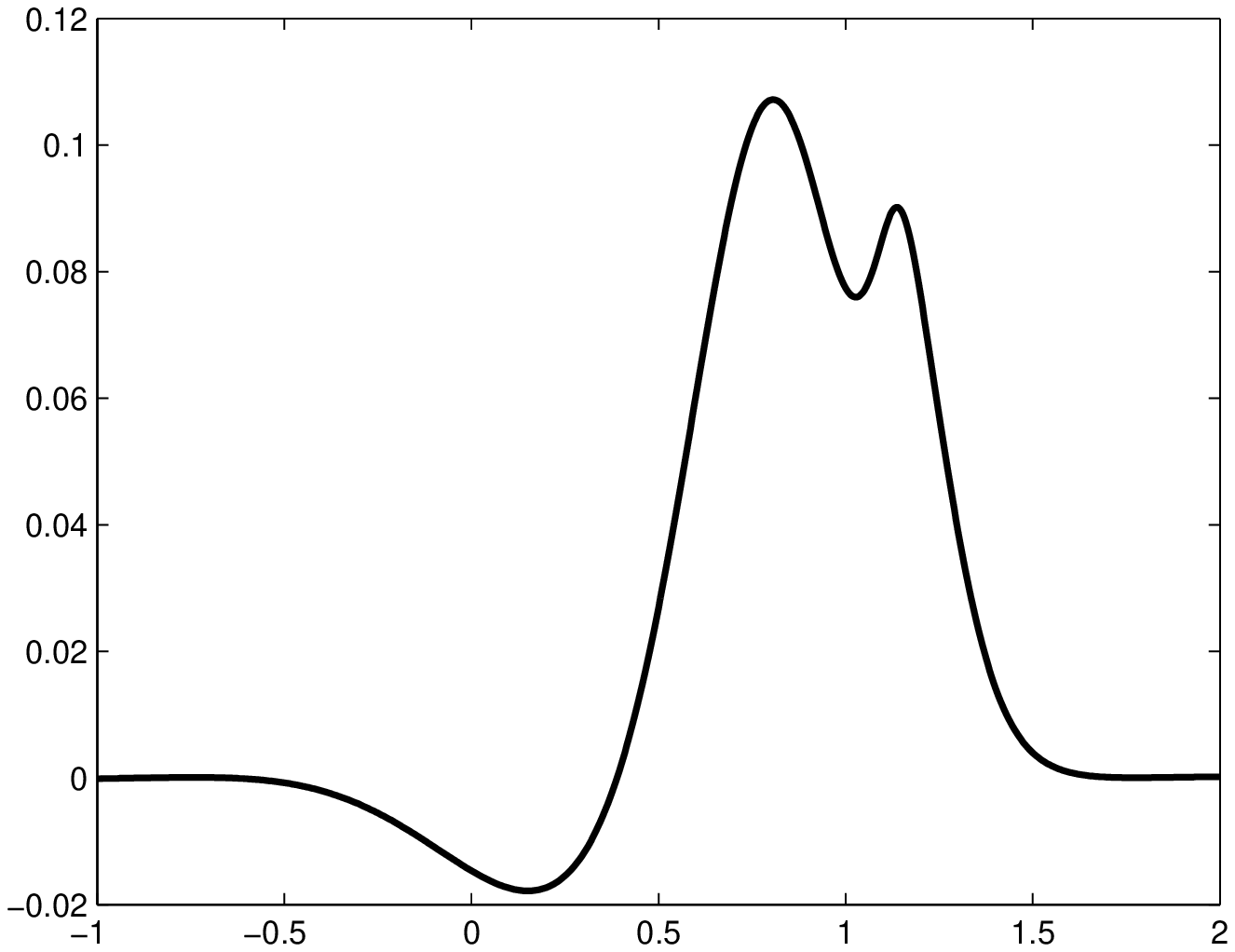}
}
\subfigure[Heat flux $q_x$ at $y=0$]{
  \includegraphics[scale=.45]{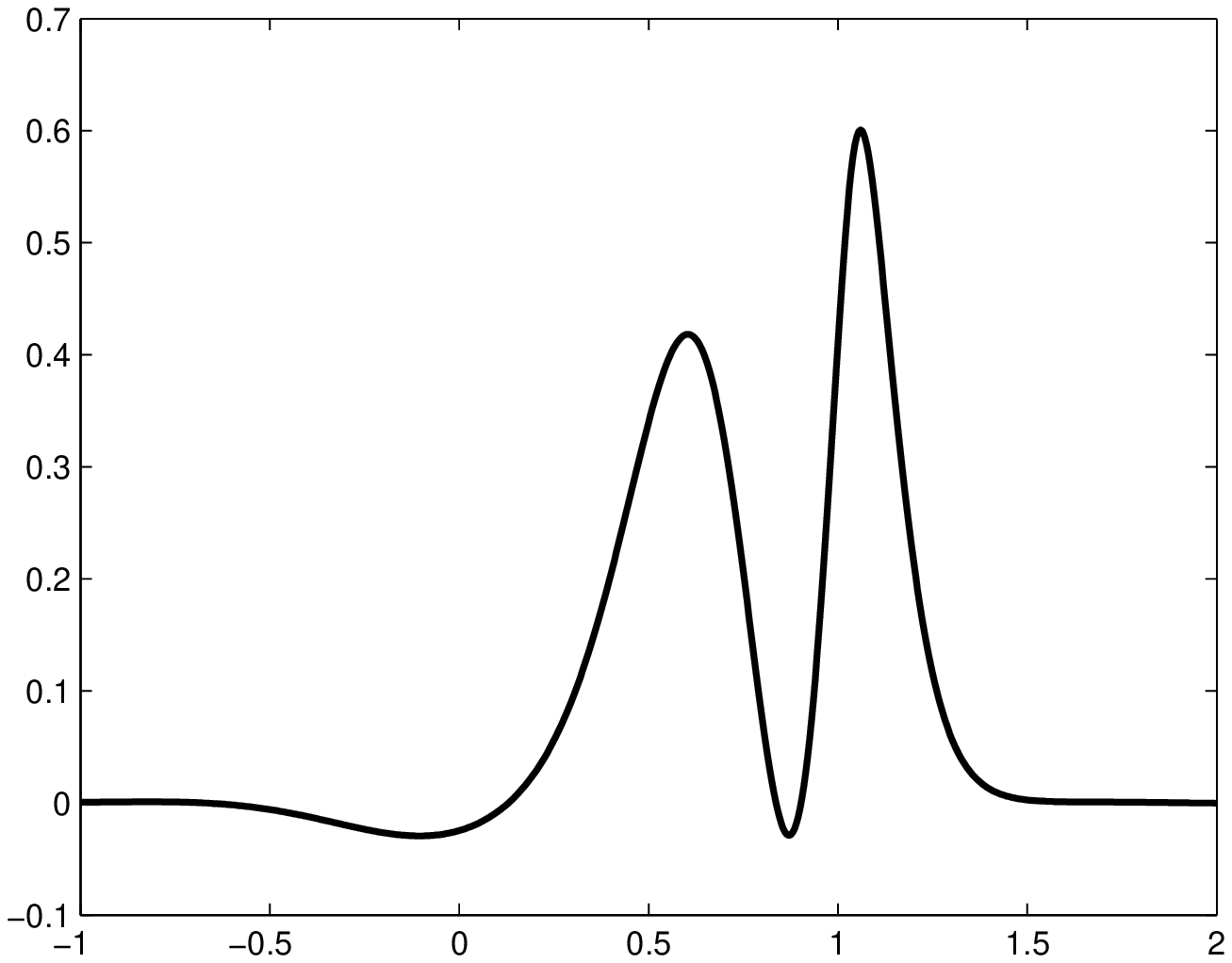}
}
\caption{R20 results of shock-bubble interaction for
$\Kn = 0.05$ at $t = 0.8$}
\label{fig:ShockBubble_Kn=0.05}
\end{figure}

\begin{figure}[!ht]
\centering
\subfigure[Density]{
  \includegraphics[scale=.5,clip]{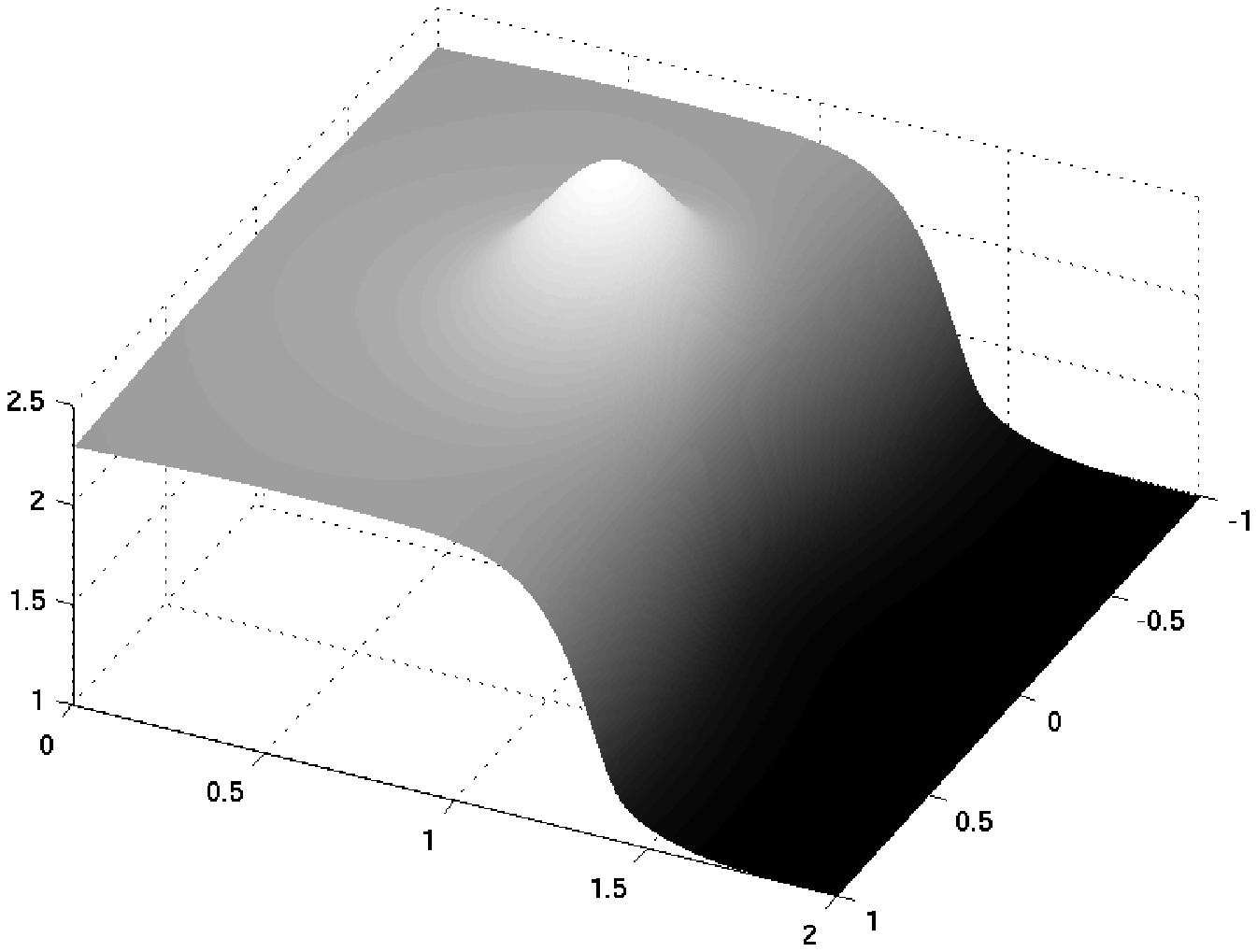}
}
\subfigure[Temperature]{
  \includegraphics[scale=.5,clip]{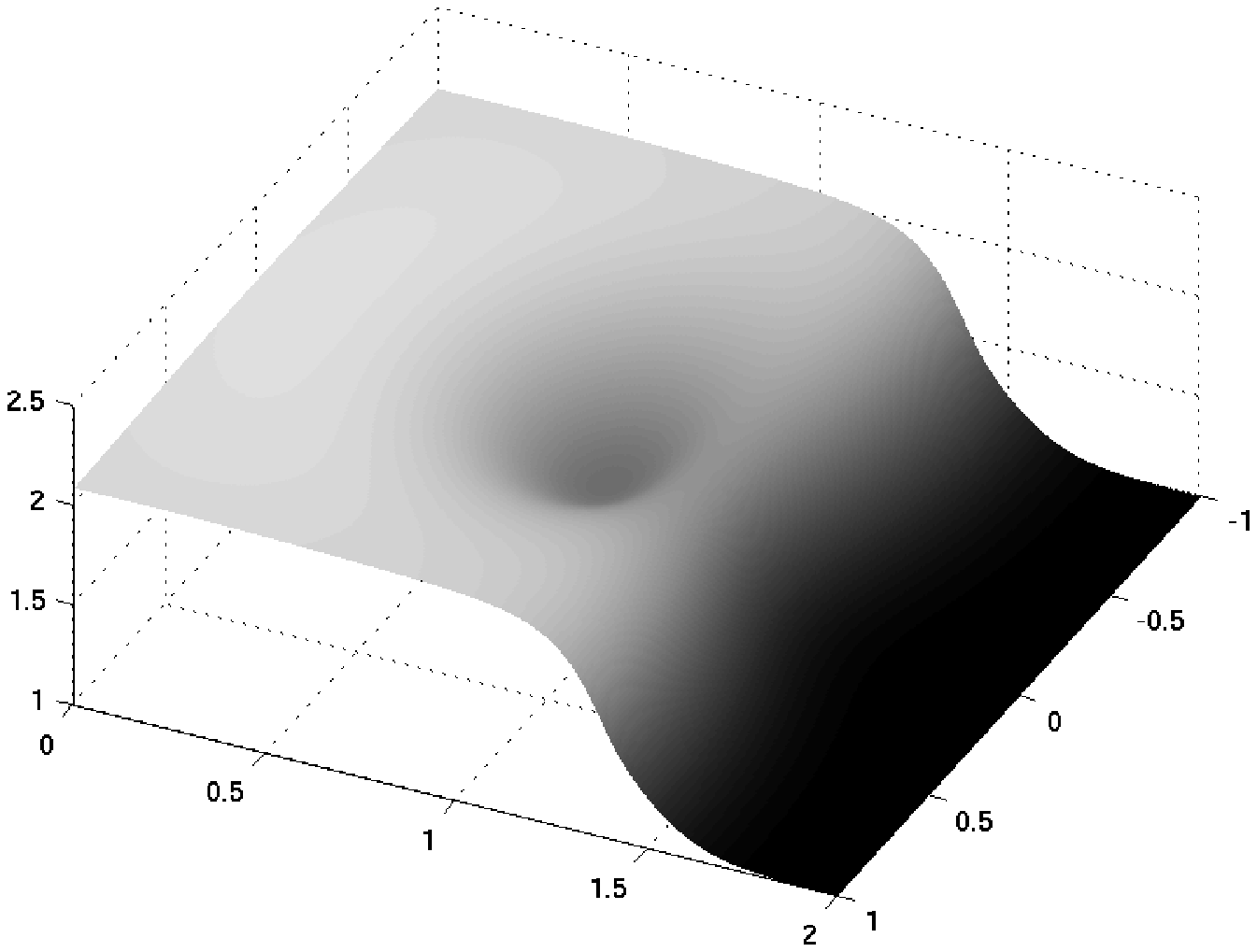}
}
\caption{R20 results of shock-bubble interaction for
$\Kn = 0.05$ at $t = 0.9$}
\label{fig:ShockBubble_Kn=0.05_t=0.9}
\end{figure}

For $\Kn = 0.1$, we know from Figure \ref{fig:periodic} that R20
results deviate from the BGK solution slightly, so some deviation
between R13 and R20 results is reasonable.  Our R20 results are
plotted in Figure \ref{fig:ShockBubble_Kn=0.1_M=3}. A comparison with
R13 results in \cite{Torrilhon2006} shows that both R13 and R20
equations are able to give correct structures of density and
temperature, while NSF is not.

\begin{figure}[!ht]
\centering
\subfigure[Density plot]{
  \includegraphics[scale=.45]{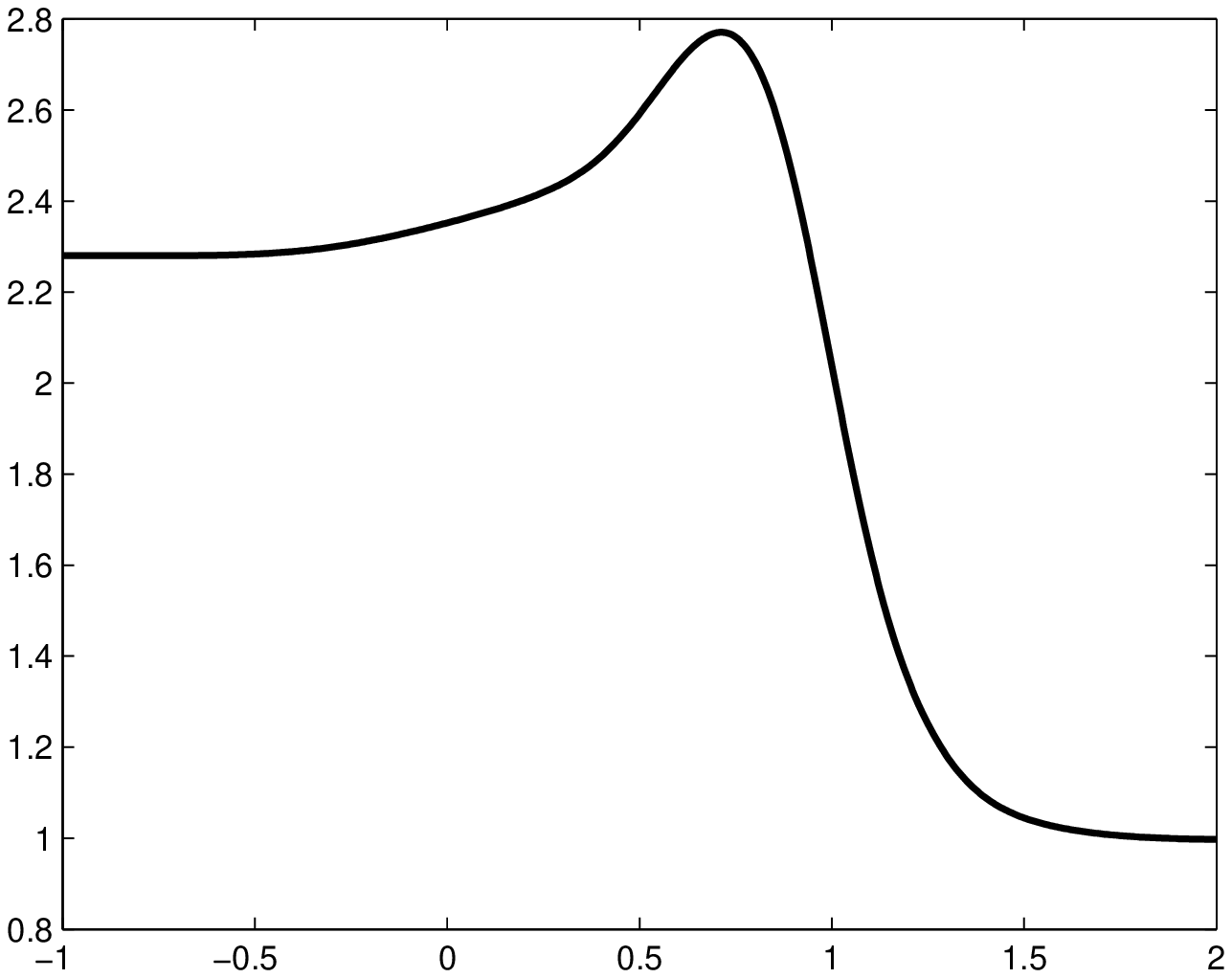}
}
\subfigure[Temperature plot]{
  \includegraphics[scale=.45]{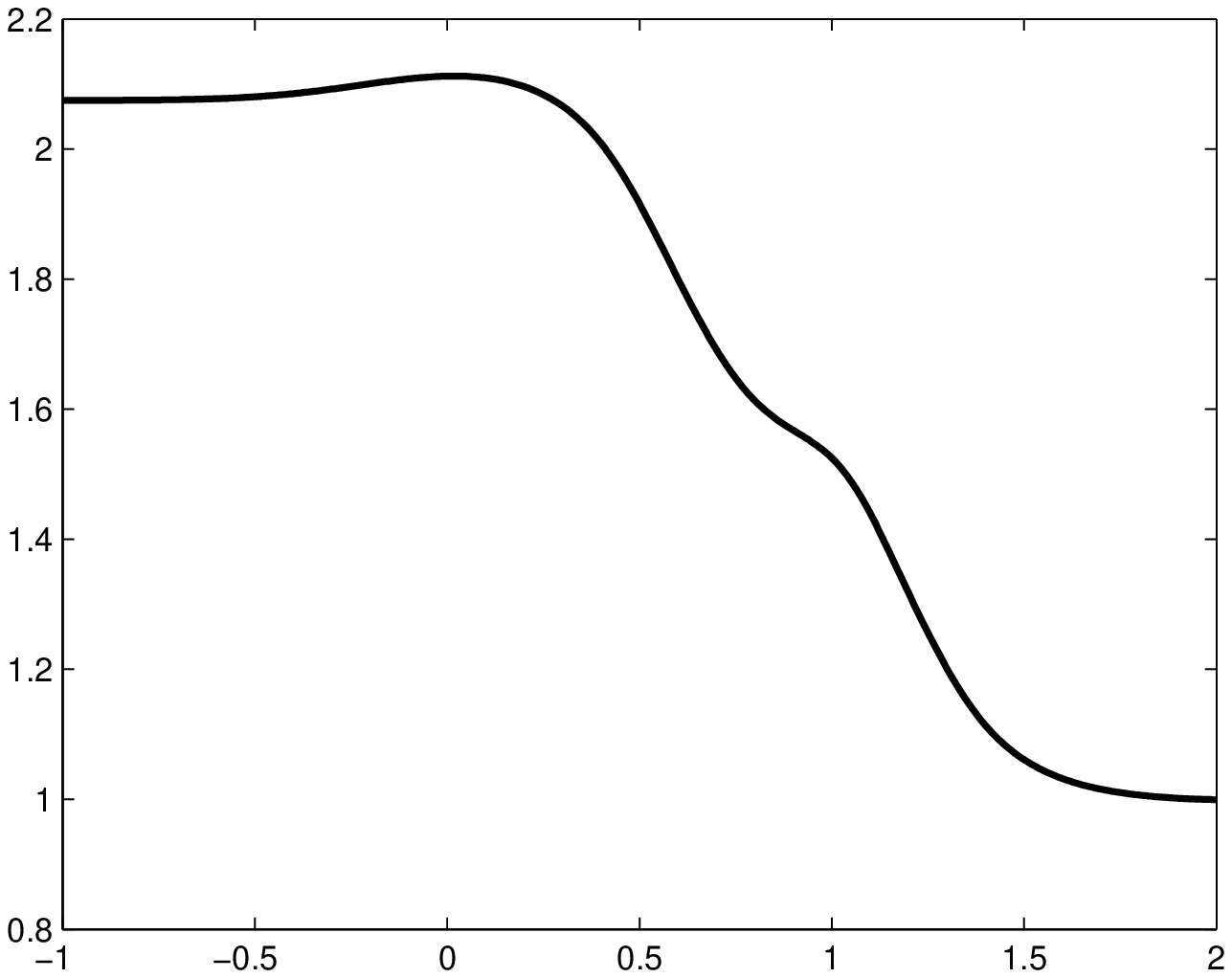}
}
\caption{R20 results of shock-bubble interaction for
$\Kn = 0.1$ at $t = 0.8$ and $y = 0$}
\label{fig:ShockBubble_Kn=0.1_M=3}
\end{figure}

\subsubsection{An example with three-dimensional velocity}
In all the numerical examples above, the $z$-component of
velocity is always zero. Now we consider an example with
three-dimensional velocity with initial conditions as
\begin{equation} \label{eq:periodic2D}
\begin{split}
\rho(0, \bx) & = 2 + \frac{1}{2} \cos(\pi x_1)
  + \frac{1}{2} \sin(\pi x_2), \quad p(0, \bx) = 1, \\
\bu(0, \bx) & = \left( \begin{array}{c}
  1 + \frac{1}{2} \sin(\pi x_1) + \frac{1}{2} \cos(\pi x_2) \\[5pt]
  \frac{1}{2} \sin(\pi x_1) + \frac{1}{2} \cos(\pi x_2) \\[5pt]
  \frac{1}{2} \sin(\pi x_1) + \frac{1}{2} \cos(\pi x_2)
\end{array} \right).
\end{split}
\end{equation}
The fluid is in equilibrium over the whole computational domain
$[-1,1] \times [-1,1]$ together with periodic boundary condition.

For this example, the simulations of the R20 and R84 equations with
$\Kn = 0.1$ are carried out. The numerical solution of the R84
equations on meshes with different sizes are compared to check the
spatial convergence order of our scheme. In the case of no exact
solution being available, we take the numerical result on a mesh with
$500\times 500$ grids as the reference solution. Other results are
computed on meshes with $N_x \times N_x$ grids, where $N_x = 10$ up
to $200$. HLL flux without gradient reconstruction is used in our
finite volume scheme, so the convergence rate is expected to be the
first order.

The numerical solution is shown in Figure \ref {fig:Periodic2D_t=0.2}
and Figure \ref{fig:Periodic2D_t=0.4}. For density and temperature,
R20 and R84 results are almost identical at both $t=0.2$ and $t=0.4$.
However, observable deviation appears in the vertical heat flux in
both Figure \ref{fig:Periodic2D_t=0.2} and Figure
\ref{fig:Periodic2D_t=0.4}. The nontrivial vertical heat flux $q_3$
declares the capacity of our method to simulate 3D non-equilibrium
processes. The $L^1$ errors of the solutions on different meshes are
illustrated in Figure \ref{fig:Periodic2D_err}, where $E(\cdot)$ is
calculated by
\begin{equation}
E(\psi) = \log_{10} \sum_{i=1}^{N_x^{\mathrm{(ref)}}}
  \sum_{j=1}^{N_x^{\mathrm{(ref)}}}
  \Delta x_1^{\mathrm{(ref)}}
  \Delta x_2^{\mathrm{(ref)}}
    |\psi^{\mathrm{(num)}} (\bx_{i,j}^{\mathrm{(ref)}})
      - \psi^{\mathrm{(ref)}} (\bx_{i,j}^{\mathrm{(ref)}})|.
\end{equation}
Here all symbols with superscript ``$\mathrm{(ref)}$'' stand for the
corresponding quantities in the reference solution, i.e. the solution
on the $500\times 500$ mesh, and the symbol with superscript
``$\mathrm{(num)}$'' is the solution on the coarse mesh, which is
considered as piecewise constant. Obviously, first order convergence
rate is achieved.

\begin{figure}[!ht]
\centering
\begin{tabular}{c@{}ccc}
& Density $\rho$ & Temperature $\theta$ & Vertical heat flux $q_3$ \\
\raisebox{43pt}{\small R20}
  & \includegraphics[scale=.3,bb=111 247 509 544,clip]{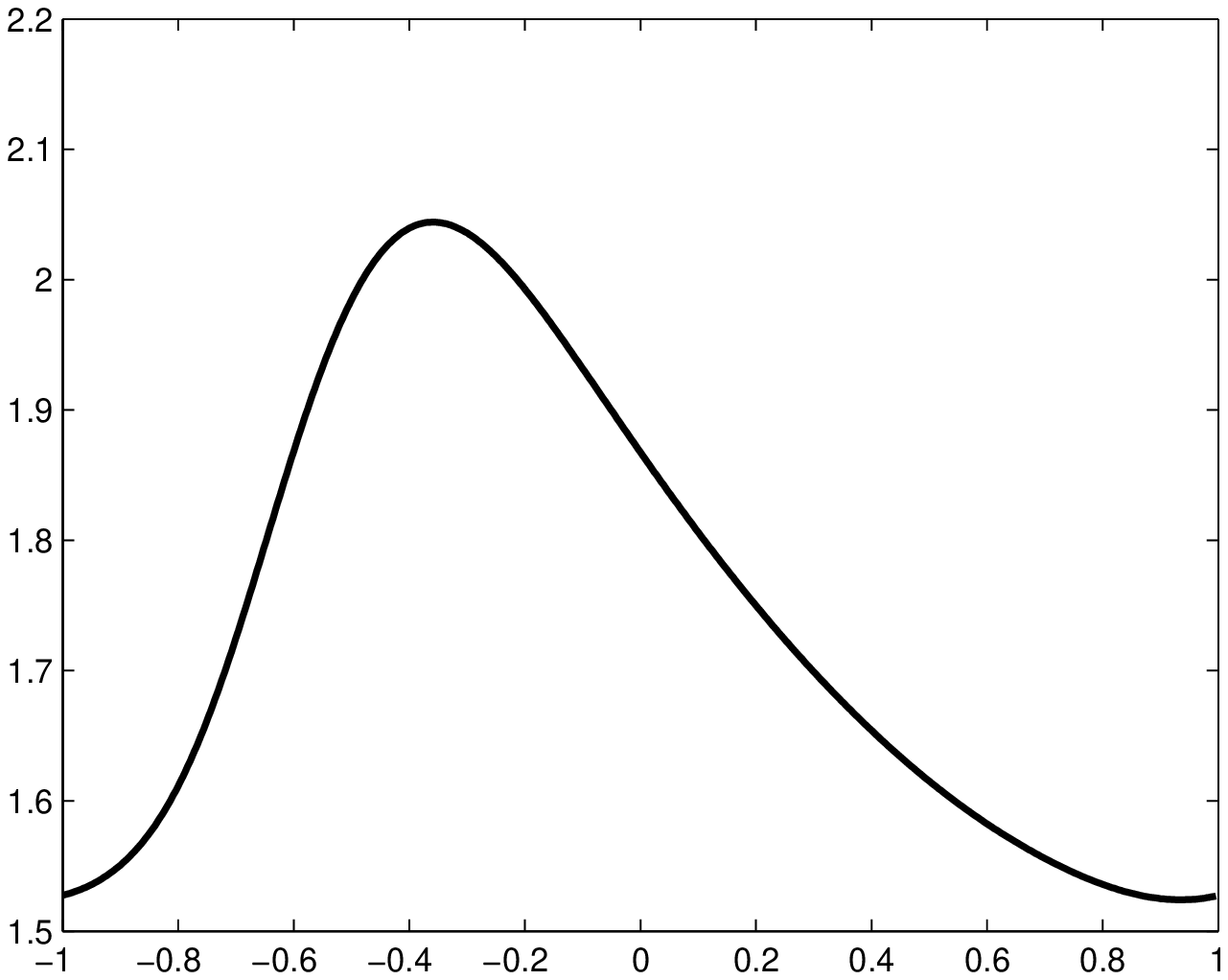}
  & \includegraphics[scale=.3,bb=111 247 509 544,clip]{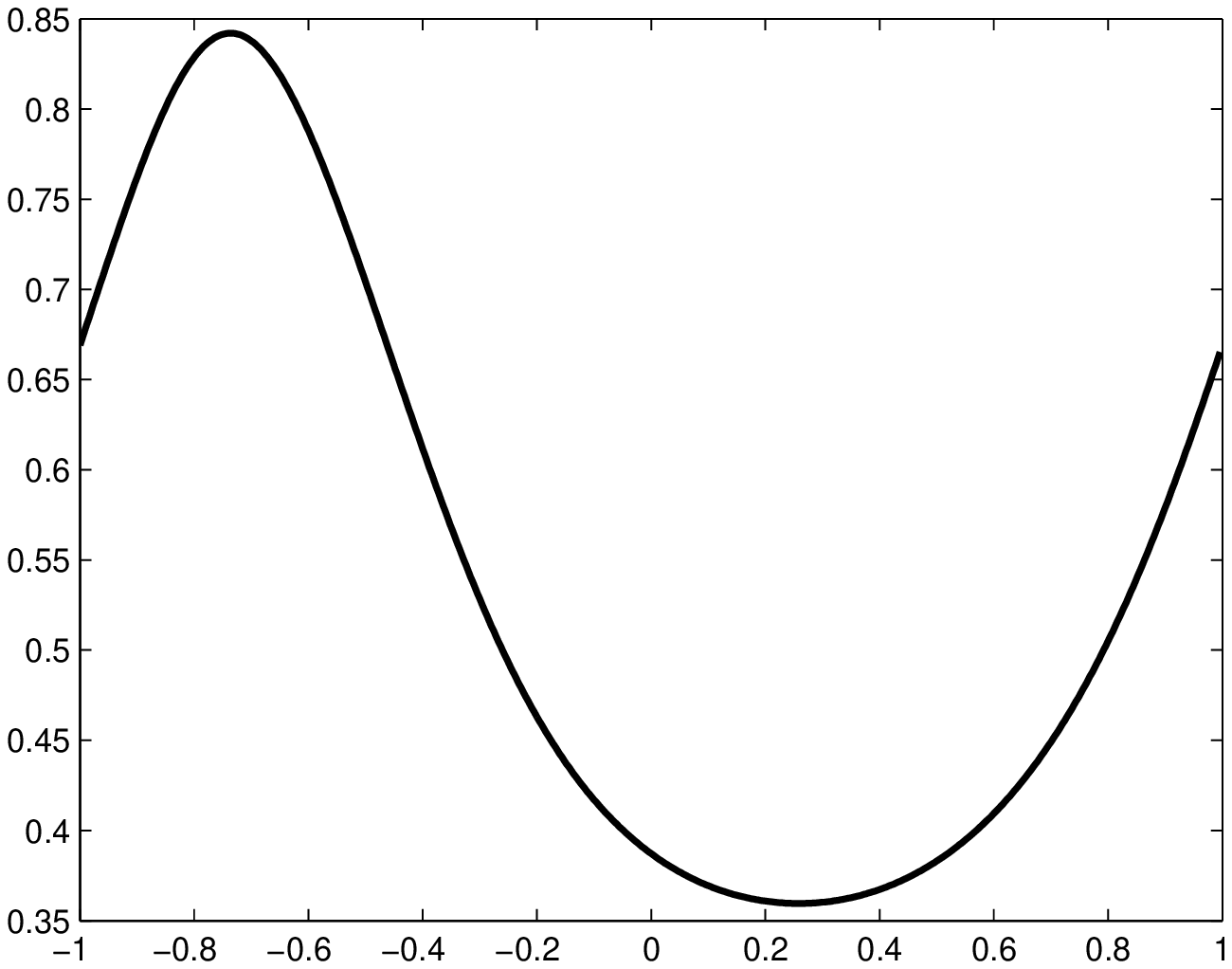}
  & \includegraphics[scale=.3,bb=111 247 509 544,clip]{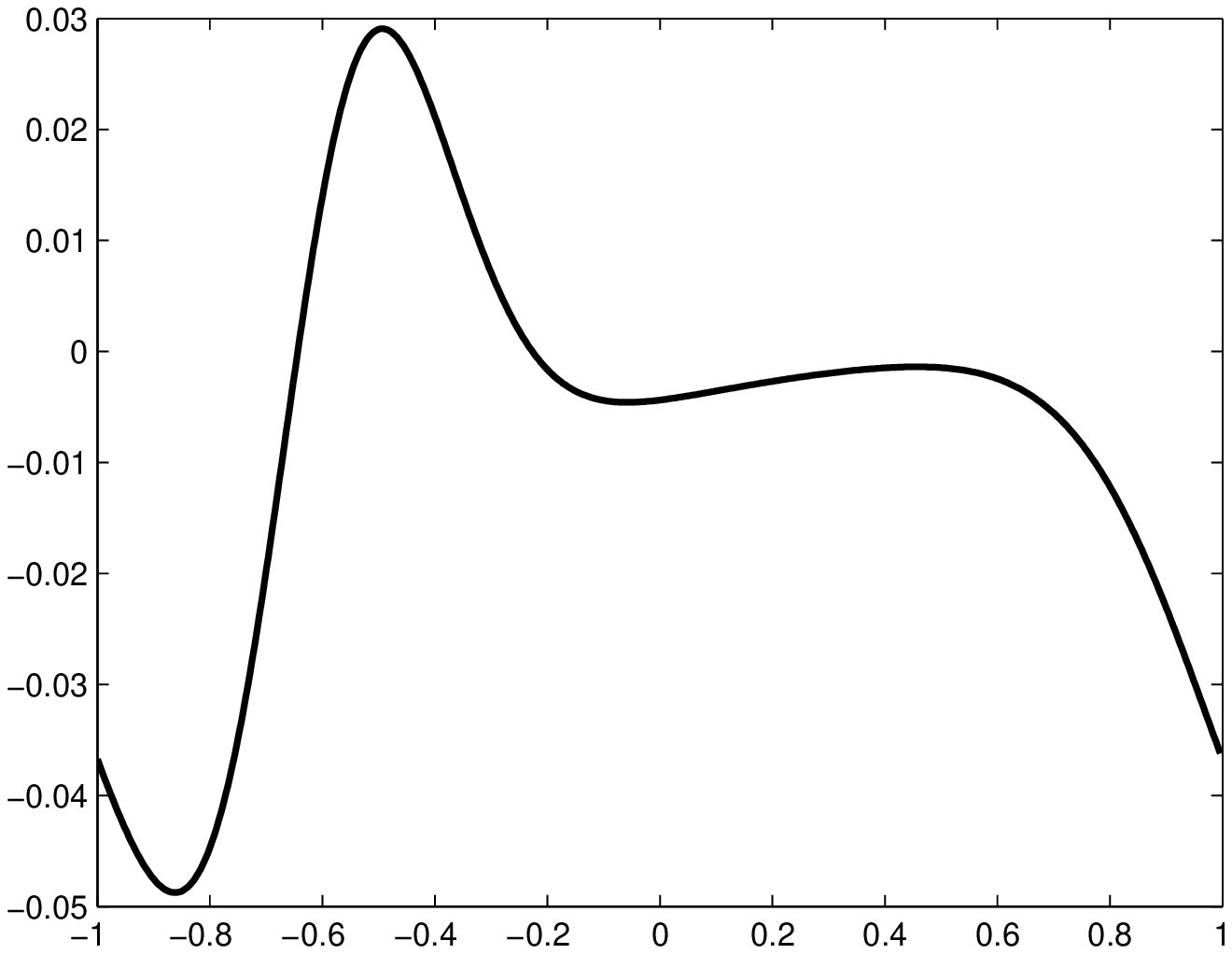} \\
\raisebox{43pt}{\small R84}
  & \includegraphics[scale=.3,bb=111 247 509 544,clip]{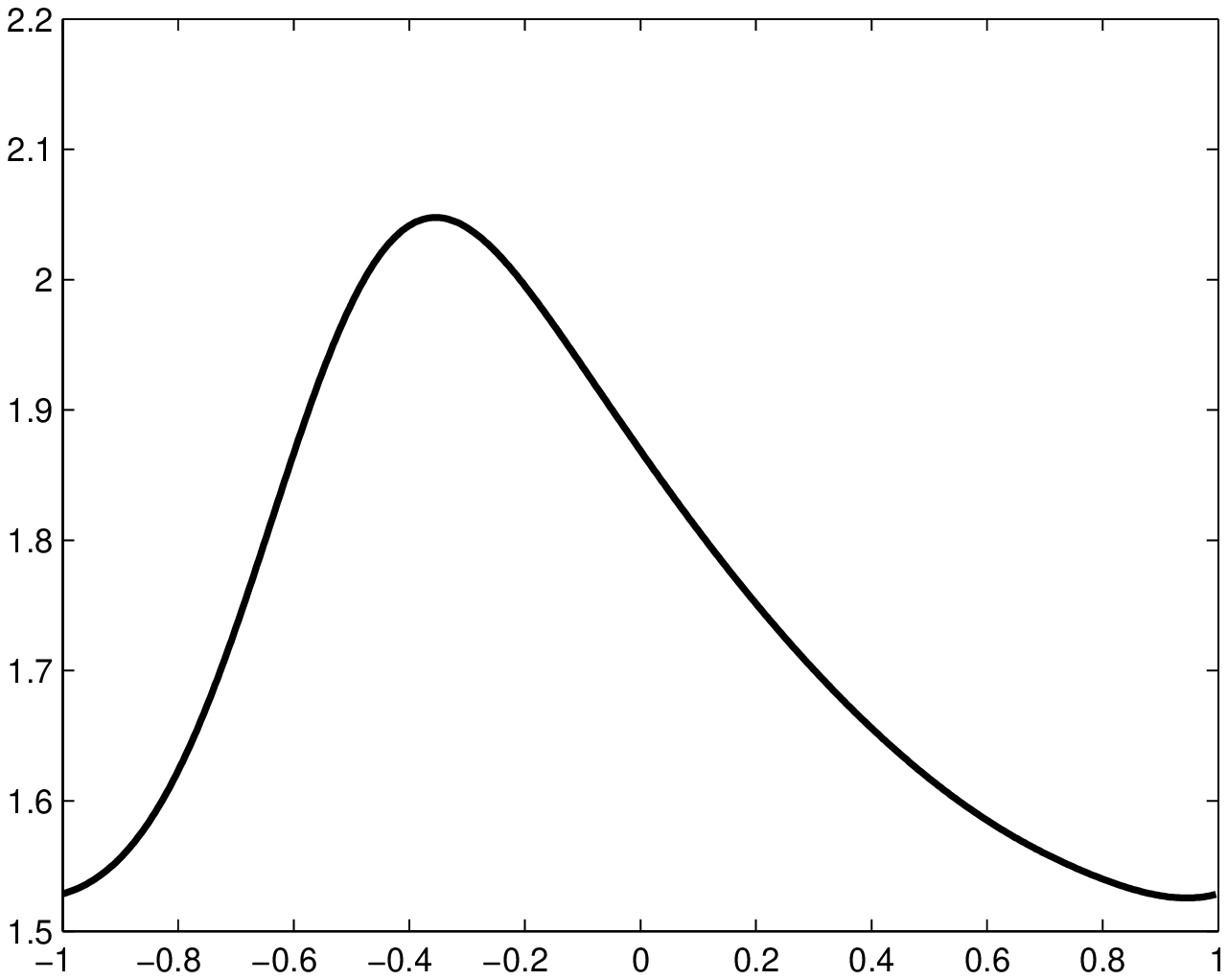}
  & \includegraphics[scale=.3,bb=111 247 509 544,clip]{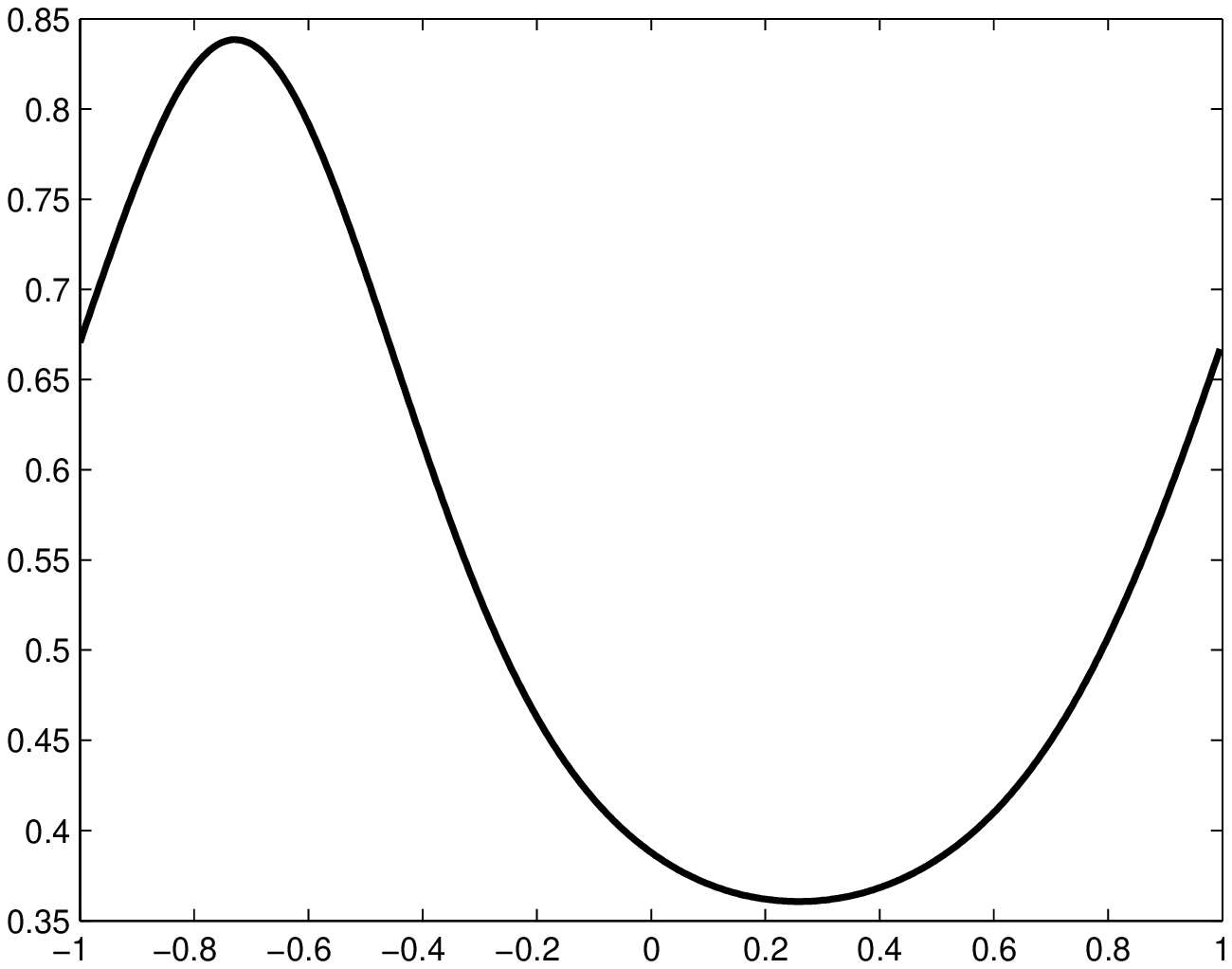}
  & \includegraphics[scale=.3,bb=111 247 509 544,clip]{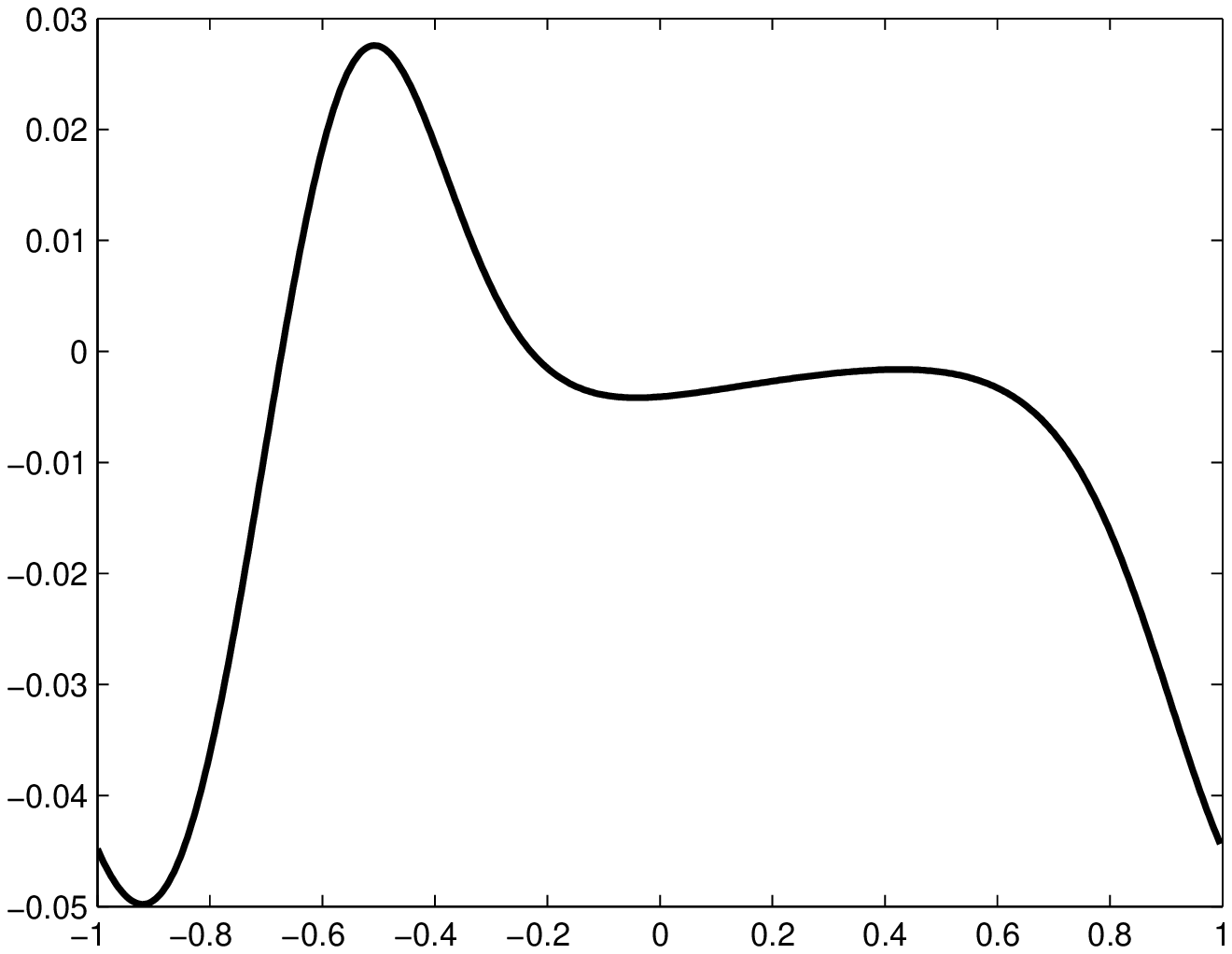}
\end{tabular}
\caption{Numerical solution of problem \eqref{eq:periodic2D}
at $t = 0.2$ and $y = 0$.}
\label{fig:Periodic2D_t=0.2}
\end{figure}

\begin{figure}[!ht]
\centering
\begin{tabular}{c@{}ccc}
& Density $\rho$ & Temperature $\theta$ & Vertical heat flux $q_3$ \\
\raisebox{43pt}{\small R20}
  & \includegraphics[scale=.3,bb=111 247 509 544,clip]{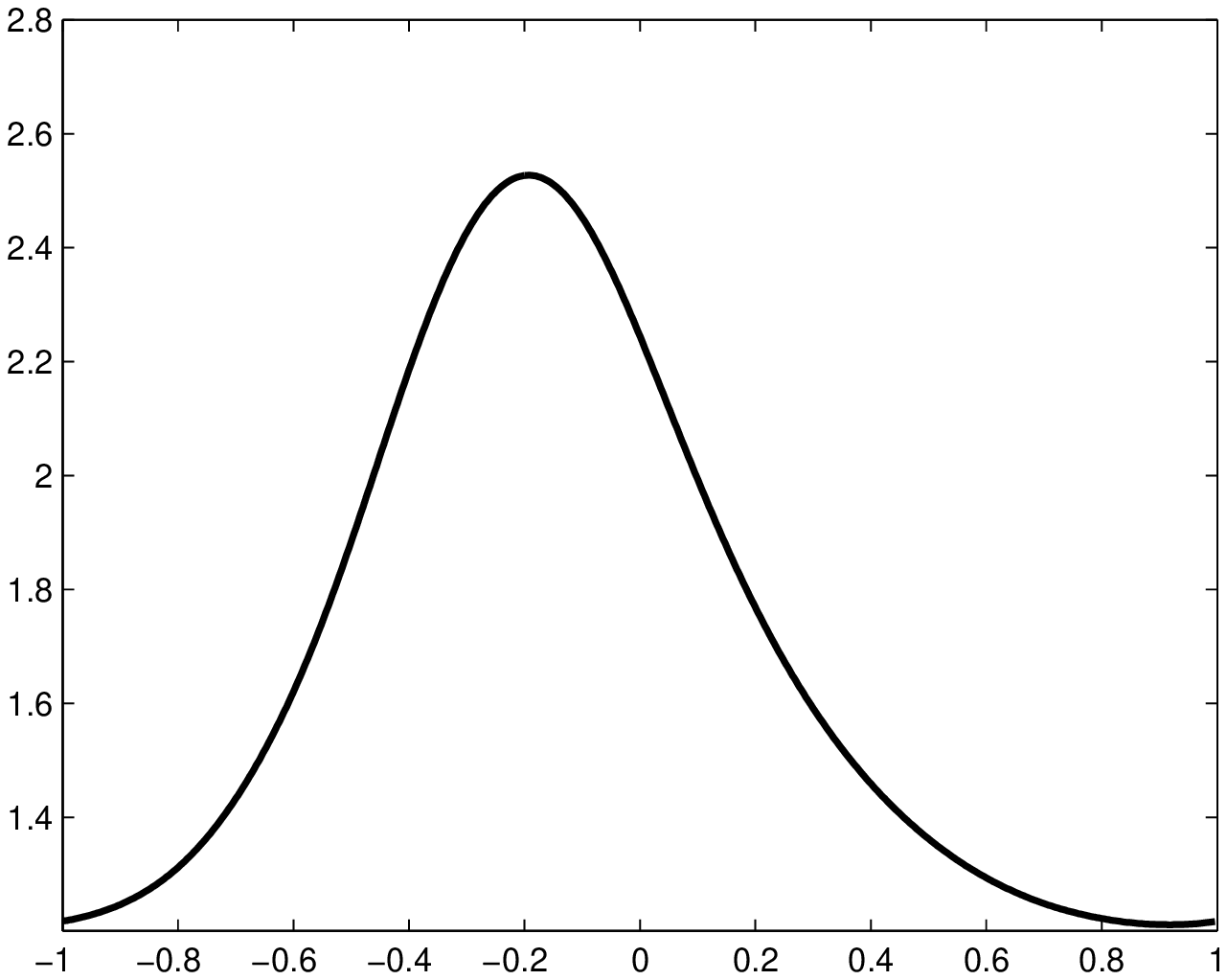}
  & \includegraphics[scale=.3,bb=111 247 509 544,clip]{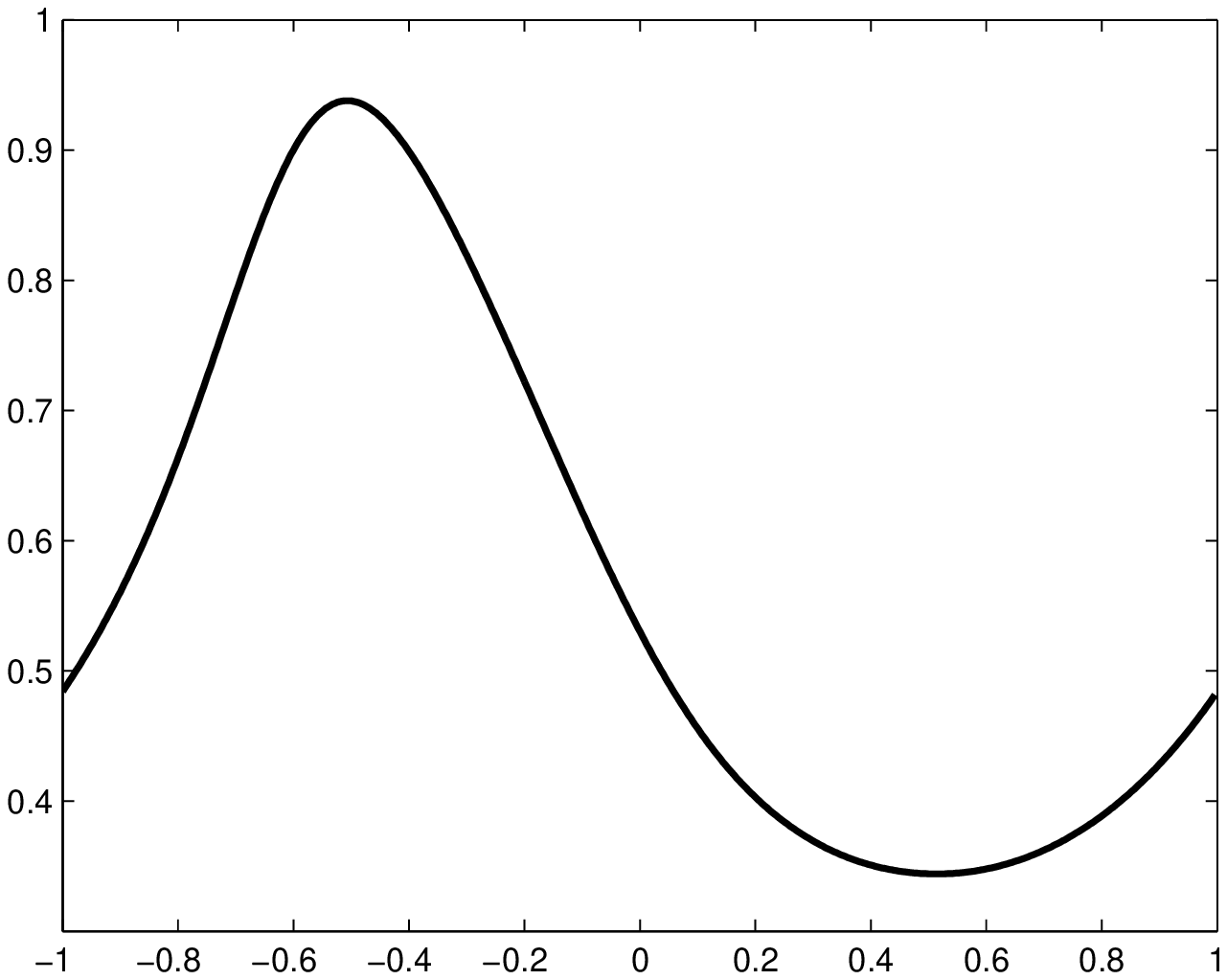}
  & \includegraphics[scale=.3,bb=111 247 509 544,clip]{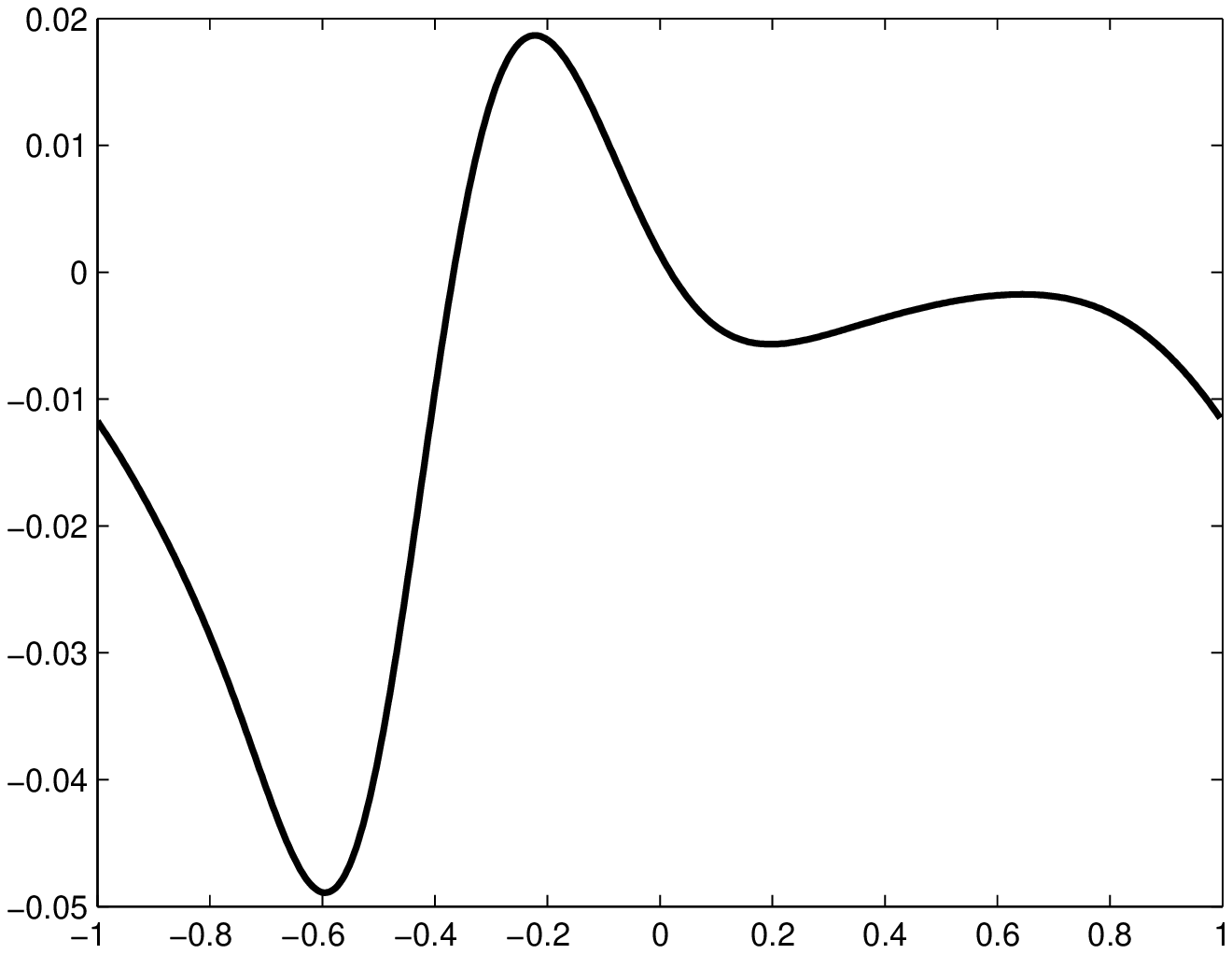} \\
\raisebox{43pt}{\small R84}
  & \includegraphics[scale=.3,bb=111 247 509 544,clip]{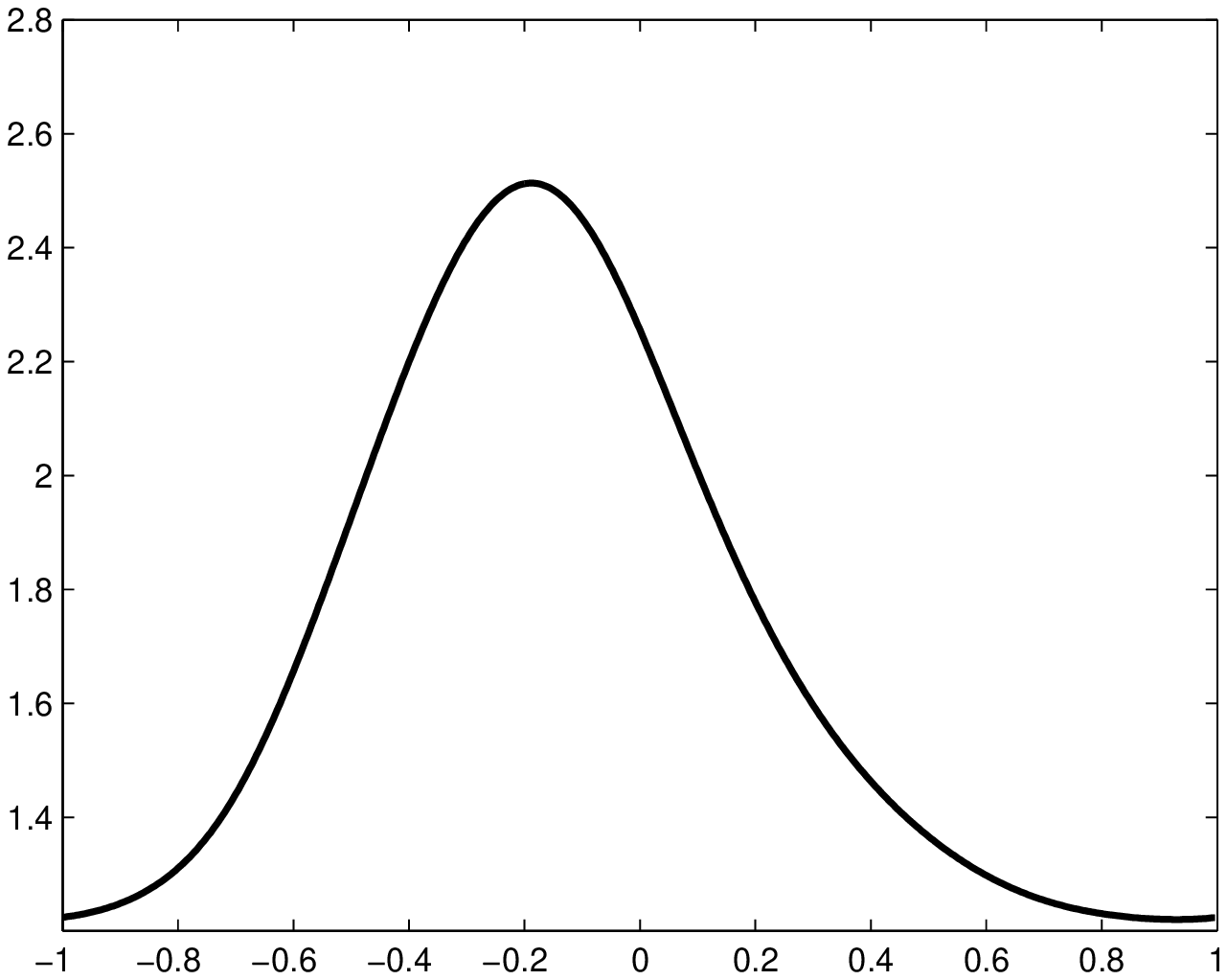}
  & \includegraphics[scale=.3,bb=111 247 509 544,clip]{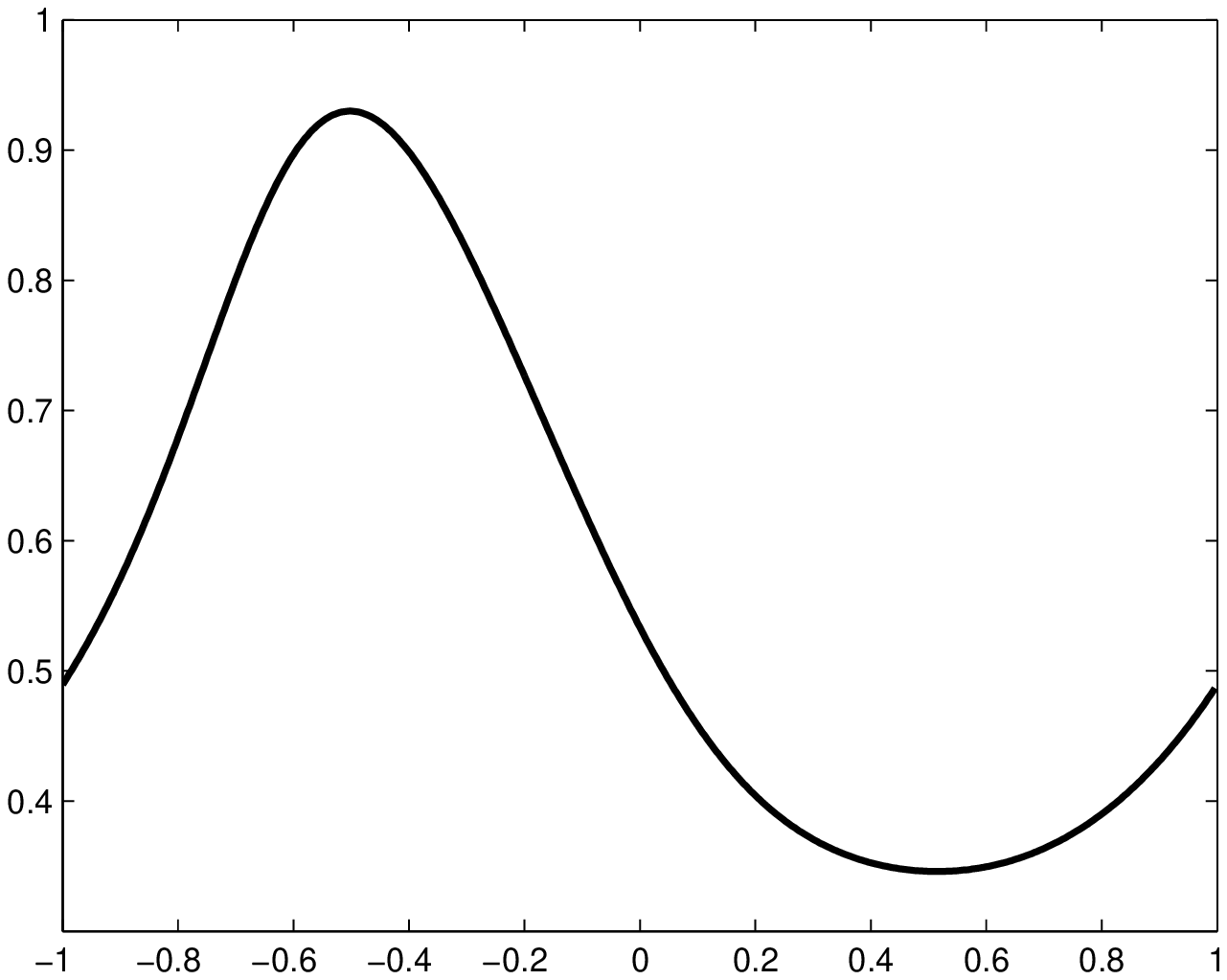}
  & \includegraphics[scale=.3,bb=111 247 509 544,clip]{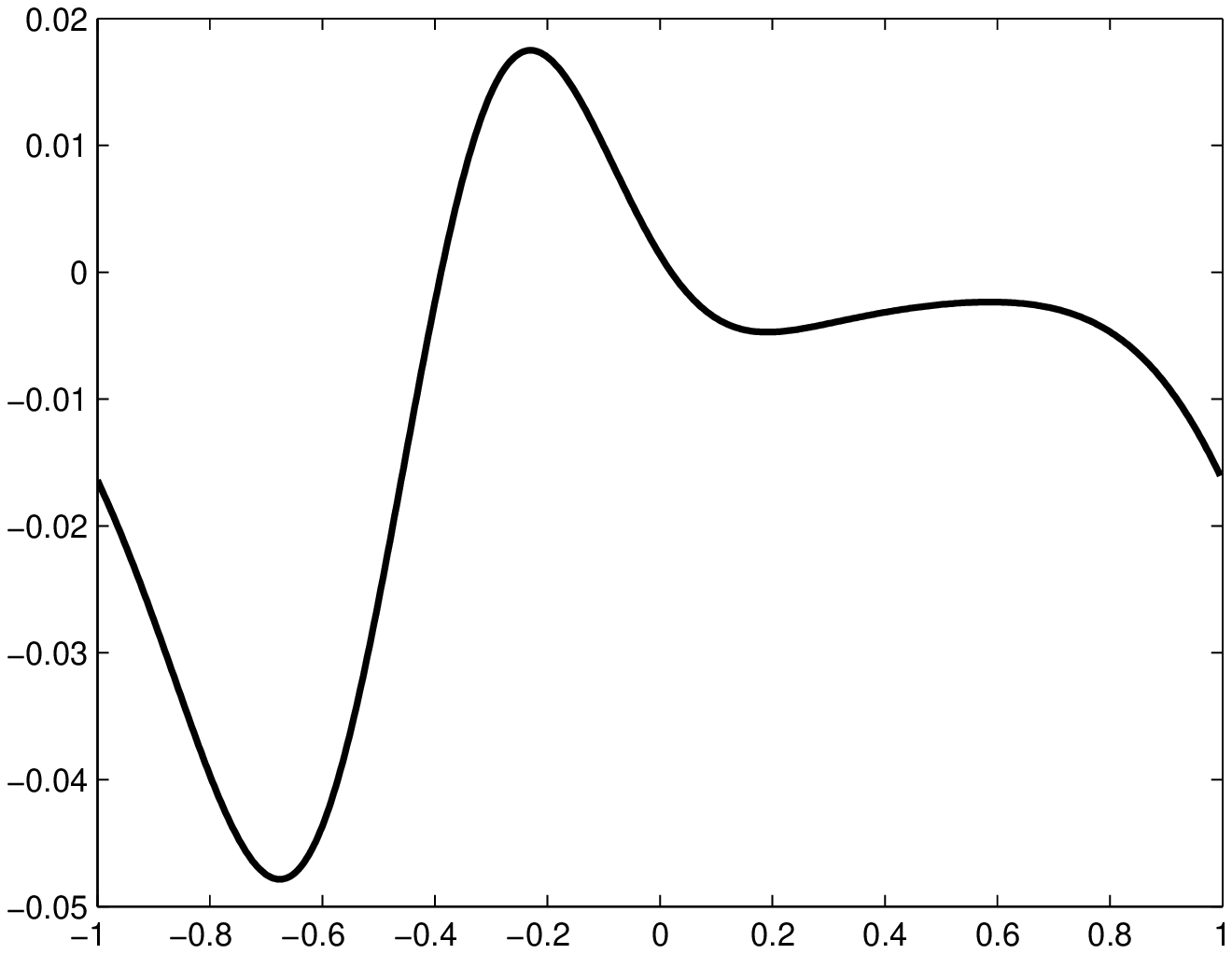}
\end{tabular}
\caption{Numerical solution of problem \eqref{eq:periodic2D}
at $t = 0.4$ and $y = 0$.}
\label{fig:Periodic2D_t=0.4}
\end{figure}

\begin{figure}[!ht]
\centering
\subfigure[$t=0.2$]{
\begin{overpic}[scale=.45]{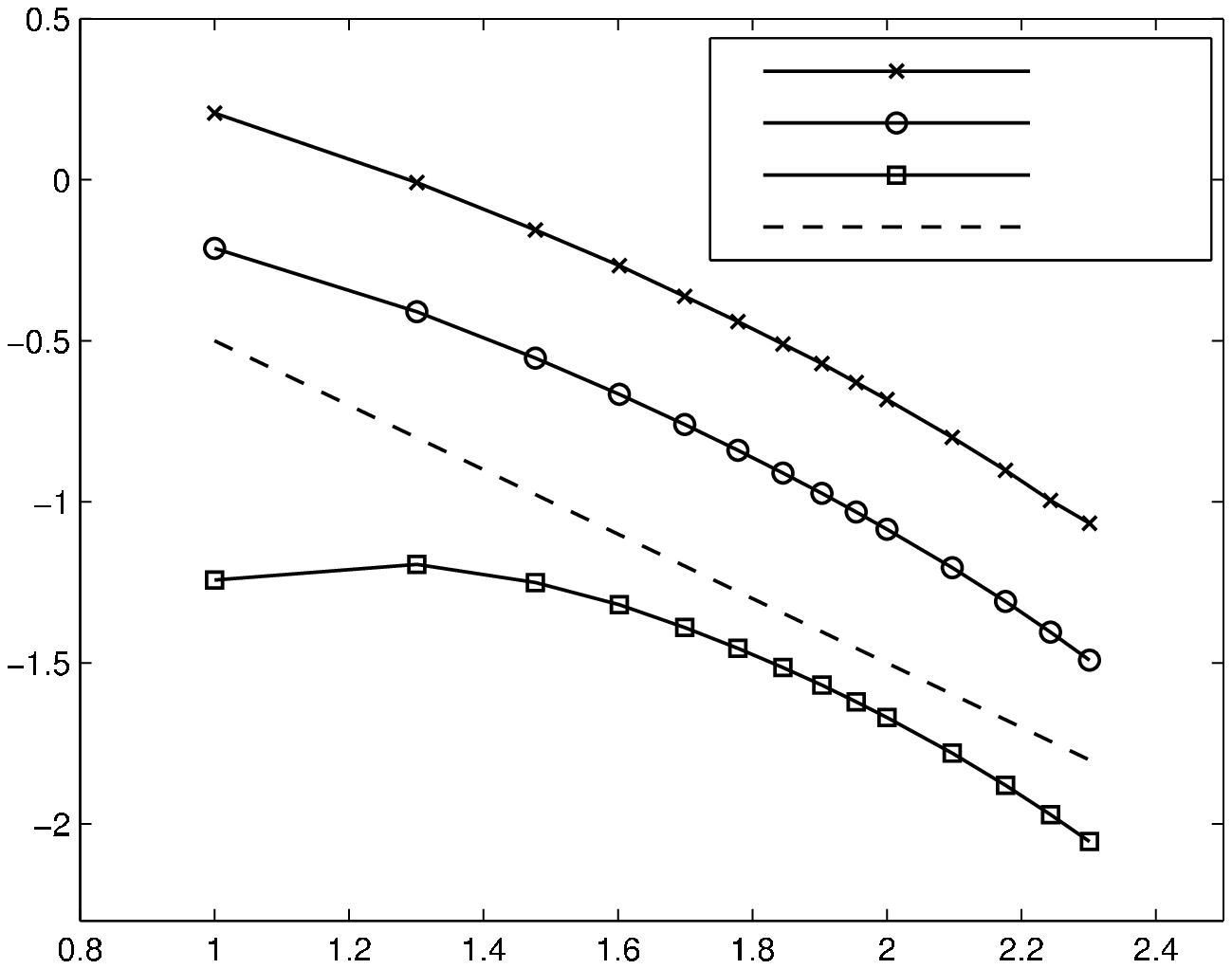}
\put(85,72){\scalebox{.5}{$E(\rho)$}}
\put(85,67.5){\scalebox{.5}{$E(\theta)$}}
\put(85,63.3){\scalebox{.5}{$E(q_3)$}}
\put(85,59.5){\scalebox{.5}{slope $-1$}}
\end{overpic}
}
\hspace{.5cm}
\subfigure[$t=0.4$]{
\begin{overpic}[scale=.45]{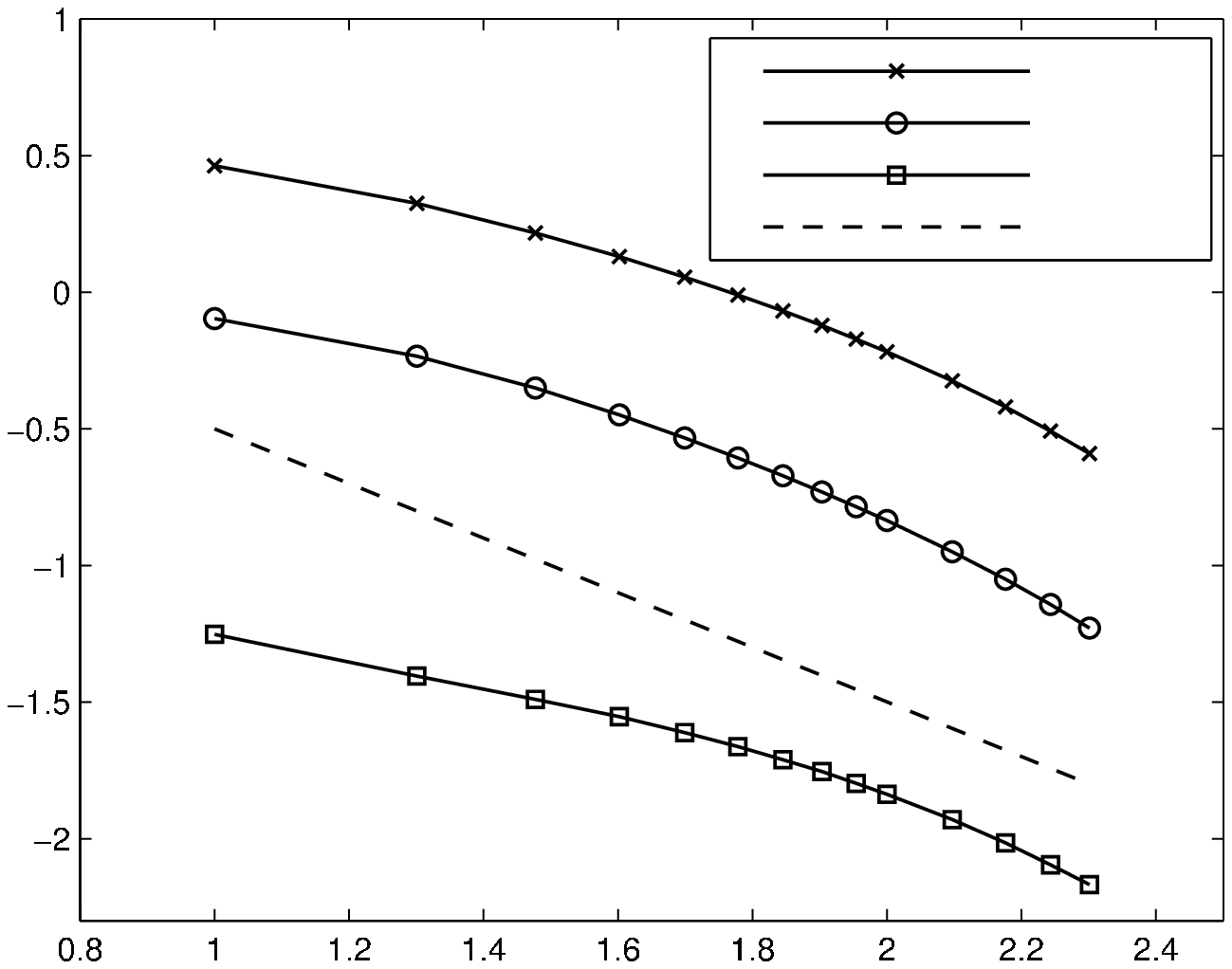}
\put(85,72){\scalebox{.5}{$E(\rho)$}}
\put(85,67.5){\scalebox{.5}{$E(\theta)$}}
\put(85,63.3){\scalebox{.5}{$E(q_3)$}}
\put(85,59.5){\scalebox{.5}{slope $-1$}}
\end{overpic}
}
\caption{The error plots for problem \eqref{eq:periodic2D}.
The $x$-axis is the logarithm of $N_x$, and the $y$-axis is
the logarithm of the norm of the $L^1$ error.}
\label{fig:Periodic2D_err}
\end{figure}


%% file: article_conclusion.tex
\section{Concluding remarks}
A uniform method to solve the regularized moment equations for
arbitrary order is proposed. This is the first time that the method
for arbitrary order regularized moment equations is developed, and the
moment method for large systems is applied to two-dimensional
problems. We are now devoting our efforts to the mesh adaptation and
parallelization of the algorithm to improve the computational
efficiency so that the proposed method can be applied to practical
applications.

\section*{Acknowledgements}
We thank Dr M. Torrilhon for providing us the software $ET_{XX}$ for
comparison and some useful discussions. The research of the second
author was supported in part by a Foundation for the Author of
National Excellent Doctoral Dissertation of PRC, the National Basic
Research Program of China under the grant 2005CB321701 and the
National Science Foundation of China under the grant 10731060.


%% file: article_appendix.tex
\section*{Appendix}

\appendix

\section{Collection of the mathematical symbols}
Since lots of mathematical symbols are used in this paper, in order to
provide convenience to the readers, we list some of them here.

\vspace{.5cm}

{ \renewcommand\arraystretch{1.3}
\begin{tabular}{lp{12cm}}
$Q(f,f)$                           & The Boltzmann collision operator \\
$Q_{\mathrm{BGK}}(f)$              & The BGK collision operator \\
$f_M$                              & The Maxwellian distribution \\
$\mathcal{H}_{\theta,\alpha}$      & The basis functions for Grad's expansion \\
$\He_n(x)$                         & The Hermite polynomials \\
$F_M(\bu, \theta)$                 & The finite dimensional space spanned by
                                     $\mathcal{H}_{\theta,\alpha}((\bxi - \bu) / \sqrt{\theta})$,
				     $|\alpha| \leqslant M$ \\
$f_{\beta}^n$                      & The discrete distribution function on the cell indexed
                                     by $\beta$ at time $t^n$ \\
$\bu_{\beta}^n, \theta_{\beta}^n$  & The mean velocity and temperature on the cell indexed by
                                     $\beta$ at time $t^n$ \\
$F_{\beta + \frac{1}{2} e_j}$      & The numerical flux between cells indexed by $\beta$ and
                                     $\beta + e_j$ \\
$\lambda_j^L, \lambda_j^R$         & The fastest signal velocities travelling in the direction
                                     of $-x_j$ and $x_j$ \\
$Q_h(\cdot)$                       & The discrete collision operator \\
$\mathcal{H}_{\beta,\alpha}^n$     & Equivalent to $\mathcal{H}_{\theta_{\beta}^n,\alpha}$ \\
$C_{\theta,\alpha}$                & See \eqref{eq:C} \\
$\Pi_{\bu,\theta} f$               & The function generated by projecting $f$ into
                                     $F_M(\bu, \theta)$ \\
$\Pi_{\beta}^n$                    & Abbreviation of $\Pi_{\bu_{\beta}^n, \theta_{\beta}^n}$ \\
$\Pi_f$                            & The projection operator from $F_{\infty}(\bu, \theta)$
                                     to $F_M(\bu, \theta)$, where $\bu$ and $\theta$ is the
                                     mean velocity and temperature of the distribution
                                     function $f$ \\
$\Pi_{f_1,f_2}$                    & The projection operator from $F_M(\bu_1, \theta_1)$ to
                                     $F_M(\bu_2, \theta_2)$, where $\bu_i$ and $\theta_i$ are
                                     the mean velocity and temperature of the distribution
                                     function $f_i$, $i = 1,2$ \\
$f^0$                              & A truncation of the distribution \eqref{eq:expansion},
                                     defined by \eqref{eq:truncation} \\
$f^k$                              & The $k$th order term in the Chapman-Enskog expansion \\
$(\cdot)_{\alpha}$                 & The coefficient indexed by $\alpha$ in the expansion of
                                     the parameter function \\
$\tilde{\Pi}_{\bu,\theta}$         & The projection operator to the space $F_{M+1}(\bu,\theta)$ \\
$\tilde{\Pi}_{\beta}^n$            & Abbreviation of $\tilde{\Pi}_{\bu_{\beta}^n, \theta_{\beta}^n}$ \\
\end{tabular}}

\section{Calculation of the partial derivative $\partial F / \partial \tau$}
\label{sec:df_dt}
The calculation of the temporal partial derivative of \eqref{eq:F}
is performed here. Define $A(v,w,\tau)$ and $B(v,w)$ as
\begin{equation}
A(v,w,\tau) = [(\htheta - 1)\tau + 1]v + w\tau, \quad
B(v,w) = \frac{\partial A}{\partial \tau} = (\htheta - 1)v + w.
\end{equation}
It follows from \eqref{eq:differential} that
\begin{equation} \label{eq:diff_He_raw}
\begin{split}
& \frac{\partial}{\partial \tau} \left[
  \He_m(A(v,w,\tau)) \exp\left(
    -\frac{[A(v,w,\tau)]^2}{2} \right) \right] \\
= & -B(v,w) \He_{m+1}(A(v,w,\tau)) \exp\left(
    -\frac{[A(v,w,\tau)]^2}{2} \right).
\end{split}
\end{equation}
With the definition of $R(\tau)$ and $S(\tau)$ \eqref{eq:RS}, $B(v,w)$
can be related with $A(v,w,\tau)$ by
\begin{equation} \label{eq:AB}
B(v,w) = R(\tau) A(v,w,\tau) + w S(\tau).
\end{equation}
Substituting \eqref{eq:AB} into \eqref{eq:diff_He_raw}, and employing
the recursion relation of Hermite polynomials, we have
\begin{equation} \label{eq:diff_He}
\begin{split}
& \frac{\partial}{\partial \tau}
  \left[\He_m(A) \exp(-A^2/2)\right] \\
= & -[R \He_{m+2}(A) + w S \He_{m+1}(A) + (m + 1) R \He_m(A)]
  \exp(-A^2/2),
\end{split}
\end{equation}
where all parameters of $A$, $S$ and $R$ are omitted.  With
\eqref{eq:diff_He}, the partial derivative of
$\mathcal{H}_{\theta_1,\alpha}(\bA(\bv, \bw, \tau))$ can be naturally
obtained:
\begin{equation}
\frac{\partial}{\partial \tau}
  \mathcal{H}_{\theta_1,\alpha}(\bA)
= -\sum_{d=1}^D \Big[ \theta_1 R \,
  \mathcal{H}_{\theta_1,\alpha+2e_d}(\bA) +
w_d \sqrt{\theta_1} S \,
  \mathcal{H}_{\theta_1, \alpha + e_d}(\bA)
  + (\alpha_d + 1) R \,
  \mathcal{H}_{\theta_1, \alpha}(\bA) \Big],
\end{equation}
where
\begin{equation}
\bA = \bA(\bu,\bw,\tau) = [(\hat{\theta}-1) \tau + 1]\bv + \tau\bw
= [ A(u_1, w_1, \tau), \cdots, A(u_D, w_D, \tau) ]^T.
\end{equation}
Since
\begin{equation}
F(\bv, \tau) = F_{\alpha}(\tau) S(\tau)^{-(|\alpha| + D)}
  \mathcal{H}_{\theta_1,\alpha} (\bA(\bv, \bw, \tau)),
\end{equation}
we finally get
\begin{equation}
\begin{split}
\frac{\partial}{\partial \tau} F(\bv, \tau) &=
  \sum_{\alpha \in \bbN^D} S^{-(|\alpha| + D)} \Bigg\{
    \mathcal{H}_{\theta_1,\alpha}
    \frac{\mathrm{d}}{\mathrm{d} \tau} F_{\alpha}
    - F_{\alpha} \cdot \\
& \qquad \Bigg[ \sum_{d=1}^D \left(
  \theta_1 R \mathcal{H}_{\theta_1, \alpha + 2 e_d}
  + w_d \sqrt{\theta_1} S \mathcal{H}_{\theta_1, \alpha + e_d}
  + (\alpha_d + 1) R \mathcal{H}_{\theta_1, \alpha}
\right) \\
& \qquad \phantom{\Bigg[} + (|\alpha| + D) (1 - \htheta)
  S \mathcal{H}_{\theta_1, \alpha} \Bigg] \Bigg\} \\
&= \sum_{\alpha \in \bbN^D} S^{-(|\alpha| + D)}
  \mathcal{H}_{\theta_1,\alpha} \left\{
    \frac{\mathrm{d}}{\mathrm{d}\tau} F_{\alpha}
    - \sum_{d=1}^D S^2 \left[
      \theta_1 R F_{\alpha-2e_d}
      + w_d \sqrt{\theta_1} F_{\alpha - e_d}
    \right]
  \right\},
\end{split}
\end{equation}
where the parameter of $\mathcal{H}_{\theta_1,\alpha}$,
i.e. $\bA(\bv,\bw,\tau)$, is also omitted, and for $\alpha$
with negative components, $F_{\alpha}$ is taken to be zero.